\documentclass{LMCS}

\def\dOi{11(2:6)2015}
\lmcsheading%
{\dOi}
{1--48}
{}
{}
{Jan.~28, 2014}
{Jun.~10, 2015}
{}

\ACMCCS{[{\bf Theory of computation}]: Semantics and
  reasoning---Program Semantics---Denotational Semantics} 
\subjclass{F.3.2}

\usepackage[english]{babel}
\usepackage[utf8]{inputenc}
\usepackage{amssymb,stmaryrd}
\usepackage[all]{xy}
\usepackage{url}
\usepackage{graphicx}
\usepackage{turnstile}
\usepackage{mathpartir}
\usepackage{todonotes}
\usepackage{hyperref}

\newcommand{\restrict}{\upharpoonright}
\newcommand{\from}{\leftarrow}
\newcommand{\tto}{\Rightarrow}

\newcommand{\intr}[1]{\llbracket #1 \rrbracket}
\newcommand{\lpair}[1]{\langle #1 \rangle}

\newcommand{\leg}[1]{\mathcal{L}_{#1}}

\newcommand{\inter}{\mathop{|\!|}}

\newcommand{\unfold}[1]{\widetilde{A}}

\newcommand{\vis}{\mathcal{V}}
\newcommand{\church}[1]{\underline{#1}}
\newcommand{\size}{\mathrm{rsize}}
\newcommand{\cosize}{\mathrm{rcosize}}
\newcommand{\dom}{\mathrm{dom}}

\newcommand{\lhr}{\mathrm{lhr}}

\newcommand{\bs}{\mathrm{bs}}
\newcommand{\embeds}{\hookrightarrow}
\newcommand{\norm}{\mathcal{N}}
\newcommand{\depth}{\mathrm{depth}}
\newcommand{\depthbs}{\mathrm{depth}}

\newcommand{\nodes}{\mathrm{nodes}}
\newcommand{\lh}{\mathrm{lh}}
\newcommand{\id}{\mathrm{id}}
\newcommand{\daimon}{*}
\newcommand{\fv}{\mathrm{fv}}
\newcommand{\h}{\mathrm{h}}
\newcommand{\lv}{\mathrm{lv}}

\newcommand{\loc}{\mathrm{L{ls}}}
\newcommand{\gen}{\mathrm{L{gen}}}
\newcommand{\laml}{{\lambda\mathrm{l}}}
\newcommand{\ord}{\mathrm{ord}}
\newcommand{\length}{\mathrm{l}}
\newcommand{\purea}{I_{\infty}}
\newcommand{\enb}{\mathrel{\vdash}}
\newcommand{\tbs}{\mathcal{T}}
\newcommand{\sz}{\mathrm{size}}
\newcommand{\df}{\mathrm{df}}
\newcommand{\ensquare}[1]{\begin{array}{|l|}\hline #1 \\ \hline \end{array}}

\newlength{\viewht}
\newlength{\viewlift}
\newlength{\viewdp}
\newlength{\viewdrop}

\newcommand{\pview}[1]{
\settoheight{\viewht}{\makebox{$#1$}}
\setlength{\viewlift}{\viewht}%
\addtolength{\viewlift}{-1ex}%
\raisebox{0.3\viewlift}{
  \makebox{$\ulcorner$}}
  \!#1\!
\settoheight{\viewht}{\makebox{$#1$}}
\setlength{\viewlift}{\viewht}%
\addtolength{\viewlift}{-1ex}%
\raisebox{0.3\viewlift}{
  \makebox{$\urcorner$}}
}

\newcommand{\oview}[1]{
\settodepth{\viewdp}{\makebox{$#1$}}
\setlength{\viewdrop}{0.3\viewdp}%
\addtolength{\viewdrop}{0.5ex}%
\raisebox{-\viewdrop}{
  \makebox{$\llcorner$}}
  \!#1\!
\settodepth{\viewdp}{\makebox{$#1$}}
\setlength{\viewdrop}{0.3\viewdp}%
\addtolength{\viewdrop}{0.5ex}%
\raisebox{-\viewdrop}{
  \makebox{$\lrcorner$}}
} 

\begin{document}
\title[Bounding linear head reduction
and visible interaction through skeletons]{Bounding linear head reduction
and visible interaction through skeletons\rsuper*}
\author[P.~Clairambault]{Pierre Clairambault}
\address{CNRS, ENS de Lyon, Inria, UCBL, Universit\'e de Lyon}
\email{pierre.clairambault@ens-lyon.fr}

\keywords{Linear head reduction, complexity, game semantics}

\titlecomment{{\lsuper*}This work presents the combined results of two conference papers \cite{fossacs11,tlca13}.}
\thanks{This work was partially supported by the ERC Advanced Grant ECSYM and
by the French ANR grant ELICA, ANR-14-CE25-0005.}

\maketitle

\begin{abstract}
In this paper, we study the complexity of execution in higher-order programming languages. 
Our study has two facets: on the one hand
we give an upper bound to the length of interactions between bounded P-visible strategies in Hyland-Ong game semantics. This
result covers models of programming languages with access to computational effects like non-determinism, state or control
operators, but its semantic formulation causes a loose connection to syntax.
On the other hand we give a syntactic counterpart of our semantic study: a non-elementary upper bound to 
the length of the linear head reduction sequence (a low-level notion of reduction, close to the actual implementation of
the reduction of higher-order programs by abstract machines) of simply-typed $\lambda$-terms. In both cases our upper
bounds are proved optimal by giving matching lower bounds.

These two results, although different in scope, are proved using the same method: we introduce a simple reduction on finite
trees of natural numbers, hereby called \emph{interaction skeletons}. We study this reduction and give upper bounds to its complexity.
We then apply this study by giving two simulation results: a semantic one measuring progress in game-theoretic interaction
via interaction skeletons, and a syntactic one establishing a correspondence between linear head reduction of terms satisfying a
locality condition called \emph{local scope} and the reduction of interaction skeletons. This result is then generalized to
arbitrary terms by a local scopization transformation.
\end{abstract}

\section{Introduction}

In the last two decades there has been a significant interest in the study of \emph{quantitative} or \emph{intensional} aspects of
higher-order programs; in particular, the study of their \emph{complexity} has generated a lot of effort. In the context of
the $\lambda$-calculus, the first result that comes to mind is the work by Schwichtenberg \cite{schw91}, later improved by
Beckmann \cite{beck}, establishing upper bounds to the length of $\beta$-reduction sequences for simply-typed $\lambda$-terms.
In the related line of work of \emph{implicit complexity}, type systems have been developed to characterize extensionally certain
classes of functions, such as polynomial \cite{lll} or elementary \cite{ell} time. Such systems rely on a soundness theorem
establishing that well-typed terms normalize in a certain restricted time, which is itself established using syntactic methods that
are specific to the system being studied. This calls for the development of syntax-independent tools to study precisely the execution
time of higher-order programs. The present paper is a contribution towards that goal.

In order to develop tools for estimating the complexity of programs that are at the same time syntax-independent but still
precise and computationally relevant, game semantics is a good place to start. Indeed Hyland-Ong game semantics, originally
introduced to solve the full abstraction problem for PCF \cite{hogames}, have since proved powerful and flexible enough to provide
fully abstract models for various computational features of programming languages: non-determinism, state, exceptions, control
operators, concurrency\dots~At the same time game semantics are computationally informative: moves in an interaction exactly
correspond to computation steps as performed by some abstract machines, as was formally proved by Danos, Herbelin and Regnier \cite{dhr}
for the simply-typed $\lambda$-calculus. Accordingly, the first contribution of this paper is an upper bound to the
length of interactions between bounded P-visible strategies in Hyland-Ong game semantics. Bounded P-visible strategies allow
computational effects such as non-determinism, control operators or ground state but disallow, for instance, higher-order references
using which a fixpoint combinator can be defined. 
In this context, we give a non-elementary upper bound to the length of interactions (so implicitly, computation), for which
we prove optimality by providing a matching lower bound. This is done by extracting from interaction sequences a finite tree
of natural numbers called an \emph{interaction skeleton}. We then prove a simulation result, showing that progress in the game theoretic
interaction is witnessed by a simple reduction on the corresponding interaction skeletons. This reduction is then analyzed independently
of its connection with game semantics, and this analysis informs a bound on the length of P-visible interactions.

This syntax-independent approach advocated above yields a result with a broader scope, however there is a price to pay:
the innocent strategies of game semantics are representations of $\eta$-long $\beta$-normal $\lambda$-terms, so their interactions
are only easily linked to syntactic reduction of terms of the form $M~N_1~\dots~N_p$ where $M$ and all the $N_i$s are
$\beta$-normal $\eta$-long normal forms; we will call such terms \emph{game situations} in this paper. In order to compensate for
this structural restriction, we investigate the direct relationship between interaction skeletons and linear head reduction,
independently of game semantics. We define a set of terms called \emph{generalized game situations}, that are both $\eta$-long in a
sense adapted to linear head reduction, and satisfy a key locality constraint called \emph{local scope}. Generalized game situations
are preserved by linear head reduction. Moreover they can be simulated by interaction skeletons, from which we deduce
optimal upper bounds to their complexity. These bounds are then generalized to arbitrary simply-typed $\lambda$-terms by
means of a syntactic $\lambda$-lifting \cite{llift} transformation putting arbitrary terms to locally scoped form.

\emph{Related works.} There are multiple approaches to the complexity analysis of higher-order programs, but they seem to separate
into two major families. On the one hand, Beckmann \cite{beck}, extending earlier work by Schwichtenberg \cite{schw91}, gave exact
bounds to the maximal length of $\beta$-reduction on simply-typed $\lambda$-terms. His analysis uses very basic information on the
terms (their length, or height, and order), but gives bounds that are in general very rough. On the other hand other groups,
including Dal Lago and Laurent \cite{DBLP:conf/csl/LagoL08}, De Carvalho \cite{DBLP:journals/tcs/CarvalhoPF11}, or Bernardet and
Lengrand \cite{DBLP:conf/fossacs/BernadetL11}, use semantic structures (respectively, game semantics, relational semantics, or
non-idempotent intersection types) to capture abstractly the precise complexity of particular terms. Their bounds are much more
precise on particular terms, but require information on the terms whose extraction is in general as long to obtain as actual
execution. The present contribution belongs to the first family. However, unlike Beckmann and Schwichtenberg our core tools are
syntax-independent. Moreover we focus on \emph{linear head reduction}, the notion of reduction implemented by several call-by-name
abstract machines \cite{pam}, closer to the actual execution of functional programming languages.

\emph{Outline.} In Section \ref{sec:prelim}, we introduce a few basic notions or notations useful to the rest of the paper. In
Section \ref{sec:games} we present the game semantics framework from which interaction skeletons were originally extracted, and 
prove our semantic simulation result. Section \ref{sec:bs} is a largely standalone section in which we study interaction
skeletons and their reduction, and prove our main complexity result. Finally, Section \ref{sec:lhr} focuses on the syntactic
implications of our study of skeletons and details their relationship with linear head reduction.

This paper is organized around the notion of interaction skeleton and their reduction, with two largely independent
applications to game semantics and to the complexity of linear head reduction. We chose to present first the game-theoretic
development, since it motivates the definition of interaction skeletons. However, our intention is that
the paper should be accessible to semantically-minded as well as more syntactically-minded readers. In particular, readers
not interested in game semantics should be able to skip Section \ref{sec:games} and still have everything needed to
understand Sections \ref{sec:bs} and \ref{sec:lhr}.

\section{Preliminaries}
\label{sec:prelim}

\subsection{Syntax and dynamics of the $\lambda$-calculus}

In this paper, we consider the simply-typed $\lambda$-calculus $\Lambda$ built from one base type $o$. Its \textbf{types} 
and \textbf{terms} are:
\begin{eqnarray*}
A, B &::=& o \mid A \to B\\
M, N &::=& \lambda x^A.~M \mid M~N \mid x \mid \daimon_A
\end{eqnarray*}
subject to the usual typing rules defining the typing relation $\Gamma \vdash M:A$ (for definiteness, contexts $\Gamma$
are considered to be sets of pairs $x:A$, where $x$ is a variable name and $A$ is a simple type). All the terms in this paper
are considered well-typed, although we will not always make it explicit.
 Note that we work here with the simply-typed $\lambda$-calculus \emph{à la Church}, 
\emph{i.e.} the variables are explicitly annotated with types (although we often omit the annotations for the sake of readability).
For each type $A$, there is a constant $\daimon_A : A$ of type $A$. We will often omit the index and write $\daimon$.
As usual, we write $\fv(M)$ for the set of \emph{free variables} of a term $M$, \emph{i.e.} variables appearing in the
term but not bound by a $\lambda$. If $A$ is a type, we write $\id_A$ for the term $\vdash \lambda x^A.~x: A \to A$.
Terms are assumed to obey Barendregt's convention, and are considered up to $\alpha$-equivalence.
Note that our design choices -- only one atom, each type is inhabited -- merely make the presentation
simpler and are not strictly required for our results to hold.

If $\Gamma, x:A \vdash M: B$ and $\Gamma \vdash N:A$, we write $M[N/x]$ for the \textbf{substitution} of $x$ by $N$ in $M$, 
\emph{i.e.} $M$ where all occurrences of $x$ have been replaced by $N$.
Although this paper focuses on linear head reduction (to be defined later), we will occasionally need \textbf{$\beta$-reduction}.
It is the usual notion of reduction in the $\lambda$-calculus, defined by the context-closure of:
\[
(\lambda x.~M)~N \to_\beta M[N/x]
\]
Likewise we consider \textbf{$\eta$-expansion} to be the context closure of:
\[
M \to_\eta \lambda x^A.~M~x
\]
valid whenever $M$ has type $A\to B$ for some types $A, B$ and $x\not \in \fv(M)$.

The \textbf{level} of a type is defined by $\lv(o) = 0$ and $\lv(A\to B) = \max(\lv(A) + 1, \lv(B))$.
Likewise, the \textbf{level} $\lv(M)$ of a term $M$ is the level of its type. Finally, the \textbf{order}
$\ord(M)$ of a term $M$ is the maximal $\lv(N)$, for all subterms $N$ of $M$. Within a term $\Gamma \vdash M:A$ such that
$(x:B) \in \Gamma$, we write $\lv_M(x) = \lv(B)$. The term $M$ will generally be implicit from the context, so we will just
write $\lv(x)$.

\subsection{Growth rates of functions}
We recall some standard notations for comparing growth rates of functions. For functions
$f, g: \mathbb{N} \to \mathbb{N}$,
we write $f(n) = \Theta(g(n))$ when there exists reals $c_1, c_2 > 0$ and $N \in \mathbb{N}$ such that for all $n\geq N$, $c_1 g(n) \leq f(n)
\leq c_2 g(n)$. This is generalized to functions of multiple variables $f, g: \mathbb{N}^p \to \mathbb{N}$ by setting that $f(n_1, \dots, n_p) =
\Theta(g(n_1, \dots, n_p))$ iff there are $c_1, c_2 > 0$ and $N \in \mathbb{N}$ such that for all $n_i \geq N$
we have $c_1 g(n_1, \dots, n_p) \leq f(n_1, \dots, n_p) \leq c_2 g(n_1, \dots, n_p)$. If $h: \mathbb{N} \to \mathbb{N}$ is another function,
we write $f(n_1, \dots, n_p) = h(\Theta(g(n_1, \dots, n_p)))$ iff there is a function $\phi : \mathbb{N}^p \to \mathbb{N}$ such that
$f(n_1, \dots, n_p) = h(\phi(n_1, \dots, n_p))$ and $\phi(n_1, \dots, n_p) = \Theta(g(n_1, \dots, n_p))$.

\section{Visible pointer structures and interaction skeletons}
\label{sec:games}

As we said before, the general purpose of this paper is to develop syntax-independent tools to reason about termination and
complexity of programming languages. \emph{Game semantics} provide such a framework: in this setting, programs
are identified with the set of their possible interactions with the execution environment, presented as a \emph{strategy}. All
the syntactic information is forgotten but at the same time no dynamic information is lost, since the strategy exactly describes
the behaviour of the term within any possible evaluation context.

More importantly, game semantics can be seen both as a \emph{denotational semantics} and
as an \emph{operational semantics} in the sense that the interaction process at the heart of games is operationally informative and
strongly related to the actual evaluation process as implemented, for instance, by abstract machines. This important intuition was
made formal for the first time, to the author's knowledge, by Danos, Herbelin and Regnier in \cite{dhr}. There, they showed that
assuming one has a simply-typed $\lambda$-term $M~N_1~\dots~N_p$ where $M$ and the $N_i$s are both $\eta$-long and $\beta$-normal,
then there is a step-by-step correspondence between:
\begin{itemize}
\item The \emph{linear head reduction sequence} of $M~N_1~\dots~N_p$,
\item The \emph{game-theoretic interaction} between strategies $\intr{M}$ and $\Pi_{1\leq i \leq p} \intr{N_i}$.
\end{itemize}
In this paper, we will refer to simply-typed $\lambda$-terms of that particular shape as \textbf{game situations}, because 
it is in those situations that the connection between game semantics and execution of programs is the most direct. Indeed
(innocent) game semantics can be seen as a reduction-free way of composing Böhm trees, \emph{i.e.} of computing game situations.

Game situations also provide the starting point of our present contributions. Indeed, the connection above reduces the termination and complexity
analysis of the execution of a game situation $M~N_1~\dots~N_p$ to the syntax-independent analysis of the game-theoretic witness
of this execution, \emph{i.e.} the game-theoretic interaction between the strategies $\intr{M}$ and $\Pi_{1\leq i \leq n} \intr{N_i}$.
More precisely, it turns out that from this interaction one just has to keep the structure of pointers in order to get a precise
estimate of the complexity of execution.

In this section, we will start by recalling a few basic definitions of Hyland-Ong game semantics. We will show that the mechanism
of composition gives rise to structure called \emph{visible pointer structures}. We will then show the main result of this section,
that visible pointer structures can be simulated by a simple rewriting system called \emph{interaction skeletons}.

\subsection{Brief reminder of Hyland-Ong games}

We start this section by recalling some of the basic definitions of Hyland-Ong games. The presentation
of game semantics will be intentionally brief and informal, firstly because it is only there to provide context and is
not a prerequisite to understand the paper, and secondly because good introductions can be easily found in the literature, 
see \emph{e.g.} \cite{harmer2004innocent}.

We are interested in games with two participants: Opponent (O, the \emph{environment}) and Player (P, the \emph{program}).
They play on directed graphs called \emph{arenas}, which are semantic versions of \emph{types}.
Formally, an \textbf{arena} is a structure $A=(M_A,\lambda_A,I_A,\vdash_{A})$ where:
\begin{itemize}
\item $M_{A}$ is a set of \textbf{moves},
\item $\lambda_{A}:M_{A} \to \{O,P\}$ is a polarity function indicating whether
a move is an Opponent or Player move ($O$-move or $P$-move).
\item $I_A\subseteq \lambda_A^{-1}(\{O\})$ is a set of \textbf{initial moves}.
\item $\vdash_{A} \subset M_{A}\times M_A$ is a relation called \textbf{enabling}, such that
if $m \vdash_A n$, $\lambda_{A}(m)\neq \lambda_{A}(n)$.
\end{itemize}

In our complexity analysis we will use notions of \emph{depth}.
A move $m \in M_A$ has \textbf{depth} $d\in \mathbb{N}$ if there is an enabling sequence:
\[
m_0 \vdash_A m_1 \vdash \dots \vdash m_{n-1} = m
\]
where $m_0 \in I_A$, and no shortest enabling sequence exists; for instance any initial move
has depth $0$. Likewise, an arena $A$ has depth $d$ if $d$ is the largest depth
of any move in $M_A$.

We now define plays as \textbf{justified sequences} over ${A}$: these are
sequences $s$ of moves of ${A}$, each non-initial move $m$ in $s$ being equipped with a pointer to an earlier move
$n$ in $s$, satisfying $n\vdash_{A} m$. In other words, a justified sequence $s$ over ${A}$ is such that
each reversed pointer chain $s_{i_0}\from s_{i_1} \from \dots \from s_{i_n}$ is a path on ${A}$ (viewed as a graph).
The role of pointers is to allow \emph{reopenings} or \emph{backtracking} in plays. When writing justified sequences, we will often omit the
justification information if this does not cause any ambiguity. The symbol $\sqsubseteq$ will denote the prefix ordering on 
justified sequences, and $s_1 \sqsubseteq^P s_2$ (resp. $s_1 \sqsubseteq^O s_2$ will mean that $s_1$ is a $P$-ending 
(resp. $O$-ending) prefix of $s_2$. If $s$ is a justified sequence on ${A}$, $|s|$ will denote its length. If $s$ is a justified
sequence over $A$ and $\Sigma \subseteq M_A$, then the \textbf{restriction} $s\restrict \Sigma$ comprises the moves of $s$ in
$\Sigma$. Pointers in $s\restrict \Sigma$ are those obtained by following a path of pointers in $s$:
\[
s_{j_1} \to s_{j_2} \to \dots \to s_{j_n}
\]
where $s_{j_1}, s_{j_n} \in \Sigma$ but for $2\leq k \leq n-1$ $s_{j_k} \not \in \Sigma$. If $s = s_0\dots s_n$ is
a justified sequence and $i\leq n$, write $s_{\leq i}$ for its subsequence $s_0\dots s_i$.
The \textbf{legal plays} over ${A}$ are the justified sequences $s$ on ${A}$ satisfying the
\textbf{alternation} condition, \emph{i.e.} that if $tmn \sqsubseteq s$, then $\lambda_{A}(m)\neq \lambda_A(n)$. The
set of legal plays on $A$ is denoted by $\leg{A}$.

Given a justified sequence $s$ on ${A}$, it has two subsequences of particular interest: the P-view and O-view.
The view for P (resp. O) may be understood as the subsequence of the play where P (resp. O) only sees his own duplications.
Practically, the P-view $\pview{s}$ of $s$ is computed recursively by forgetting everything
under Opponent's pointers, as follows:
\begin{itemize}
\item $\pview{sm} = \pview{s}m$ if $\lambda_{A}(m)=P$;
\item $\pview{sm} = m$ if $m\in I_A$ and $m$ has no justification pointer;
\item $\pview{s_1m s_2n} = \pview{s_1}mn$ if $\lambda_{A}(n)=O$ and $n$ points to $m$.
\end{itemize}
The O-view $\oview{s}$ of $s$ is defined dually, without the special treatment of initial moves.

In this subsection, we will present several classes of strategies on arena games that are of interest to us in the present paper.
A \textbf{strategy} $\sigma$ on ${A}$ is a set of even-length legal plays on ${A}$, closed under even-length prefix. A strategy from $A$ to $B$
is a strategy $\sigma: A\tto B$, where $A\tto B$ is the usual arrow arena defined by 

\begin{eqnarray*}
M_{A\tto B} &=& M_A + M_B\\
\lambda_{A\tto B} &=& [\overline{\lambda_A}, \lambda_B]\\
I_{A\tto B} &=& I_B \\
\vdash_{A\tto B} &=& \vdash_A + \vdash_B + I_B\times I_A
\end{eqnarray*}
where $\overline{\lambda_A}$ means $\lambda_A$ with polarity $O/P$ reversed.

\subsubsection{Composition.} We define composition of strategies by the usual parallel interaction plus hiding mechanism.
If ${A}$, ${B}$ and ${C}$ are arenas, we define the set of \textbf{interactions}
$I({A},{B},{C})$ as the set of justified sequences $u$ over ${A}$, ${B}$
and ${C}$ such that $u{\restrict{{A},{B}}}\in \mathcal{L}_{{A}\tto {B}}$,
$u{\restrict{{B},{C}}}\in \mathcal{L}_{{B}\tto {C}}$ and
$u{\restrict{{A},{C}}}\in \mathcal{L}_{{A}\tto{C}}$. Then, if $\sigma:{A}\tto {B}$
and $\tau:{B}\tto {C}$, their parallel interaction is
$
\sigma \inter \tau = \{u\in I({A},{B},{C}) \mid   u{\restrict A, B}\in \sigma \wedge u{\restrict B, C} \in \tau\}
$.
Their composition is $\sigma;\tau = \{u{\restrict A, C} \mid  u\in \sigma||\tau\}$, 
is associative and admits copycat strategies as identities.

\subsubsection{$P$-visible strategies.} A strategy $\sigma$ is \textbf{$P$-visible} if each of its moves points to the current $P$-view. Formally, for all $sab \in \sigma$,
$b$ points inside $\pview{sa}$. $P$-visible strategies are stable under composition, and 
correspond to functional programs with ground type references \cite{DBLP:journals/entcs/AbramskyM96}.

\subsubsection{Innocent strategies.}
The class of \emph{innocent} strategies is central in game semantics, because of their correspondence with purely functional programs (or
$\lambda$-terms) and of their useful definability properties. A strategy $\sigma$ is \textbf{innocent} if
\[
sab\in \sigma \wedge t\in \sigma \wedge ta\in \mathcal{L}_{A} \wedge \pview{sa}=\pview{ta} \implies tab\in \sigma
\]
Intuitively, an innocent strategy only takes its $P$-view into account to determine its next move. Indeed, any innocent strategy is characterized by
a set of $P$-views. This observation is very important since $P$-views can be seen as abstract representations of branches of $\eta$-expanded Böhm trees
(\emph{a.k.a.} Nakajima trees \cite{nakajima1975infinite}): this is the key to the definability process on innocent strategies. 
Arenas and innocent strategies form a cartesian closed category and
are therefore a model of $\Lambda$; let us add for completeness that $o$ is interpreted as the singleton arena and 
$\daimon_A$ is interpreted as the singleton strategy on $\intr{A}$ containing only the empty play.

\subsubsection{Bounded strategies.} A strategy $\sigma$ is \textbf{bounded} if it is $P$-visible and if the length of its $P$-views
is bounded: formally, there exists $N\in \mathbb{N}$ such that for all $s\in \sigma$, $|\pview{s}|\leq N$. Bounded strategies only
have finite interactions \cite{totality}; this result
corresponds loosely to the normalisation result on simply-typed $\lambda$-calculus. Syntactically, bounded strategies include the
interpretation of all terms of a higher-order programming language with ground type references, arbitrary non-determinism and
control operators, but without recursion. 
This remark is important since it implies that our results will hold for any program written with these constructs,
as long as they do not use recursion or a fixed point operator.

Our complexity results, in this game-theoretic part of this paper, will be expressed as a function of the \textbf{size}
of the involved bounded strategies, given by:
\[
|\sigma| = \frac{max_{s\in \sigma} |\pview{s}|}{2}
\]
When $\sigma$ is an innocent strategy coming from a Böhm tree, then $|\sigma|$ is proportional to the \emph{height} of this 
Böhm tree, as defined in the syntactic part of this paper.

\subsection{Visible pointer structures}

\subsubsection{Introduction of visible pointer structures}

We note here that the notion of $P$-view (and of its size) only takes into account the structure of pointers within the plays,
and completely ignores the actual labels of the moves. In fact the underlying \emph{pointer structure} of a play will be all we
need to study the asymptotic complexity of execution. 

\begin{defi}
The \textbf{pure arena} $\purea$ is defined by:
\begin{eqnarray*}
M_{\purea} &=& \mathbb{N}\\
\lambda_{\purea} &=& \left\{\begin{array}{l} 2n \mapsto O\\2n+1 \mapsto P\end{array}\right.\\
\enb_{\purea} &=& \{(n,n+1)\mid n\in \mathbb{N}\}\\
I_{\purea} &=& \{0\}
\end{eqnarray*}
A \textbf{pointer structure} is a legal play $s\in \leg{\purea}$ with at most one initial move. Note that for any arena
$A$, a legal play $s\in \leg{A}$ with only one initial move can always be mapped to its pointer structure $\bar{s}\in \leg{\purea}$
by sending each move $s_i$ to its \emph{depth}, \emph{i.e.} the number of pointers to be taken in $s$ before reaching the initial
move.

The \textbf{depth} of a pointer structure $s$ is the largest depth of $s_i\in \mathbb{N}$, for $0\leq i \leq |s|-1$.
\end{defi}

Pointer structures retain some information about the control flow of the execution, but on the other hand forget most of the typing
information. Our results will rely on the crucial assumption that the strategies involved act in a \emph{visible} way. This is
necessary, because non-visible strategies are able to express programs with general (higher-order) references within which a
fixpoint operator is definable, so termination is lost. 

\begin{defi}
A \textbf{visible pointer structure} $s$ is a pointer structure $s\in \leg{\purea}$ such that:
\begin{itemize}
\item It is $P$-visible: for any $s'a \sqsubseteq^P s$, $a$ points within $\pview{s'}$.
\item It is $O$-visible: for any $s'a \sqsubseteq^O s$, $a$ points within $\oview{s'}$.
\end{itemize}
We write $\vis_d$ for the set of all visible pointer structures of depth lower than $d$.
\end{defi}

We are interested in the maximal length of a visible pointer structure resulting from the interaction between two bounded strategies.
Since the collapse of plays to visible pointer structures forgets the identity of moves, all that remains from bounded strategies
in this framework is their \emph{size}. Therefore, for $n, p\in \mathbb{N}$, we define the set $n \star_d p$ of all visible
pointer structures possibly resulting from an interaction of strategies of respective sizes $n$ and $p$, in an ambient arena
of depth $d$. Formally:
\[
n\star_d p = \{ s\in \vis_d \mid \forall s'\sqsubseteq s, |\pview{s'}| \leq 2n \wedge |\oview{s'}|\leq 2p+1\}
\]

In \cite{totality}, we already examined the termination problem for visible pointer structures. We proved that any interaction
between bounded strategies is necessarily finite. Therefore since $n\star_d p$, regarded as a tree, is finitely branching, it
follows that it is finite. So there is an upper bound $N_d(n,p)$ to the length of any visible pointer structure in $n\star_d p$.

\subsubsection{The visible pointer structure of an interaction} We take here the time to detail more formally our claim
that an estimation of the length of visible pointer structures in $n\star_d p$ is informative of the complexity of interaction
between bounded strategies.

Let $\sigma : A \tto B$ and $\tau : B \tto C$ be bounded strategies, of respective size $n$ and $p$. We are interested in the
possible length of interactions in $\sigma \parallel \tau$. Of course arbitrary such interactions are not bounded in size, since
$\sigma$ and $\tau$ both interact with an external opponent in $A$ and $C$, whose behavior is not restricted by any size condition.
Therefore we restrict our interest to \textbf{passive} such interactions, \emph{i.e.} interactions $u\in \sigma \parallel \tau$
such that the unique Opponent move in $u\restrict A, C$ is a unique initial question in $C$. When $\sigma$ and $\tau$ are both
innocent and correspond to $\lambda$-terms, this corresponds by the results of \cite{dhr} to a linear head reduction sequence from
a game situation. In particular, the \emph{passivity} condition ensures that the interaction stops if a free variable ever arrives
in head position.

\begin{prop}
Let $\sigma : A \tto B$ be a bounded strategy of size $p$, let $\tau: B \tto C$ be a bounded strategy of size $n$, and suppose
that $B$ has depth $d-1$, for $d\geq 2$. Then for all passive $u\in \sigma \parallel \tau$:
\[
|u| \leq N_d(n, p) + 1
\]
\end{prop}
\begin{proof} 
Since $u$ is passive, it consists in a play $\circ_C u'$ with $\circ_C$ initial in $C$ and $u'\in \leg{B}$, possibly
followed by a trailing move $\bullet$ in $A$ or $C$. Take any strict prefix $\circ_C u'' \sqsubseteq u$.
Then $u''$ is either in $\sigma$ or an immediate prefix of a play in $\sigma$, so $|\pview{u''}| \leq 2p$. Accordingly, 
$|\oview{\circ_C u''}| \leq 2p+1$. Likewise $\circ_C u''$ is either in $\tau$ or an immediate prefix of a play in $\tau$, 
so $|\pview{\circ_C u''}| \leq 2n$. 

Writing $\overline{u}$ for the \textbf{pointer structure} of a play $u$ (\emph{i.e.} the play of $\purea$ obtained by sending
any move to its depth). Then we have $\overline{\circ_C u'} \in n\star_d p$. Indeed the size of $P$-views and $O$-views in 
a play only depends on pointers, so it is unchanged by forgetting arena labels. Moreover $\circ_C u' \in \leg{B\tto \{\circ_C\}}$
(where $\{\circ_C\}$ is the singleton arena), which has depth $d$. 

From the above, we deduce that $|\circ_C u'|\leq N_d(n,p)$, so $|u|\leq N_d(n,p) + 1$.
\end{proof}

So the study of visible pointer structure is sufficient to study the length of interactions in terms of the
sizes of the strategies involved. 

\subsubsection{Interaction skeletons and simulation of visible pointer structures}

We now introduce \emph{interaction skeletons} (or just \emph{skeletons} for short), the main tool used in this paper in order to
study the complexity of execution.

As we mentioned repeatedly, game-theoretic interaction corresponds to linear head reduction --- which is itself efficiently
implemented by machines with environment such as the Krivine Abstract Machine (KAM). In such machines, game situations produce
by reduction situations where the terms interacting are no longer plain closed terms but rather terms-in-environments, also
known as \emph{closures}. Following this phenomenon, whereas the measure of the first move $s_0$ of a visible pointer 
structure is given by the sizes of the strategies involved (so by a pair of natural numbers), the measure of a later move $s_i$
will be given by a finite tree of natural numbers reminiscent of the structure of \emph{closures}.

We will call a \textbf{pointed visible pointer structure} a pair $(s, i)$ where $s$ is a visible pointer structure and
$i\leq |s|-1$ is an arbitrary ``starting" move. We adapt the notions of size and depth for them, and introduce a notion of
\emph{context}.

\begin{defi}
Let $(s,i)$ be a pointed visible pointer structure. The \textbf{residual size} of $s$ at $i$, written $\size(s, i)$, is defined as follows:
\begin{itemize}
\item If $s_i$ is an Opponent move, it is $\max_{s_i \in \pview{s_{\leq j}}} |\pview{s_{\leq j}}| - |\pview{s_{\leq i}}| + 1$
\item If $s_i$ is a Player move, it is $\max_{s_i \in \oview{s_{\leq j}}} |\oview{s_{\leq j}}| - |\oview{s_{\leq i}}| + 1$
\end{itemize}
where $s_i \in \pview{s_{\leq j}}$ means that the computation of $\pview{s_{\leq j}}$ reaches\footnote{So starting from $s_j$ and following Opponent's pointers
eventually reaches $s_i$.} $s_i$. Dually, we have the notion of
\textbf{residual co-size} of $s$ at $i$, written $\cosize(s, i)$, defined as follows:
\begin{itemize}
\item If $s_i$ is an Opponent move, it is $\max_{s_i \in \oview{s_{\leq j}}} |\oview{s_{\leq j}}| - |\oview{s_{\leq i}}| + 1$
\item Otherwise, $\max_{s_i \in \pview{s_{\leq j}}} |\pview{s_{\leq j}}| - |\pview{s_{\leq i}}| + 1$
\end{itemize}
The \emph{residual depth} of $s$ at $i$ is the maximal length of a pointer chain in $s$ starting from $s_i$.
\end{defi}

\begin{defi}
Let $s$ be a visible pointer structure. We define the \textbf{context} of $(s,i)$ as:
\begin{itemize}
\item If $s_i$ is an O-move, the set $\{s_{n_1}, \dots, s_{n_p}\}$ of O-moves appearing in $\pview{s_{< i}}$,
\item If $s_i$ is a P-move, the set $\{s_{n_1}, \dots, s_{n_p}\}$ of P-moves appearing in $\oview{s_{< i}}$.
\end{itemize}
In other words it is the set of moves to which $s_{i+1}$ can point whilst abiding to the visibility condition, except $s_i$. We also need the dual notion
of co-context, which contains the moves the other player can point to. The \textbf{co-context} of $(s, i)$ is:
\begin{itemize}
\item If $s_i$ is an O-move, the set $\{s_{n_1},\dots, s_{n_p}\}$ of P-moves appearing in $\oview{s_{< i}}$,
\item If $s_i$ is a P-move, the set $\{s_{n_1},\dots, s_{n_p}\}$ of O-moves appearing in $\pview{s_{< i}}$.
\end{itemize}
\end{defi}

\begin{defi}
A \textbf{skeleton} is a finite tree, whose nodes and edges are both labeled by natural numbers. If $a_1, \dots, a_p$ 
are skeletons and $d_1, \dots, d_p$ are natural numbers, we write:
\[
n[\{d_1\} a_1, \dots, \{d_p\} a_p] = 
\raisebox{20pt}{\xymatrix{
&n      \ar@{-}[dl]_{d_1}
        \ar@{-}[dr]^{d_p}\\
a_1&\dots&a_p
}}
\]
\end{defi}

We now define what it means for $(s, i)$ to respect a skeleton $a$.

\begin{defi}[Trace, co-trace, interaction]
The two notions $Tr$ and $coTr$ are defined by mutual recursion, as follows:
let $a = n[\{d_1\}a_1, \dots, \{d_p\}a_p]$ be a skeleton. We say that $(s,i)$ is a \textbf{trace} (resp. a \textbf{co-trace}) of $a$, denoted $(s,i)\in Tr(a)$ (resp.
$(s, i)\in coTr(a)$) if the following conditions are satisfied:
\begin{itemize}
\item $\size(s, i) \leq 2n$ (resp. $\cosize(s, i) \leq 2n+1$),
\item If $\{s_{n_1}, \dots, s_{n_p}\}$ is the context of $(s,i)$ (resp. co-context), then for each $k\in \{1, \dots, p\}$ we have $(s,n_k) \in coTr(a_k)$.
\item If $\{s_{n_1}, \dots, s_{n_p}\}$ is the context of $(s,i)$ (resp. co-context), then for each $k\in \{1, \dots, p\}$ the residual depth of $s$ at $n_k$ is less than $d_k$.
\end{itemize}
Then, we define an \textbf{interaction} of two skeletons $a$ and $b$ at depth $d$ as a pair $(s, i)\in Tr(a)\cap coTr(b)$ where the residual depth of $s$ at $i$ is less than $d$, which we write $(s, i)\in a\star_d b$.
\end{defi}

Notice that we use the same notation $\star$ both for natural numbers and skeletons. This should not generate any
confusion, since the definitions above coincide with the previous ones in the special case of ``atomic" skeletons:
if $n$ and $p$ are natural numbers, then obviously $s\in n \star_d p$ (according to the former definition) if and only if
$(s, 0)\in n[] \star_d p[]$ (according to the latter). In fact, we will sometimes in the sequel write $n$ for the atomic
skeleton with $n$ at the root and without children, \emph{i.e.} $n[]$.

\subsubsection{Simulation of visible pointer structures}

We introduce now our main tool, a reduction on skeletons which simulates visible pointer structures: if $n[\{d_1\}a_1, \dots, \{d_p\}a_p]$ and $b$ are skeletons ($n>0$),
we define the non-deterministic reduction relation $\leadsto$ on triples $(a, d, b)$, where $d$ is a depth (a natural number) and $a$ and $b$ are skeletons, by the following two cases:
\begin{eqnarray*}
(n[\{d_1\}a_1, \dots, \{d_p\}a_p],d, b) &\leadsto& (a_i,d_i-1,(n-1)[\{d_1\}a_1, \dots, \{d_p\}a_p, \{d\}b])\\
(n[\{d_1\}a_1, \dots, \{d_p\}a_p],d, b) &\leadsto& (b, d-1, (n-1)[\{d_1\}a_1, \dots, \{d_p\}a_p, \{d\}b])
\end{eqnarray*}
where $i\in \{1, \dots , p\}$, $d_i>0$ in the first case and $d>0$ in the second case. 

In order to prove our simulation result, we will make use of the following lemma.

\begin{lem}
Let $s$ be a pointed visible pointer structure and $a=n[\{d_1\}a_1, \dots, \{d_p\}a_p]$ a skeleton such that $(s, i)\in coTr(a)$. Then
if $s_j\rightarrow s_i$, $(s, j)\in Tr(a)$.
\label{lem_point}
\end{lem}
\begin{proof} 
Let us suppose without loss of generality that $s_i$ is an Opponent move; the other case can be obtained just by switching Player/Opponent and $P$-views/$O$-views
everywhere. Then $s_j$ being a Player move, we have to check first that $\size(s, j)\leq 2n$, \emph{i.e.}
\[
\max_{s_j \in \oview{s_{\leq k}}} |\oview{s_{\leq k}}| - |\oview{s_{\leq j}}| + 1 \leq 2n
\]
We use that $\cosize(s, i) \leq 2n+1$, \emph{i.e.}
\[
\max_{s_i \in \oview{s_{\leq k}}} |\oview{s_{\leq k}}| - |\oview{s_{\leq i}}| + 1 \leq 2n+1
\]
But $s_j \rightarrow s_i$, hence $|\oview{s_{\leq j}}| = |\oview{s_{\leq i}}| + 1$ and the inequality is obvious. We need now to examine
the context of $(s, j)$. Since $s_j$ is a Player move, it is defined as the set $\{s_{n_1}, \dots, s_{n_p}\}$ of Player moves appearing in
$\oview{s_{<j}}$, which is also the set of Player moves appearing in $\oview{s_{<i}}$ and therefore the co-context of $(s, i)$. But $(s, i)\in coTr(a)$,
hence for all $k\in \{1, \dots, p\}$ we have $(s, n_k)\in coTr(a_k)$ which is exactly what we needed.
\end{proof}

\begin{prop}[Simulation]
Let $(s,i)\in a\star_d b$, then if $s_{i+1}$ is defined, there exists $(a,d,b)\leadsto (a',d',b')$ such that $(s, i+1)\in a'\star_{d'} b'$.
\label{prop:simulvps}
\end{prop}
\begin{proof}
Suppose $a = n[\{d_1\}a_1, \dots, \{d_p\}a_p]$.
Let $\{s_{n_1}, \dots, s_{n_p}\}$ be the context of $(s, i)$. By visibility, $s_{i+1}$ must either point to $s_i$ or to an element of the context. Two cases:
\begin{itemize}
\item If $s_{i+1} \rightarrow s_i$, then we claim that $(s, i+1) \in b\star_{d-1} (n-1)[\{d_1\}a_1, \dots, \{d_p\}a_p, \{d\}b]$, \emph{i.e} $(s, i+1) \in Tr(b)$,
$(s, i+1) \in coTr((n-1)[\{d_1\}a_1, \dots, \{d_p\}a_p, \{d\}b])$ and the depth of $s$ relative to $i+1$ is at most $d-1$. For the first part, we use
that $(s, i)\in a\star_d b$: in particular, $(s, i)\in coTr(b)$ and since $s_{i+1} \rightarrow s_i$ this implies by Lemma \ref{lem_point} that $(s, i+1)\in Tr(b)$.
For the second part, we must first check that $\cosize(s, i+1) \leq 2(n-1) + 1$. Let us suppose without loss of generality that $s_i$ is an Opponent move, all the
reasoning below can be adapted by switching Player/Opponent and $P$-views/$O$-views everywhere. We want to prove:
\[
\cosize(s,i+1)  = \max_{s_{i+1} \in \pview{s_{\leq j}}} |\pview{s_{\leq j}}| - |\pview{s_{\leq i+1}}| + 1 \leq 2(n-1) + 1
\]
But since $(s, i) \in Tr(a)$, we already know:
\[
\size(s, i) = \max_{s_i \in \pview{s_{\leq j}}} |\pview{s_{\leq j}}| - |\pview{s_{\leq i}}| + 1 \leq 2n
\]
Thus we only need to remark that $|\pview{s_{\leq i+1}}| = |\pview{s_{\leq i}}| + 1$ since $s_{i+1}$ is a Player move.
Now, we must examine the co-context of $(s, i+1)$, but by definition of $P$-view it is $\{s_{n_1}, \dots, s_{n_p}, s_i\}$ where $\{s_{n_1}, \dots, s_{n_p}\}$
is the context of $(s, i)$. Since $(s,i) \in Tr(n[a_1, \dots, a_p])$ we have as required $(s,n_k)\in coTr(a_k)$ for each $k\in \{1, \dots, p\}$
and $(s,i)\in coTr(b)$ because $(s, i)\in a\star_d b$. For the third part, we have to prove that the depth of $s$ relative to $i+1$ is at most $d-1$, but
it is obvious since the depth relative to $i$ is at most $d$ and $s_{i+1}\rightarrow s_i$.

\item Otherwise, we have $s_{i+1} \rightarrow s_{n_j}$ for $j\in \{1, \dots, p\}$. Then, we claim
that $(s, i+1) \in a_j \star_{d_i - 1} (n-1)[\{d_1\}a_1, \dots, \{d_p\}a_p, \{d\}b]$. We do have
$(s, i+1) \in Tr(a_j)$ because $(s, i)\in Tr(n[\{d_1\}a_1, \dots, \{d_p\}a_p])$, thus $(s, n_j)\in coTr(a_j)$ and $(s, i+1) \in Tr(a_j)$ by Lemma \ref{lem_point}.
It remains to show that $(s, i+1) \in coTr((n-1)[\{d_1\}a_1, \dots, \{d_p\}a_p, \{d\}b])$ and that the depth of $s$ relative to $i+1$ is at most $d_1 - 1$, but
the proofs are exactly the same as in the previous case.\qedhere
\end{itemize}
\end{proof}

\noindent Before going on to the study of skeletons, let us give a last simplification. If 
$a = n[\{d_1\}t_1, \dots, \{d_p\}t_q]$ and $b$ are skeletons, then $a\cdot_d b$ will denote the skeleton obtained by appending $b$ as a
new son of the root of $a$ with label $d$, \emph{i.e.} $n[\{d_1\}t_1, \dots, \{d_p\}t_q, \{d\}b]$. Consider the following
non-deterministic rewriting rule on skeletons:
\[
n[\{d_1\}a_1, \dots, \{d_p\}a_p] \leadsto a_i \cdot_{d_i - 1} (n-1)[\{d_1\}a_1, \dots, \{d_p\}a_p]
\]
Both rewriting rules on triples $(a, d, b)$ are actually instances of this reduction, by the isomorphism 
$(a, d, b) \mapsto a \cdot_d b$. We leave the obvious verification to the reader. Taking this into account, 
all that remains to study is this reduction on skeletons illustrated in Figures \ref{rewrite} and \ref{fig_examplereduction}, 
and analyzed in the next section. 

If $\norm(a)$ denotes the length of the longest reduction sequence starting from a skeleton $a$, we have the following property.

\begin{prop}
Let $n, p\geq 0$, $d\geq 2$, then $N_d(n,p) \leq \norm(n[\{d\}p])+1$.
\label{equiv}
\end{prop}
\begin{proof} 
Obvious from Proposition \ref{prop:simulvps}, adding $1$ for the initial move which is not accounted for by the reduction on skeletons. 
\end{proof}

We postpone the analysis of the reduction of skeletons to the next section. However with the results proved there, we get
the following result:

\begin{thm}
For fixed $d\geq 3$, we have
$N_d(n,p) = 2_{d-2}^{\Theta(n\log(p))}$.
\label{th:boundvps}
\end{thm}
\begin{proof}
\emph{Upper bound.} Follows from Proposition \ref{equiv} and Theorem \ref{th:gamesituations}.

\emph{Lower bound.} To construct the example providing the lower bound, it is convenient to consider the extension $\Lambda_\times$
of $\Lambda$ with finite product types $\Pi_{1 \leq i \leq n} A_i$. Tupling of $\Gamma \vdash M_i : A_i$ for $1\leq i \leq n$ is
written  $\lpair{M_1, \dots, M_n}$. For simplicity we use $\lambda \lpair{x_1, \dots, x_n}^{\Pi_{1\leq i \leq n} A_i}. M$ as
syntactic sugar for $\lambda x^{\Pi_{1\leq i \leq n} A_i}.~M[\pi_i x/x_i]$, where $\pi_i : \Pi_{1\leq i \leq n} A_i \to A_i$ is the
corresponding projection. Products are interpreted in the games model following standard lines \cite{hogames}.
In $\Lambda_\times$ (and $\Lambda$), we define higher types for Church integers by setting $A_{-2} = o$ and $A_{n+1} = A_n \to A_n$.
For $n, p\in \mathbb{N}$, we write $\church{n}_p$ for the Church integer for $n$ of type $A_p$. For $A$ a type, if $M:A\to A$ and $N:A$,
we write $M^n(N)$ for the $n$-th iteration of $M$ applied to $N$. Then for $n, p, d\in \mathbb{N}$, we define:
\[
(\lambda \lpair{x, y, z_1, \dots, z_d}.~x^n(y)~z_1~\dots~z_d)~\lpair{\church{p}_{d+1}, \church{2}_d, \church{2}_{d-1}, \dots, \church{2}_0}
\]
By elementary calculations on $\lambda$-terms, we know that this is $\beta$-equivalent to $\church{2_{d+1}^{p^n}}_0$. So, taking a maximal 
(necessarily passive) interaction:
\[
u \in \intr{\lpair{\church{p}_{d+1}, \church{2}_d, \church{2}_{d-1}, \dots, \church{2}_0, \id_o}}\parallel \intr{\lambda \lpair{x, y, z_1, \dots, z_d, i}.~x^n(y)~z_1~\dots~z_d~i}
\]
we know that $u$ must have length at least $2_{d+1}^{p^n}$. Inspecting these strategies, we see that the left hand side one
has size $p+d+2$ and the right hand side one has size $n+d+3$, and that they interact in an arena $\intr{\Pi_{i=0}^{d+2} A_{d+1-i}}$
of depth $d+3$. It follows that the underlying pointer structure of $u$ (of length at least $2_{d+1}^{p^n}$) is in $(n+d+3)\star_{d+4} (p+d+2)$.
Therefore,
\[
N_{d+4}(n+d+3, p+d+2) \geq 2_{d+1}^{p^n}
\]
So for $d\geq 4, n\geq d+3, p\geq d+2$ we have $N_d(n,p) \geq 2_{d-2}^{(n-d-3)\log(p-d-2)}$, providing the lower bound. Although
this example only proves the other bound for $d\geq 4$ it also holds for $d = 3$: this is proved using the same reasoning on the 
maximal interaction between strategies $\intr{\church{p}_0}$ and $\intr{\lambda x.~x^n(\id_o)}$, having length at least $p^n$.
\end{proof}

The strength of this result is that being a theorem about interactions of strategies in game semantics, its scope includes any
programming language whose terms can be interpreted as bounded strategies. Its weakness however, is that it only applies
syntactically to game situations. 
In order to increase the generality of our study and give exact bounds to the linear head reduction of arbitrary simply-typed
$\lambda$-terms, we will in Section \ref{sec:lhr} detail a direct connection between linear head reduction and reduction of 
skeletons. Before that, as announced we focus Section \ref{sec:bs} on the analysis of skeletons.

\section{Skeletons and their complexity analysis}
\label{sec:bs}

In the previous section we proved a simulation result of \emph{plays} in the sense of Hyland-Ong games into a reduction on
simple combinatorial objects, \emph{interaction skeletons}. In the present section we investigate the properties of
this reduction independently of game semantics or syntax, proving among other things Theorem \ref{th:gamesituations}
used in the previous section.

As announced in the introduction, the rest of this paper -- starting from here -- is essentially self-contained. We start this section
by recalling the definition of interaction skeletons and investigating their basic properties. Then, we will prove our main
result about the length of their reduction.

\subsection{Skeletons and their basic properties}

\subsubsection{Skeletons and their dynamics}

\textbf{Interaction skeletons}, or \emph{skeletons} for short, are finite trees of natural numbers, whose nodes and edges are labeled
by natural numbers. To denote these finite trees, we use the notation illustrated below:
\[
n[\{d_1\} a_1, \dots, \{d_p\} a_p] = 
\raisebox{20pt}{\xymatrix{
&n      \ar@{-}[dl]_{d_1}
        \ar@{-}[dr]^{d_p}\\
a_1&\dots&a_p
}}
\]
Each natural number $n$ can be seen as an atomic skeleton $n[]$ without subtrees, still denoted by $n$. That should
never cause any confusion.
Given a skeleton $a$, we define:
\begin{itemize}
\item Its \textbf{order} $\ord(a)$, the maximal edge label in $a$,
\item Its \textbf{maximum} $\max(a)$, the maximal node label in $a$,
\item Its \textbf{depth} $\depth(a)$, the maximal depth of a node in $a$, the root having depth $1$.
\end{itemize}

We will also use the notation $a\cdot_d b$ for the skeleton $a$ with a new child $b$ added to the root of $a$, with
edge label $d$. Formally if $a = n[\{d_1\}a_1, \dots, \{d_p\}a_p]$:
\[
a \cdot_d b = n[\{d_1\} a_1, \dots, \{d_p\} a_p, \{d\} b]
\]

With that in place, we define the reduction on skeletons by:
\[
n[\{d_1\}a_1, \dots, \{d_p\}a_p] \leadsto a_i \cdot_{d_i - 1} (n-1)[\{d_1\}a_1, \dots, \{d_p\}a_p]
\]
which is allowed whenever $n, d_i \geq 1$ so that no index ever becomes negative. The reduction is illustrated in
Figure \ref{rewrite}, and performed on an example in Figure \ref{fig_examplereduction} (where the node selected for
the next reduction is highlighted). It is important to insist that this reduction is only ever performed in head position: $n$
needs to be the actual root of the tree for the reduction to be allowed. We do not know which properties of
this reduction are preserved in the generalized system where reduction can occur everywhere. 

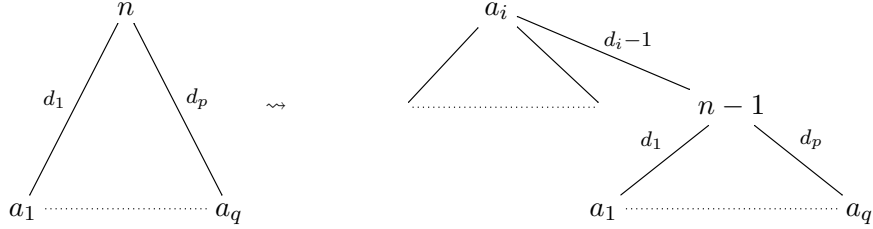
\begin{figure}
\[
\xymatrix{
&n\ar@{-}[ddr]^{d_p}
  \ar@{-}[ddl]_{d_1}&&&&a_i\ar@{-}[dl]
                    \ar@{-}[dr]
                \ar@{-}[drr]^{d_i-1}\\
                &&\ar@{}[r]|{\leadsto}&&\ar@{.}[rr]&&&n-1
                                        \ar@{-}[dr]^{d_p}
                                        \ar@{-}[dl]_{d_1}\\
a_1\ar@{.}[rr]&&a_q&&&&a_1\ar@{.}[rr]&&a_q
}
\]
\caption{Rewriting rule on skeletons}
\label{rewrite}
\end{figure}

\begin{figure}
\begin{center}
\includegraphics[scale=.25]{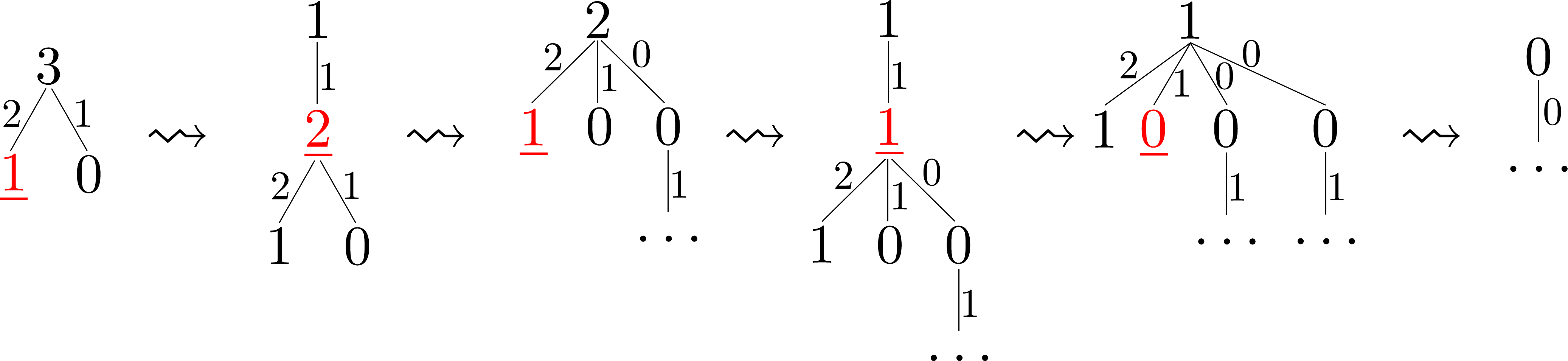}
\end{center}
\caption{Example reduction sequence on interaction skeletons}
\label{fig_examplereduction}
\end{figure}

Let us write $\norm(a)$ for the \textbf{norm} of a skeleton $a$, \emph{i.e.} the length of its longest
reduction sequence. We will show later that it is always finite; in the meantime for definiteness we define
it as a member of $\mathbb{N} \cup \{+\infty\}$.

\subsubsection{Embedding lemma}

The norm of a skeleton is unchanged by permutation of subtrees, or merging of identical subtrees, and is only
increased by an increase of labels.
If $a = n[\{d_1\}a_1, \dots, \{d_p\}a_p]$ and $a' = n'[\{d'_1\}a'_1, \dots, \{d'_{p'}\}a'_{p'}]$, we say that $a$
\textbf{embeds} in $a'$, written $a\embeds a'$, if $n \leq n'$ and for any $i\in \{1, \dots, p\}$ there exists $j\in \{1, \dots, p'\}$ such that
$d_i\leq d'_j$ and $a_i \embeds a'_j$. Then we have:

\begin{lem}[Embedding lemma]
For any skeletons $a$ and $a'$ with $a \embeds a'$, $\norm(a) \leq \norm(a')$.
\label{lem:embedsbs}
\end{lem}
\begin{proof} 
Suppose $a \embeds a'$, and $a \leadsto b$. Write $a = n[\{d_1\} a_1, \dots, \{d_p\}a_p]$, and suppose
$b = a_i \cdot_{d_i-1} (n-1)[\{d_1\}a_1, \dots, \{d_p\}a_p]$. Likewise, write $a' = n'[\{d'_1\}a'_1, \dots, \{d'_{p'}\}a'_{p'}]$.
Since $a \embeds a'$ we have $n'\geq n \geq 1$ and there is $j\in \{1, \dots, p'\}$ such that $d'_j\geq d_i \geq 1$ and $a_i \embeds a'_j$. Set
$b' = a'_j \cdot (n'-1)[\{d'_1\}a'_1, \dots, \{d'_{p'}\}a'_{p'}]$, it is straightforward to check that $b \embeds b'$. Therefore,
$\embeds$ is a simulation, which implies that if $a \embeds a'$ we have $\norm(a) \leq \norm(a')$.
\end{proof}

From this it follows that permuting subtrees in a skeleton does not affect the possible reductions in any way.
Perhaps more surprisingly, it shows that two identical subtrees can be merged without any effect on the possible reductions: the
number of copies of identical subtrees does not matter. Following this idea, we are going to show that any skeleton
embeds into a simple thread-like one, and this only increases the length of possible reductions.

\begin{defi}
Let $d, o, m\geq 1$ be natural numbers. The \emph{thread-like} skeleton $\tbs(d, o, m)$ is:
\begin{eqnarray*}
\tbs(1, o, m) &=& m\\
\tbs(d+1, o, m) &=& m[\{o\}\tbs(d, o, m)]
\end{eqnarray*}
\end{defi}

From the definition, we have that $\depth(\tbs(d, o, m)) = d$, $\max(\tbs(d, o, m)) = m$ and $\ord(\tbs(d, o, m)) = o$. We also have:

\begin{prop}
If $a$ has depth $d$, order $o$ and maximum $m$, then $\norm(a) \leq \norm(\tbs(d, o, m))$.
\label{prop_merge}
\end{prop}
\begin{proof} 
It is obvious that $a \embeds \tbs(d, o, m)$, therefore the result follows from Lemma \ref{lem:embedsbs}.
\end{proof}

\subsubsection{Constructions on skeletons}

If $(a_i)_{1 \leq i \leq n}$ is a finite family of skeletons, then writing
$
a_i = n_i[\{d_{i, 1}\}b_{i, 1}, \dots, \{d_{i, p_i}\}b_{i, p_i}]
$,
we define:
\begin{eqnarray*}
\bigsqcup_{i=1}^n a_i &=& (\max_{1\leq i \leq n} n_i)\cdot [\{d_{i, j}\} b_{i, j} \mid 1\leq i \leq n~\&~1\leq j\leq p_i]\\
\sum_{i=1}^n a_i &=& (\sum_{i=1}^n n_i)\cdot [\{d_{i,j}\} b_{i, j} \mid 1\leq i \leq n~\&~1\leq j\leq p_i]
\end{eqnarray*}
so, they either take the maximum or the sum of the roots, and simply append all the subtrees of the $a_i$s.
In the binary case, we write as usual $+$ for the sum.

\begin{lem}
We have the following embeddings:
\begin{itemize}
\item If $(a_i)_{1 \leq i \leq n}$, $(b_i)_{1 \leq i \leq n}$ are finite families of skeletons such that for all $1\leq i \leq n$,
we have $a_i \embeds b_i$, then
\begin{mathpar}
\bigsqcup_{i=1}^n a_i \embeds \bigsqcup_{i=1}^n b_i \and
\sum_{i=1}^n a_i \embeds  \sum_{i=1}^n b_i
\end{mathpar}
\item If $(a_i)_{1\leq i \leq n}$ is a finite family of skeletons and $b$ is a skeleton, then
\begin{mathpar}
\bigsqcup_{i=1}^n (a_i\cdot b) \embeds (\bigsqcup_{i=1}^n a_i )\cdot b \and
\sum_{i=1}^n (a_i \cdot b) \embeds (\sum_{i=1}^n a_i)\cdot b
\end{mathpar}
\item If $a, b, c$ are skeletons and $d\in \mathbb{N}$, then:
\[
a + b\cdot_d c \embeds (a+b)\cdot_d c
\]
\end{itemize}
\label{lem:op_embeddings}
\end{lem}
\begin{proof} 
Direct from the definitions.
\end{proof}

\subsection{Upper bounds}

We calculate upper bounds to the length of possible reductions on skeletons. This is done by adapting
a technique used by Schwichtenberg \cite{schw91} and Beckmann \cite{beck} to bound  the length of possible
$\beta$-reduction sequences on simply-typed $\lambda$-terms.

The idea of the proof is to define an inductive predicate $\sststile{\rho}{\alpha}$ on terms/skeletons whose
proofs/inhabitants combine aspects of a syntax tree and of a reduction tree --- witnesses of this predicate are called
\emph{expanded reduction trees} by Beckmann \cite{beck}. Their mixed nature will allow us to define a transformation
gradually eliminating their \emph{syntactic} (or \emph{static}) nodes, yielding an alternative
expanded reduction tree for the term/skeleton under study, whose height is more easily controlled.

\begin{defi}
The predicate $\sststile{\rho}{\alpha}$ (where $\rho, \alpha$ range over natural numbers) is defined on skeletons in the following inductive way.
\begin{itemize}
\item \textsc{Base.} $\sststile{\rho}{\alpha} 0[\{d_1\} a_1, \dots, \{d_p\} a_p]$
\item \textsc{Red.} Suppose $a = n[\{d_1\}a_1, \dots, \{d_p\} a_p]$ for $n\geq 1$. Then if for all $a'$ such that $a\leadsto a'$ we have $\sststile{\rho}{\alpha} a'$ and if
we also have $\sststile{\rho}{\alpha}\max(n-1, 0)[\{d_1\}a_1, \dots, \{d_p\} a_p]$, then $\sststile{\rho}{\alpha+1} a$.
\item \textsc{Cut.} If $\sststile{\rho}{\alpha} a$, $\sststile{\rho}{\beta} b$ and $d\leq \rho$, then $\sststile{\rho}{\alpha+\beta} a \cdot_d b$.
\end{itemize}
\end{defi}

\begin{defi}
A \textbf{context-skeleton} $a()$ is a finite tree whose edges are labeled by natural numbers, and whose nodes are labeled either by natural numbers, or by the variable $x$, with the constraint that
all edges leading to $x$ must be labeled by the same number $d$; $d$ is called the \textbf{type} of $x$ in $a()$.
We denote by $a(b)$ the result of substituting all occurrences of $x$ in $a()$ by $b$. We denote by $a(\emptyset)$ the skeleton obtained by deleting in $a$
all occurrences of $x$, along with the edges leading to them.
\end{defi}

\begin{lem}[Monotonicity]
If $\sststile{\rho}{\alpha} a$, then $\sststile{\rho'}{\alpha'} a$ for all $\alpha \leq \alpha'$ and $\rho \leq \rho'$, where
the witness trees have the same number of occurrences of \textsc{Cut}.
\label{monotonicity}
\end{lem}
\begin{proof} 
Straightforward by induction on the derivation for $\sststile{\rho}{\alpha} a$.
\end{proof}

\begin{lem}[Permutation lemma]
If $a$ is obtained from $a'$ by permuting some subtrees in $a$, then for all $\rho, \alpha$,  $\sststile{\rho}{\alpha} a$ iff
$\sststile{\rho}{\alpha} a'$.
\label{permutation}
\end{lem}
\begin{proof} 
Straightforward by induction on the derivation for $\sststile{\rho}{\alpha} a$.
\end{proof}

\begin{lem}[Null substitution lemma]
If $\sststile{\rho}{\alpha} a(\emptyset)$ and the type of $x$ in $a$ is $0$, then for all $b$ we still have $\sststile{\rho}{\alpha} a(b)$. Moreover, the witness includes as many \textsc{Cut} rules
as for $\sststile{\rho}{\alpha} a(\emptyset)$.
\label{first_substitution}
\end{lem}
\begin{proof}
We prove by induction on derivations $\sststile{\rho}{\alpha} a$
that the property above holds for all context-skeleton $a'$ such that the type of $x$ in $a'$ is $0$ and $a'(\emptyset) = a$.
\begin{itemize}
\item \textsc{Base.} The root of $a$ is $0$, hence the result is trivial.
\item \textsc{Red.} Suppose $a'$ has the form $n[\{d_1\}a'_1, \dots, \{d_p\}a'_p, \{0\}x]$, where $a'_1, \dots, a'_p$ 
possibly include occurrences of $x$ (the case where $x$ appears as a son of the root encompasses the other)
and $a_i = a'_i(\emptyset)$. The premises
of \textsc{Red} are then that for $1 \leq i \leq p$ such that
$d_i\geq 1$, $\sststile{\rho}{\alpha-1} a_i \cdot_{d_i-1} (n-1)[\{d_1\}a_1, \dots \{d_p\}a_p]$
and $\sststile{\rho}{\alpha-1} \max(n-1, 0)[\{d_1\}a_1, \dots \{d_p\}a_p]$. The IH on these premises
give witnesses for the two following properties:
\begin{eqnarray}
\sststile{\rho}{\alpha-1} (a'_i \cdot_{d_i-1} (n-1)[\{d_1\}a'_1, \dots, \{d_p\}a'_p, \{0\}x])(b)\label{eq1}
\end{eqnarray}
\begin{eqnarray}
\sststile{\rho}{\alpha-1} (\max(n-1, 0)[\{d_1\}a'_1, \dots \{d_p\}a'_p, \{0\}x])(b)\label{eq2}
\end{eqnarray}
This covers all the possible reductions of $a'(b)$, thus by \textsc{Red} we
have $\sststile{\rho}{(\alpha-1)+1} a'(b)$ as required.
\item \textsc{Cut.} Let us suppose $\sststile{\rho}{\alpha+\gamma} a$ is obtained by \textsc{Cut}, hence $a$ has the 
form $a_1 \cdot_{d_1} a_2$. Let us assume that $a'$ has the form $(a'_1 \cdot_{d'} a'_2) \cdot_0 x$ with $a'_1(\emptyset) = a_1$
and $a'_2(\emptyset) = a_2$, since again the case where $x$ is a child of the root of $a'$ encompasses the other.
The premises of \textsc{Cut} are then
$\sststile{\rho}{\alpha} a_1$ and $\sststile{\rho}{\gamma} a_2$, and $d'\leq \rho$. We also
have $(a'_1 \cdot_0 x)(\emptyset) = a_1$, therefore
the IH on $\sststile{\rho}{\alpha} a_1$ along with $\sststile{\rho}{\beta} b$ and $0\leq \rho+1$ implies that
$\sststile{\rho}{\alpha} a'_1(b) \cdot_0 b$. But by IH we also have $\sststile{\rho}{\gamma}a'_2(b)$, hence by \textsc{Cut}:
\[
\sststile{\rho}{\alpha + \gamma} (a'_1(b) \cdot_0 b) \cdot_{d'} a'_2(b)
\]
Which was what was required for $(a'_1(b) \cdot_{d'} a'_2(b)) \cdot_0 b$, thus it suffices since trees are considered up
to permutation.\qedhere
\end{itemize}
\end{proof}

\begin{lem}[Main substitution lemma]
If $\sststile{\rho}{\alpha} a(\emptyset)$, $\sststile{\rho}{\beta} b$ and $d \leq \rho + 1$ (where $d$ is the type of $x$ in $a$), then $\sststile{\rho}{\alpha(\beta+1)} a(b)$
\label{main_substitution}
\end{lem}
\begin{proof}
We prove by induction on derivations $\sststile{\rho}{\alpha} a$ that the property above holds for all context-skeleton $a'$
such that the type of $x$ in $a'$ is $d\leq \rho+1$, and such that $a = a'(\emptyset)$.
\begin{itemize}
\item \textsc{Base.} The root of $a$ is $0$, hence the result is trivial.
\item \textsc{Red.} Suppose $a'$ has the form $n[\{d_1\}a'_1, \dots, \{d_p\}a'_p, \{d\}x]$, where $a_1 = a'_1(\emptyset), \dots, 
a_p = a'_p(\emptyset)$ (the case where $x$ appears as a son of the root encompasses the other). 
The premises of \textsc{Red} are that for $1 \leq i \leq p$ such
that $d_i\geq 1$, $\sststile{\rho}{\alpha-1} a_i \cdot_{d_i-1} (n-1)[\{d_1\}a_1, \dots \{d_p\}a_p]$
and $\sststile{\rho}{\alpha-1} (n-1)[\{d_1\}a_1, \dots \{d_p\}a_p]$. The IH on these premises give
witnesses for the two following properties:
\begin{eqnarray}
\sststile{\rho}{(\alpha-1)(\beta+1)} (a'_i \cdot_{d_i-1} (n-1)[\{d_1\}a'_1, \dots, \{d_p\}a'_p, \{d\}x])(b)\label{eq1}
\end{eqnarray}
\begin{eqnarray}
\sststile{\rho}{(\alpha-1)(\beta+1)} ((n-1)[\{d_1\}a'_1, \dots \{d_p\}a'_p, \{d\}x])(b)\label{eq2}
\end{eqnarray}
By hypothesis we have $\sststile{\rho}{\beta} b$, hence by \textsc{Cut} (since $d-1\leq \rho$), we have:
\begin{eqnarray}
\sststile{\rho}{(\alpha-1)(\beta+1) + \beta} b \cdot_{d-1} (n-1)[\{d_1\}a'_1(b), \dots \{d_p\}a'_p(b), \{d\}b]\label{eq3}
\end{eqnarray}
Using (\ref{eq1}) for all $i\in \{1, \dots, p\}$, (\ref{eq2}) (adjusted to $\sststile{\rho}{(\alpha-1)(\beta+1)+\beta}$ by Lemma \ref{monotonicity}) and (\ref{eq3}) we deduce by \textsc{Red} that
\[
\sststile{\rho}{(\alpha-1)(\beta+1)+ \beta + 1} n[\{d_1\}a'_1(b), \dots, \{d_p\}a'_p(b), \{d\}b]
\]
Which is what was required.
\item \textsc{Cut.} Let us suppose $\sststile{\rho}{\alpha+\gamma} a$ is obtained by \textsc{Cut}, hence $a$
has the form $a_1 \cdot_{d'} a_2$. Let us suppose that $a'$ has the form $(a'_1 \cdot_{d'} a'_2) \cdot_d x$ with
$a_1 = a'_1(\emptyset)$ and $a_2 = a'_2(\emptyset)$, since once again the case where $x$ is a child of the root
of $a'$ encompasses the other. The premises of \textsc{Cut} are then
$\sststile{\rho}{\alpha} a_1$ and $\sststile{\rho}{\gamma} a_2$, and $d'\leq \rho$. We also
have $(a'_1 \cdot_d x)(\emptyset) = a_1$, therefore
the IH on $\sststile{\rho}{\alpha} a_1$ along with $\sststile{\rho}{\beta} b$ and $d\leq \rho+1$ implies that
$\sststile{\rho}{\alpha(\beta+1)} a'_1(b) \cdot_d b$. But by IH we also have $\sststile{\rho}{\gamma(\beta+1) }a'_2(b)$,
hence by \textsc{Cut}:
\[
\sststile{\rho}{\alpha(\beta+1) + \gamma(\beta+1)} (a'_1(b) \cdot_d b) \cdot_{d'} a'_2(b)
\]
Which is what was required, up to permutation.\qedhere
\end{itemize}
\end{proof}

\noindent The following lemma is the core of the proof, allowing to eliminate instances of the \textsc{Cut} rule in the expanded
head reduction tree.

\begin{lem}[Cut elimination lemma]
Suppose $\sststile{\rho+1}{\alpha} a$. Then if $\alpha = 0$, $\sststile{\rho}{0} a$. Otherwise, $\sststile{\rho}{2^{\alpha-1}} a$.
\label{lem_cutelimination}
\end{lem}
\begin{proof} 
By induction on the tree witness for $\sststile{\rho+1}{\alpha} a$.
\begin{itemize}
\item \textsc{Base.} Trivial.
\item \textsc{Red.} Suppose $a=n[\{d_1\}a_1, \dots, \{d_p\}a_p]$. 
The premises of \textsc{Red} are $\sststile{\rho + 1}{\alpha-1} a_i \cdot_{d_i - 1} (n-1)[\{d_1\}a_1, \dots, \{d_p\}a_p]$ for all
$i\in \{1, \dots, p\}$ and $\sststile{\rho+1}{\alpha-1} (n-1)[\{d_1\}a_1, \dots, \{d_p\}a_p]$. If $\alpha \geq 2$, then it follows by
IH that $\sststile{\rho}{2^{\alpha-2}} a_i \cdot_{d_i - 1} (n-1)[\{d_1\}a_1, \dots, \{d_p\}a_p]$ and
$\sststile{\rho}{2^{\alpha-2}}(n-1)[\{d_1\}a_1, \dots, \{d_p\}a_p]$, which implies by \textsc{Red} and Lemma \ref{monotonicity} that
$\sststile{\rho}{2^{\alpha-1}} a$. If $\alpha = 1$, then the premises of \textsc{Red} are $\sststile{\rho + 1}{0} a_i \cdot_{d_i - 1}
(n-1)[\{d_1\}a_1, \dots, \{d_p\}a_p]$ for all $i\in \{1, \dots, p\}$ and $\sststile{\rho+1}{0} (n-1)[\{d_1\}a_1, \dots, \{d_p\}a_p]$. By induction
hypothesis this is still true with $\rho$ instead of $\rho+1$, thus by \textsc{Red} we have $\sststile{\rho}{1} [\{d_1\}a_1, \dots, \{d_p\}a_p]$
which is what we needed to prove.
\item \textsc{Cut.} Suppose $a= a_1 \cdot_d a_2$, the premises of \textsc{Cut} are $\sststile{\rho+1}{\alpha} a_1$, $\sststile{\rho+1}{\beta} a_2$ and $d\leq \rho+1$. If $\alpha, \beta \geq 1$
then by IH it follows that $\sststile{\rho}{2^{\alpha-1}} a_1$ and $\sststile{\rho}{2^{\beta-1}} a_2$, in particular if we define a context-skeleton $a'_1 = a_1 \cdot_d x$ we have
$\sststile{\rho}{2^{\alpha-1}} a'_1(\emptyset)$, hence by the substitution lemma (since $d\leq \rho+1$) we have
$\sststile{\rho}{2^{\alpha-1} (2^{\beta-1}+1)} a'_1(a_2) = a_1 \cdot_d a_2 = a$, thus $\sststile{\rho}{2^{\alpha+\beta-1}} a$ thanks to Lemma \ref{monotonicity} (since it is always true
than $2^{\alpha+\beta-1} \geq 2^{\alpha-1} (2^{\beta-1}+1)$).
If $\alpha = 0$ then by IH we have $\sststile{\rho}{0} a_1$ and $\sststile{\rho}{\beta'} a_2$. We use then the substitution lemma (since $d\leq \rho+1$) to get $\sststile{\rho}{0}(a_1 \cdot_d a_2)$, which
is stronger that what was required whatever was the value of $\beta$. The last remaining case is when $\alpha=1$ and $\beta=0$, then by IH $\sststile{\rho}{1} a_1$ and $\sststile{\rho}{0} a_2$,
thus by the substitution lemma we have as required $\sststile{\rho}{1}(a_1 \cdot_d a_2)$.\qedhere
\end{itemize}
\end{proof}

\noindent The lemma above allows us to transform any expanded reduction tree
into a purely dynamic one (using only rules \textsc{Base} and \textsc{Red}\footnote{In fact at this point
the expanded reduction tree can still contain cuts of order $0$, eliminated later.}). Now, we show how 
an expanded reduction tree can be automatically inferred for any skeleton.

\begin{lem}[Recomposition lemma]
Let $a$ be a skeleton such that $\ord(a) \geq 1$. Let $\nodes(a)$ denote the multiset of labels of nodes in $a$. Then,
$\sststile{\ord(a) - 1}{\Pi_{n \in \nodes(a)} (n+1)} a$.
\label{lem_recomposition}
\end{lem}
\begin{proof} 
First, let us show that the following rule \textsc{Base'} is admissible, for any $\alpha$ and $\rho$.
\[
\sststile{\rho}{\alpha + n} n
\]
If $n=0$ this is exactly \textsc{Base}. Otherwise we apply \textsc{Red}. There is no possible reduction, so the only thing we have
to prove is $\sststile{\rho}{\alpha + n - 1} (n-1)[]$, which is by IH.

The lemma follows by applying Lemma \ref{main_substitution} once for each node.
\end{proof}

Now, we show how to deduce from a purely dynamic expanded reduction tree a bound to the length of possible reductions.

\begin{lem}[Bound lemma]
Let $a$ be a skeleton, then if $\sststile{0}{\alpha} a$, $\norm(a) \leq \alpha$.
\label{bound}
\end{lem}
\begin{proof} 
First of all we prove that if there is a witness for $\sststile{0}{\alpha} a$, then it can be supposed \textsc{Cut}-free. We reason by induction
on $\sststile{0}{\alpha} a$.
\begin{itemize}
\item \textsc{Base.} The rule has no premise, so the witness tree for $\sststile{0}{\alpha} a$ is already \textsc{Cut}-free.
\item \textsc{Red.} By IH, \textsc{Cut} can be eliminated in the premises of $\sststile{0}{\alpha+1} a$. Therefore by $\textsc{Red}$,
there is a \textsc{Cut}-free witness for $\sststile{0}{\alpha+1} a$.
\item \textsc{Cut.} Suppose we have $\sststile{0}{\alpha + \beta} a \cdot_0 b$ by \textsc{Cut}, whose premises are $\sststile{0}{\alpha} a$
and $\sststile{0}{\beta} b$. By IH, we can assume the witness trees for $\sststile{0}{\alpha} a$ and $\sststile{0}{\beta} b$
to be \textsc{Cut}-free. Let us form the context-skeleton $a() = a \cdot_0 x$. Then by Lemma \ref{first_substitution} we have
$\sststile{0}{\alpha} a(b) = a \cdot_0 b$, and by Lemma \ref{monotonicity} that $\sststile{0}{\alpha + \beta} a \cdot_0 b$. Since the witness
tree for $\sststile{0}{\alpha} a(\emptyset)$ is \textsc{Cut}-free, so are the witness trees for $\sststile{0}{\alpha} a \cdot_0 b$ and
$\sststile{0}{\alpha + \beta} a \cdot_0 b$.
\end{itemize}

Then, we prove the lemma by induction on the \textsc{Cut}-free witness tree for $\sststile{0}{\alpha} a$:
\begin{itemize}
\item \textsc{Base.} Necessarily, the root of $a$ is $0$, thus $\norm(a) = 0$; there is nothing to prove.
\item \textsc{Red.} The premises of $\sststile{0}{\alpha} a$ include in particular that for all $a'$ such that $a\leadsto a'$, we have $\sststile{0}{\alpha-1} a'$. By IH, this
means that for all such $a'$ we have $\norm(a') \leq \alpha-1$, hence $\norm(a) \leq \alpha$.\qedhere
\end{itemize}
\end{proof}

\noindent From all this, it is possible to give a first upper bound by using the recomposition lemma, then iterating the cut elimination lemma. However, we will first prove here a refined version of the
cut elimination lemma when $\rho=1$, which will allow to decrease by one the height of the tower of exponentials. First, we need the following adaptation of the substitution lemma:
\begin{lem}[Base substitution lemma]
If $\sststile{0}{\alpha} a(\emptyset)$, $\sststile{0}{\beta} b$ and the type of $x$ in $a$ is $1$, then $\sststile{0}{\alpha+\beta} a(b)$.
\label{specialized_substitution}
\end{lem}
\begin{proof} 
We prove by induction on derivations $\sststile{0}{\alpha} a(\emptyset)$ that the property above holds for all context-arena
$a'$ such that the type of $x$ in $a'$ is $1$ and $a = a'(\emptyset)$.
\begin{itemize}
\item \textsc{Base.} The root of $a$ is $0$, hence the result is trivial.
\item \textsc{Red.} Suppose $a'$ has the form $n[\{d_1\}a'_1, \dots, \{d_p\}a'_p, \{1\}x]$, with 
$a_i = a'_i(\emptyset)$ (the case where $x$ appears as a son of the root encompasses the other).
The premises of \textsc{Red} are then that for $1 \leq i \leq p$ such that
$d_i\geq 1$, $\sststile{0}{\alpha-1} a_i \cdot_{d_i-1} (n-1)[\{d_1\}a_1, \dots \{d_p\}a_p]$
and $\sststile{0}{\alpha-1} (n-1)[\{d_1\}a_1, \dots \{d_p\}a_p]$. The IH on these premises give witnesses
for the two following properties:
\begin{eqnarray}
\sststile{0}{\alpha-1+\beta} (a'_i \cdot_{d_i-1} (n-1)[\{d_1\}a'_1, \dots, \{d_p\}a'_p, \{1\}x])(b)\label{eq4}
\end{eqnarray}
\begin{eqnarray}
\sststile{0}{\alpha-1 +\beta} ((n-1)[\{d_1\}a'_1, \dots \{d_p\}a'_p, \{1\}x])(b)\label{eq5}
\end{eqnarray}
By hypothesis we have $\sststile{0}{\beta} b$, hence by Lemma \ref{first_substitution} we have
\begin{eqnarray}
\sststile{0}{\beta} b \cdot_{0} (n-1)[\{d_1\}a'_1(b), \dots \{d_p\}a'_p(b), \{1\}b]\label{eq6}
\end{eqnarray}
Hence, using (\ref{eq4}) for all $i\in \{1, \dots, p\}$, (\ref{eq5}) and (\ref{eq6}) (adjusted to $\sststile{0}{\alpha-1+\beta}$ 
by Lemma \ref{monotonicity}) we deduce by \textsc{Red} that
\[
\sststile{0}{\alpha + \beta} n[\{d_1\}a'_1(b), \dots, \{d_p\}a'_p(b), \{1\}b]
\]
Which is what was required.
\item \textsc{Cut.} Let us suppose $\sststile{0}{\alpha+\gamma} a$ is obtained by \textsc{Cut}, 
hence $a$ has the form $a_1 \cdot_0 a_2$. Let us suppose
that $a'$ has the form $(a'_1 \cdot_0 a'_2) \cdot_1 x$, with $a_i = a'_i(\emptyset)$, again the case
where $x$ is a child of the root encompasses the other. The premises of \textsc{Cut} are then
$\sststile{0}{\alpha} a_1$ and $\sststile{0}{\gamma} a_2$. We also
have $(a'_1 \cdot_1 x)(\emptyset) = a_1$, therefore
the IH on $\sststile{0}{\alpha} a_1$ along with $\sststile{0}{\beta} b$ implies that
$\sststile{0}{\alpha+\beta} a'_1(b) \cdot_1 b$ and all that remains is to substitute $a'_2(b)$ in $(a'_1(b) \cdot_1 b) \cdot_0 x$.
But since the type of $x$ is $0$, Lemma \ref{first_substitution} implies
that $\sststile{0}{\alpha+\beta} (a'_1(b) \cdot_1 b) \cdot_{0} a'_2(b)$, which concludes
since trees are considered up to permutation.\qedhere
\end{itemize}
\end{proof}

\begin{lem}[Base cut elimination lemma]
If $\sststile{1}{\alpha} a$, then $\sststile{0}{\alpha} a$.
\label{specialized_cutelim}
\end{lem}
\begin{proof} 
By induction on the witness tree for $\sststile{1}{\alpha} a$.
\begin{itemize}
\item \textsc{Base.} Trivial.
\item \textsc{Red.} Suppose $a$ has the form $n[\{d_1\}a_1, \dots, \{d_p\}a_p]$. The premises of \textsc{Red} are that for all $i\in \{1, \dots, p\}$ we have
$\sststile{1}{\alpha-1} a_1 \cdot_{d_i - 1} (n-1)\{d_1\}a_1, \dots, \{d_p\}a_p]$ and $\sststile{1}{\alpha-1} (n-1)\{d_1\}a_1, \dots, \{d_p\}a_p]$. The result is then
trivial by IH and \textsc{Red}.
\item \textsc{Cut.} Suppose $a = a_1 \cdot_d a_2$ with $d\leq 1$, the premises of \textsc{Cut} are that $\sststile{1}{\alpha} a_1$ and $\sststile{1}{\beta} a_2$. If $d=0$, then the result
is trivial by the IH and \textsc{Cut}. If $d=1$, we just apply Lemma \ref{specialized_substitution} instead of \textsc{Cut}.\qedhere
\end{itemize}
\end{proof}

\noindent From all of this put together, we deduce the main theorem of this section.

\begin{thm}[Upper bound]
If $\ord(a), \depthbs(a), \max(a) \geq 1$,
$\norm(a) \leq 2_{\ord(a)-1}^{\depthbs(a)\log(\max(a) + 1)}$.
\label{thm_upbs}
\end{thm}
\begin{proof} 
Let us set $a' = \tbs(\depth(a), \ord(a), \max(a))$. By definition, we have $\depth(a') = \depth(a)$, $\ord(a') = \ord(a)$ and $\max(a') = \max(a)$.
Moreover by Proposition \ref{prop_merge}, we have:
\[
\norm(a) \leq \norm(a')
\]

Write $o = \ord(a')$, $d = \depth(a')$, $m = \max(a')$. By Lemma \ref{lem_recomposition}, we have 
$\sststile{o - 1}{\Pi_{n \in \nodes(a')} (n+1)} a'$.
But all nodes in $a'$ have the same label $m$ and there are $d$ nodes, therefore $\sststile{o - 1}{(m+1)^d} a'$.
By $o-2$ applications of Lemma \ref{lem_cutelimination} it follows that $\sststile{1}{2_{o - 1}^{d\log(m+1)}} a'$,
so $\sststile{0}{2_{o - 1}^{d\log(m+1)}} a'$ as well by Lemma \ref{specialized_cutelim}. Finally,
$\norm(a) \leq \norm(a') \leq 2_{o-1}^{d\log(m+1)}$ by Lemma \ref{bound}.
\end{proof}

This upper bound is optimal, since it yields bounds on linear head reduction that we will prove optimal
in the next section.

\subsection{On game situations}
Before we conclude this section, let us mention a specialized form of our result of special importance to the previous section.

\begin{thm}
If $n, p \geq 1$ and $d\geq 3$, then $\norm(n[\{d\}p]) \leq 2_{d-3}^{2\frac{p^{n+1} - 1}{p-1} - 1}$.
\label{th:gamesituations}
\end{thm}
\begin{proof}
We show by induction on $n\geq 1$ that $\sststile{d-2}{1 + \Sigma_{k=1}^n 2p^k} n[\{d\}p]$.
For $n=1$, we need to show that $\sststile{d-2}{1 + 2p} 1[\{d\}p]$. By \textsc{Red}, this amounts to
$\sststile{d-2}{2p} 0[\{d\}p]$ (true by \textsc{Base}) and:
\[
\sststile{d-2}{2p}{p[\{d-1\}0[\{d\}p]]}
\]
This follows from Lemma \ref{main_substitution} from $\sststile{d-2}{p} p$ and $\sststile{d-2}{1} 0[\{d\}p]$, where
the former follows from $p$ instances of \textsc{Red} and one \textsc{Base} and the latter follows from \textsc{Base}.

Now, suppose $\sststile{d-2}{1+ \Sigma_{k=1}^{n-1} 2p^k} (n-1)[\{d\}p]$. The skeleton $n[\{d\}p]$ has
one possible reduction:
\[
n[\{d\}p] \leadsto p[\{d-1\} (n-1)[\{d\}p]]
\]
We know by $p$ instances of \textsc{Red} and one of \textsc{Base} that $\sststile{d-2}{p} p$. Therefore by Lemma \ref{main_substitution}:
\[
\sststile{d-2}{p(2 + \Sigma_{k=1}^{n-1} 2p^k)} p[\{d-1\} (n-1)[\{d\}p]]
\]
We also have $\sststile{d-2}{p(2 + \Sigma_{k=1}^{n-1} 2p^k)} (n-1)[\{d\}p]$
by IH and Lemma \ref{monotonicity}. Therefore by \textsc{Red}, $\sststile{d-2}{1 + \Sigma_{k=1}^n 2p^k} n[\{d\}p]$,
which concludes the induction. 

From this, we use as before $d-3$ times Lemma \ref{lem_cutelimination}, Lemma \ref{specialized_cutelim} and Lemma \ref{bound}, we deduce:
\[
\norm(n[\{d\}p]) \leq 2_{d-3}^{1 + \Sigma_{k=1}^n 2p^k}
\]
And finally $1 + \Sigma_{k=1}^n 2p^k = 2\frac{p^{n+1} - 1}{p-1} - 1$, yielding the announced result.
\end{proof}

This result is particularly relevant in \emph{game situations}: when studying the reduction length of one $\eta$-expanded
Böhm tree applied to another. It also provides the answer to the question raised in the previous section about
the possible length of bounded visible pointer structures --- and hence of interactions between bounded strategies.

\begin{rem}
We finish this section by a few remarks on the above result:
\begin{itemize}
\item For $d=2$, it is easy to see that $\norm(n[\{2\}p]) = 2n$.
\item For $d=3$, experiments with an implementation of skeletons and their reductions suggest that, for $n\geq 0$ and
$p\geq 2$, $\norm(n[\{3\}p]) = 2\frac{p^n-1}{p-1}$.
This quantity is $\Theta(2^{(n-1)\log(p)})$ whereas our general bound predicts $\Theta(2^{n\log(p)})$. They differ 
but do match up to an exponential, being both of the form $2^{\Theta(n\log(p))}$. 
\end{itemize}
In fact for any $d\geq 3$ we have $\norm(n[\{d\}p]) = 2_{d-2}^{\Theta(n\log(p))}$.
The upper bound is our theorem above, and the lower bound is provided by the reduction on skeletons corresponding
to the visible pointer structures used in the proof of Theorem \ref{th:boundvps}. So, in this sense our result is optimal on game
situations, just as the upper bound of Theorem \ref{thm_upbs} will appear later to be optimal in the general situation.
\end{rem}

\section{Skeletons and linear head reduction}
\label{sec:lhr}

Although it is generally understood that game-theoretic interaction (underlying interaction skeletons) has a strong
operational content, the game-theoretic toolbox lacks results making this formal. One notable exception is the
result of Danos, Herbelin and Regnier \cite{dhr} already mentioned, which describes a step-by-step
correspondence between the linear head reduction sequence of a game situation $M~N_1~\dots~N_n$
(where $M, N_1, \dots, N_n$ are $\beta$-normal and $\eta$-long) and the interaction
of the corresponding strategies. Along with Theorem \ref{th:gamesituations}, this connection suffices to immediately deduce
an (optimal) upper bound for the length of reduction sequences on game situations.
However, this reasoning has two drawbacks. Firstly it is rather indirect: the link it provides between linear head
reduction and interaction skeletons, two relatively simple combinatorial objects, is obfuscated by the 
variety of mathematical notions involved. Indeed this connection requires elaborate 
semantic notions such as visible strategies and pointer structures, and third party results such as the (very technical)
result of \cite{dhr}. Secondly it only covers game situations, and it is not clear how to obtain from that general
results on arbitrary terms.

In this final section we address these two points and proceed to analyse the direct connection between interaction
skeletons and syntactic reduction. This study culminates in optimal upper bounds to the length of linear head reduction
sequence on arbitrary simply-typed $\lambda$-terms. This requires us, on the one hand, to construct a generalization
of game situations whose reduction follows the combinatorics of interaction skeletons and, on the other hand, to show that
one can \emph{compile} arbitrary terms into these generalized situations in a way allowing us to obtain our upper bounds.
%
%
%
%
%

In Subsection \ref{subsec:lhr} we give the definition of linear head reduction and prove some basic properties. In
Subsection \ref{subsec:ggs} we define and study generalized game situations, and in Subsection \ref{subsec:simggs} 
we prove the technical core of this section: the fact that lhr on generalized game situations can be simulated within
interaction skeletons. Finally, Subsections \ref{subsec:bounds} and \ref{subsec:general} are devoted to dealing
respectively with $\eta$-expansion and with $\lambda$-lifting in order to compile arbitrary terms to generalized
game situations and deduce our results.

\subsection{Linear head reduction}
\label{subsec:lhr}

We start by recalling the definition of linear head reduction and proving some basic properties that
are folklore, but to our knowledge unpublished under this formulation. Our notion of linear head reduction
follows \cite{dhr}. We use
it rather than the more elegant approach of Accattoli \cite{DBLP:conf/rta/Accattoli12} because we believe it yields a more direct
relationship with games. Indeed the \emph{multiplicative} reductions of Accattoli's calculus have no counterpart in
games/skeletons, which only take into account the variable substitutions.

\subsubsection{Definition of linear head reduction}

This work focuses strongly on \emph{linear substitution}, for which only one variable occurrence is substituted at a time.
In this situation, it is convenient to have a distinguished notation for particular \emph{occurrences} of variables.
We will use the notations $x_0, x_1, \dots$ to denote
particular occurrences of the same variable $x$ in a term $M$. When in need of additional variable identifiers, we will use
$x^1, x^2, \dots$. Sometimes, we will still denote occurrences of $x$ by just $x$ when their index is not relevant.
If $x_0$ is a specific
occurrence of $x$, we will use $M[N/x_0]$ for the substitution of $x_0$ by $N$, leaving all other occurrences of $x$ unchanged.

Intuitively, lhr proceeds as follows. We first locate the head variable occurrence, \emph{i.e.} the leftmost variable
occurrence in the term $M$. Then we locate the abstraction, if any, that binds this variable. Then we locate (again if it exists) the
subterm $N$ of $M$ in argument position for that abstraction, and we substitute the head occurrence by $N$. We touch neither the other
occurrences of $x$ nor the redex. It is worth noting that locating the argument subterm can be delicate, as it is not necessarily part
of a $\beta$-redex. For instance in $(\lambda y^A.~(\lambda x^B.~x_0 M)) N_1 N_2$,
we want to replace $x_0$ by $N_2$, even though $N_2$ is not directly applied to $\lambda x^B.~x_0 M$. Therefore, the notion
of redex will be generalized. 

Note that a term is necessarily of the form $\daimon~M_1~\dots~M_n, x_0~M_1~\dots~M_n, \lambda x.~M$ or $(\lambda x.~M)~M_1~\dots~M_n$. This
will be used quite extensively to define and reason on lhr. The \textbf{length} of a term $M$ is the number of characters in $M$,
\emph{i.e.} $\length(\daimon) = 1, \length(x_0) = 1, \length(\lambda x.~M) = \length(M) + 1, \length(M_1~M_2) = \length(M_1) + \length(M_2)$.
Its \textbf{height} is $\h(\daimon) = 0, \h(x_0) = 1, \h(\lambda x.~M) = \h(M), \h(M_1~M_2) = \max(\h(M_1), \h(M_2) + 1)$.

\begin{defi}
Given a term $M$, we define its set of \textbf{prime redexes}. They are written as pairs $(\lambda x, N)$ where
$N$ is a subterm of $M$, and $\lambda x$ is used to denote the (if it exists, necessarily unique by Barendregt's convention) subterm of $M$ of the form
$\lambda x.~N'$. We define the prime redexes of $M$ by induction on its length, distinguishing several cases depending on the form of $M$.
\begin{itemize}
\item If $\daimon~M_1~\dots~M_n$ has no prime redex,
\item If $x_0~M_1~\dots~M_n$ has no prime redex,
\item The prime redexes of $\lambda x.~M'$ are those of $M'$,
\item The prime redexes of $(\lambda x.~M')~M_1~\dots~M_n$, are $(\lambda x, M_1)$ plus those of $M'~M_2~\dots~M_n$.
\end{itemize}
\end{defi}

The \textbf{head occurrence} of a term $M$ is the leftmost occurrence of a variable or constant in $M$. 
If $(\lambda x, N)$ is a prime redex of $M$ whose head occurrence is an occurrence $x_0$ of the variable $x$,
then the \textbf{linear head reduct} of $M$ is $M' = M[N/x_0]$. We write $M {\to_\lhr} M'$.

\begin{exa}
As an example, we give the lhr sequence of the term $(\lambda f.~\lambda x.~f~(f~x))~(\lambda y.~y)~\daimon$.
\begin{eqnarray*}
(\lambda f.~\lambda x.~f~(f~x))~(\lambda y.~y)~\daimon   
&\to_\lhr& (\lambda f.~\lambda x.~(\lambda z.~z)~(f~x))~(\lambda y.~y)~\daimon\\
&\to_\lhr& (\lambda f.~\lambda x.~(\lambda z.~f~x)~(f~x))~(\lambda y.~y)~\daimon\\
&\to_\lhr& (\lambda f.~\lambda x.~(\lambda z.~(\lambda u.~u)~x)~(f~x))~(\lambda y.~y)~\daimon\\
&\to_\lhr& (\lambda f.~\lambda x.~(\lambda z.~(\lambda u.~x)~x)~(f~x))~(\lambda y.~y)~\daimon\\
&\to_\lhr& (\lambda f.~\lambda x.~(\lambda z.~(\lambda u.~\daimon)~x)~(f~x))~(\lambda y.~y)~\daimon
\end{eqnarray*}
At this point the reduction stops since the head occurrence is a constant.
\label{ex_lhr}
\end{exa}

Given a term $M$, we overload the notation $\norm$ and write $\norm(M)$ for the length of the lhr sequence of $M$. 
It is straightforward to see that lhr is compatible with $\beta$-reduction,
in the sense that if $M \to_\lhr M'$ we have $M \equiv_\beta M'$.
Since redexes for lhr are not necessarily $\beta$-redexes, it will be necessary to consider the
following generalization of redexes: 

\begin{defi}[Generalized redex]
The \textbf{generalized redexes} of a term $M$ are the prime redexes of all subterms of $M$.
%
In particular, all prime redexes are generalized redexes.
\end{defi}

\begin{exa}
Consider the following $\lambda$-term:
\[
M = (\lambda x.~x)~((\lambda y.~(\lambda z.~u))~v~w)
\]
The only prime redex of $M$ is $(\lambda x, (\lambda y.~(\lambda z.~u))~v~w)$. The two other generalized
redexes are $(\lambda y, v)$, which is also a $\beta$-redex, and $(\lambda z, w)$, which is not.
\end{exa}

\subsubsection{Relating $\lhr$ with $\beta$-reduction}
\label{sec:betalhr}
Before we go on to relating $\lhr$ with $\beta$-reduction, let us include some preliminary investigations on its basic
properties, and in particular how it relates to $\beta$-reduction.
Consider $\lambda$-terms temporarily extended by a term $\ensquare{M}$, for each term $M$.
Those boxes can be opened when in head position.
Formally, a head context is:
\[
H[] = [] \mid H[]~N \mid \lambda x.~H[]
\]
The additional reduction rule for opening boxes is then (where $H$ is a head context):
\[
H[\ensquare{M}] \leadsto_{pop} H[M]
\]

We extend lhr in the presence of $\square$ by defining the prime redexes of $M$ by induction on its length,
as follows:
\begin{itemize}
\item $\daimon~M_1~\dots~M_n$, $x_0~M_1~\dots~M_n$ and $\ensquare{M'}~M_1~\dots~M_n$ have no prime redex,
\item The prime redexes of $\lambda x.~M'$ are those of $M'$,
\item The prime redexes of $(\lambda x.~M')~M_1~\dots~M_n$ are $(\lambda x, M_1)$ plus those of $M'~M_2~\dots~M_n$.
\end{itemize}
Writing $\square^{-1}(M)$ for the term obtained from $M$ by removing all boxes, we see immediately that the boxes can
only block prime redexes, so prime redexes of $M$ are always in $\square^{-1}(M)$.

Writing $M\to_{\square^{-1}} \square^{-1}(M)$, we also notes that within this calculus, $\beta$-reduction on standard terms
decomposes in two steps:
\begin{enumerate}[label=\arabic*.]
\item[1.] \emph{Substitution ($\beta\square$)} $(\lambda x.~M) N \leadsto_{\beta\square} M[\ensquare{N}/x]$
\item[2.] \emph{Unboxing. ($\square^{-1}$)} In one step, we remove all boxes in $M$.
\end{enumerate}
Intuitively, boxes introduce a \emph{delay} in the processing of $\beta$-reduction, and this delay corresponds exactly
to the delay of lhr with respect to head $\beta$-reduction. Here, we exploit this fact to compare the
two of them.
In this Section \ref{sec:betalhr} (but only in this section), terms are in this extended language unless otherwise specified.
If $M$ does not contain any boxes, we say it is a \textbf{standard} $\lambda$-term.

\begin{lem}
If $M, M'$ are standard terms such that $M \to_{\lhr} M'$, $N$ is a standard term, and the head occurrence $x_0$ of $M$ is
not an occurrence of $y$, then $M[\ensquare{N}/y] \to_{\lhr} M'[\ensquare{N}/y]$.
\label{lem:lhrsubstsq}
\end{lem}
\begin{proof}
Straightforward by induction on the length of $M$.
\end{proof}

\begin{lem}
Suppose $M$ is a standard term with a head occurrence $x_0$ of a variable $x$, and that $M \to_{\beta\square} M'$.
Then for any standard term $N$ such that no free variable of $N$ is bound in $M$, we also have
$M[N/x_0] \to_{\beta\square} M'[N/x_0]$.
\label{lem:betasquaresubst}
\end{lem}
\begin{proof}
Straightforward by induction on $M$. 
\end{proof}

Write $\to_{\lhr \vee pop}$ for exactly one reduction step of, either $\lhr$, or $pop$.

\begin{lem}
If $M$ is a standard $\lambda$-term with
$M \to_{\beta\square} N$ and $M \to_{\lhr} M'$, then there is $N \to_{\lhr\vee pop} N'$ such that
$M' \to_{\beta\square}^* N'$. Moreover, if $M \to_{\beta\square} N$ is not a weak head reduction (so if it not a head reduction or
operates under a lambda), none of the $M' \to_{\beta\square}^* N'$ is a weak head reduction.
\label{lem:besquare_lhr}
\end{lem}
\begin{proof}
By induction on the length of $M$, detailing only the case $M = (\lambda x.~M')~M_1~\dots~M_n$.
\begin{itemize}
\item If $M = (\lambda x.~M')~M_1~\dots~M_n$ where the head occurrence $x_0$ is an occurrence of $x$ and the $\beta\square$-reduction is:
\[
(\lambda x.~M')~M_1~\dots~M_n \to_{\beta\square} M'[\ensquare{M_1}/x]~M_2~\dots~M_n
\]
In $M'[\ensquare{M_1}/x]~M_2~\dots~M_n$, one copy of $\ensquare{M_1}$ has replaced the occurrence $x_0$, so there is $\ensquare{M_1}$ in
head position. It follows that:
\[
M'[\ensquare{M_1}/x]~M_2~\dots~M_n \to_{pop} M'[M_1/x_0][\ensquare{M_1}/x]~M_2~\dots~M_n
\]
But we also have
$(\lambda x.~M'[M_1/x_0])~M_1~\dots~M_n \to_{\beta\square} M'[M_1/x_0][\ensquare{M_1}/x]~M_2~\dots~M_n$,
so we set $N' = M'[M_1/x_0][\ensquare{M_1}/x]~M_2~\dots~M_n$.
\item If $M = (\lambda x.~M')~M_1~\dots~M_n$ whose head occurrence $x_0$ is an occurrence of $x$, and the $\beta\square$-reduction
is inside $M'$. Then no free variable in $M_1$ can be bound in $\lambda x.~M'$ by Barendregt's convention. Therefore,
by Lemma \ref{lem:betasquaresubst}, we have
\[
M'[M_1/x_0] \to_{\beta\square} M''[M_1/x_0]
\]
therefore $(\lambda x.~M'[M_1/x_0])~M_1~\dots~M_n \to_{\beta\square} (\lambda x.~M''[M_1/x_0])~M_1~\dots~M_n$.
\item If $M = (\lambda x.~M')~M_1~\dots~M_n$ where the head occurrence $x_0$ is an occurrence of $x$, and the $\beta\square$-reduction is
$M_1 \to_{\beta\square} M'_1$. Then we have:
\[
(\lambda x.~M'[M_1/x_0])~M_1~\dots~M_n \to_{\beta\square^2} (\lambda x.~M'[M'_1/x_0])~M'_1~\dots~M_n
\]
providing the necessary commutation.
\item If $M = (\lambda x.~M')~M_1~\dots~M_n$ where the head occurrence $x_0$ is an occurrence of $x$, and the $\beta\square$-reduction
is $M_i \to_{\beta\square} M'_i$ with $i \geq 2$, then the commutation is trivial.
\item If $M = (\lambda y.~M')~M_1~\dots~M_n$ where the head occurrence $x_0$ is not an occurrence of $y$, and the $\beta\square$-reduction is:
\[
(\lambda y.~M')~M_1~\dots~M_n \to_{\beta\square} M'[\ensquare{M_1}/y]~M_2~\dots~M_n
\]
Necessarily, we have
$M'~M_2~\dots~M_n \to_{\lhr} M''~M_2~\dots~M_n$.
By Lemma \ref{lem:lhrsubstsq}, it follows that
$(M'~M_2~\dots~M_n)[\ensquare{M_1}/y] \to_{\lhr} (M''~M_2~\dots~M_n)[\ensquare{M_1}/y]$.
Since $y$ can only appear in $M'$ and $M''$ by Barendregt's convention, it follows that:
\[
M'[\ensquare{M_1}/y]~M_2~\dots~M_n \to_{\lhr} M''[\ensquare{M_1}/y]~M_2~\dots~M_n
\]
which provides the required commutation.
\item If $M = (\lambda y.~M')~M_1~\dots~M_n$ where the head occurrence $x_0$ is not an occurrence of $y$, and the $\beta\square$-reduction
is $M' \to_{\beta\square} M''$, it follows directly from IH.
\item If $M = (\lambda y.~M')~M_1~\dots~M_n$ where the head occurrence $x_0$ is not an occurrence of $y$, and the $\beta\square$-reduction
is $M_1 \to_{\beta\square} M'_1$, trivial.
\item If $M = (\lambda y.~M')~M_1~\dots~M_n$ where the head occurrence $x_0$ is not an occurrence of $y$, and the $\beta\square$-reduction
is $M_i \to_{\beta\square} M'_i$ with $i\geq 2$. By definition of $\lhr$:
\[
M'~M_2~\dots~M_n \to_{\lhr} M''~M_2~\dots~M_n
\]
but we also have $M'~M_2~\dots~M_n \to_{\beta\square} M'~M_2~\dots~M'_i~\dots~M_n$. By IH, there is
$
M'~M_2~\dots~M'_i~\dots~M_n \to_{\lhr} N
$
such that $M''~M_2~\dots~M_n \to_{\beta\square}^* N$. But necessarily, $N$ must have the form $S~M_2~\dots~M'_i~\dots~M_n$
since we have $M'~M_2~\dots~M'_i~\dots~M_n \to_{\lhr} N$. So, $M''~M_2~\dots~M_n \to_{\beta\square}^* S~M_2~\dots~M'_i~\dots~M_n$.
But the original $\beta\square$-reduction was not a weak head reduction, so by IH none of the
reductions in $\to_{\beta\square}^*$ is. So they must all either operate inside of $M''$ or $M_i$: in other words
$M'' \to_{\beta\square}^* S$ and $M_i \to_{\beta\square}^* M'_i$. It follows that:
$(\lambda y.~M'')~M_1~\dots~M_n \to_{\beta\square}^* (\lambda y.~S)~M_1~\dots~M'_i~\dots~M_n$,
providing the required commutation.\qedhere
\end{itemize}
\end{proof}

\begin{lem}
Suppose $M \to_{\lhr \vee pop} N$, $M \to_{\square^{-1}} M'$ and $N \to_{\square^{-1}} N'$. Then,
\begin{itemize}
\item If $M \to_{pop} N$, then $M' = N'$.
\item If $M \to_{\lhr} N$, then $M' \to_{\lhr} N'$.
\end{itemize}
\label{lem:unboxing}
\end{lem}
\begin{proof}
Straightforward.
\end{proof}

\begin{prop}
Linear head reduction terminates on standard terms.
\label{prop:lhr_term}
\end{prop}
\begin{proof}
Note first that from Lemma \ref{lem:besquare_lhr} it follows that whenever $M\to_{\beta\square}^* N$ and
$M \to_{\lhr} M'$, then there exists $N'$ such that $N \to_{\lhr\vee pop} N'$ and $M' \to_{\beta\square}^* N'$.

Suppose standard $M$ has an infinite reduction chain:
\[
M = M_0 \to_{\lhr} M_1 \to_{\lhr} M_2 \to_{\lhr} \dots
\]
and $M\to_{\beta\square} M'$. Then, by iterating the argument above, we get an infinite chain:
\[
M' = M'_0 \to_{\lhr \vee pop} M'_1 \to_{\lhr \vee pop} M'_2 \to_{\lhr \vee pop} \dots
\]
Note that in this chain, there cannot be an infinite succession of $pop$, because each $pop$ strictly decreases
the length of terms. Therefore, there are an infinite number of $\to_{\lhr}$ in this sequence. Now, define
$M' \to_{\square^{-1}} M''$. Likewise for each $i\in \mathbb{N}$, define $M'_i \to_{\square^{-1}} M''_i$. By Lemma \ref{lem:unboxing},
if $M'_i \to_{pop} M'_{i+1}$ we have $M''_i = M''_{i+1}$, whereas if $M'_i \to_{\lhr} M'_{i+1}$, we
still have $M''_i \to_{\lhr} M''_{i+1}$; it follows that there is a chain:
\[
M'' = M''_0 \to_{\lhr \vee =} M''_1 \to_{\lhr \vee =} M''_2 \to_{\lhr \vee =} \dots
\]
Since there is an infinite number of $\lhr$ in the sequence of $M'$, there is an infinite number of $\lhr$ in this sequence. Moreover
since $M\to_{\beta\square} M' \to_{\square^{-1}} M''$, we have $M\to_{\beta} M''$. Therefore, we have proved that if the $\lhr$ reduction
chain of $M$ is infinite and $M\to_{\beta} M'$, then the $\lhr$ reduction chain of $M'$ is infinite as well. But by normalization
of $\beta$-reduction, we know that $M \to_{\beta}^* K$, where $K$ is $\beta$-normal. But if $K$ has no $\beta$-redex it has no prime
redexes either, so $K$ has no $\lhr$ reduction; absurd. Therefore, the $\lhr$ reduction chain of $M$ was finite.
\end{proof}

\begin{lem}
Let $M$ be a standard term with a head occurrence $x_0$, such that there is no prime redex $(\lambda x, N)$ within $M$.
Let $S$ be another term, and $y$ a free variable in $M$. Then, there is still no prime redex $(\lambda x, N)$ in $M[S/y]$.
\label{lem:subst_pr}
\end{lem}
\begin{proof}
Straightforward by induction on the length of $M$. 
\end{proof}

\begin{lem}
If $M$ is standard, normal for $\lhr$ and $M\to_{\beta} M'$, then $M'$ is normal for $\lhr$.
\label{lem:beta_pres_lhrnorm}
\end{lem}
\begin{proof}
By induction on the length of $M$. We only detail $M = (\lambda x.~M')~M_1~\dots~M_n$, the other cases being trivial.
Several subcases arise depending on the location of the $\beta$-reduction; the only non-trivial case
is $M \to_{\beta} M'[M_1/x]~M_2~\dots~M_n$. If the head occurrence of $M$ is a constant $\daimon$,
it is still the head occurrence of $M'[M_1/x]~M_2~\dots~M_n$, which is therefore $\lhr$-normal. Otherwise it is a variable
occurrence $y_0$, such that there is no prime redex $(\lambda y, N)$ in $M$. It follows by definition of prime redexes
that there is no prime redex either in $M'~M_2~\dots~M_n$. By Lemma \ref{lem:subst_pr}, there is still no prime
redex in $(M'~M_2~\dots~M_n)[M_1/x] = M'[M_1/x]~M_2~\dots~M_n$, but its head occurrence has not changed, so it is $\lhr$-normal.
\end{proof}

\begin{prop}
If a standard $M \to_{\beta} M'$, then $\norm(M) \geq \norm(M')$.
\label{prop:betanorm}
\end{prop}
\begin{proof}
We decompose $M\to_{\beta\square} M'' \to_{\square^{-1}} M'$. Take a maximal $\lhr$ sequence:
\[
M = M_0 \to_{\lhr} M_1 \to_{\lhr} \dots \to_{\lhr} M_n
\]
by Lemma \ref{lem:besquare_lhr}, there is a corresponding sequence of the same length:
\[
M'' = M''_0 \to_{\lhr \vee pop} M''_1 \to_{\lhr \vee pop} \dots \to_{\lhr \vee pop} M''_n
\]
with for all $i\in \{0, \dots, n\}$, $M_i \to_{\beta\square}^* M''_i$. Since $M_n$ is maximal, this means that
its head occurrence does not appear in a prime redex, or is a constant. By unboxing, we get:
\[
M' = M'_0 \to_{\lhr \vee =} M'_1 \to_{\lhr \vee =} \dots \to_{\lhr \vee =} M'_n
\]
We have $M_n \to_{\beta}^* M'_n$ by construction, so by Lemma \ref{lem:beta_pres_lhrnorm} we know that $M'_n$ is $\lhr$-normal as well,
so this chain is maximal.
So, the $\lhr$ sequence of $M'$ is the subsequence of
$M'_0 \to_{\lhr \vee =} M'_1 \to_{\lhr \vee =} \dots \to_{\lhr \vee =} M'_n$ keeping only the $\lhr$ steps, so
its length is less or equal than $n$.
\end{proof}

\subsection{Generalized game situations}
\label{subsec:ggs}

In this section, we detail the situation where the connection between lhr and skeletons is the most direct. 

That involves two aspects. Firstly, 
since the game-theoretic interpretation of terms is invariant under $\eta$-expansion, it is clear from the start that terms
whose dynamics match game-theoretic interaction should be $\eta$-expanded. Accordingly, we will give a new notion of non
$\beta$-normal $\eta$-long terms adapted to lhr. Secondly, game-theoretic interaction is also \emph{local}: strategies
involved in the interpretation
of a term will only communicate with other strategies with whom they share a redex, whereas substitution is non-local
and can span over multiple redexes. Accordingly, we introduce a notion called \emph{local scope} ensuring that the 
information flow of lhr is local. With the help of these two conditions we define \emph{generalized game situations},
that we prove to elegantly connect to interaction skeletons.
 
We now investigate these two aspects, in turn.

\subsubsection{Generalized $\eta$-long terms}

Non $\beta$-normal $\eta$-long terms are often defined as those for which any further $\eta$-expansion creates new $\beta$-redexes. Here of
course, since we work with generalized $\beta$-redexes, we will use instead the following definition:

\begin{defi}
A term $M$ is \textbf{$\eta$-long} if whenever $M \to_\eta M'$, then $M'$ has more generalized $\beta$-redexes than $M$.
\end{defi}

We now prove that $\eta$-long terms are preserved by lhr. For that, we will make use of the following lemmas
allowing us to compose and decompose $\eta$-long terms.

\begin{lem}
If $M$ is $\eta$-long, so are its subterms.
\label{lem:etalong_subterm}
\end{lem}
\begin{proof}
Straightforward by induction on $M$, since an $\eta$-expansion in a subterm $N$ of $M$ would
create a new generalized redex in $M$ that necessarily has to be in $N$.
\end{proof}

\begin{lem}
A term $(\lambda x.~M)~M_1~\dots~M_n$ is $\eta$-long, iff so are $M~M_2~\dots~M_n$ and $M_1$.
\label{lem:etalong_decrec}
\end{lem}
\begin{proof}
Straightforward by cases on the location of an $\eta$-expansion.
\end{proof}

\begin{lem}
If $\Gamma \vdash M : A \to B$ and $\Gamma \vdash N:A$ are $\eta$-long, so is $M~N$.
\label{lem:etalong_app}
\end{lem}
\begin{proof}
Straightforward by induction on the length of $M$.
\end{proof}

Using the lemmas above we prove the following substitution lemma, crucial in proving the stability of
$\eta$-long terms by lhr.

\begin{lem}
If $\Gamma \vdash M: A$ is $\eta$-long and with head occurrence $x_0$ of a variable $x:B$, and $\Gamma\vdash N : B$ is $\eta$-long, then
$M[N/x_0]$ is $\eta$-long.
\label{lem:etalong_subst}
\end{lem}
\begin{proof}
By induction on the length of $M$.
\begin{itemize}
\item If $M = x_0~M_1~\dots~M_n$, then necessarily each $M_i$ is $\eta$-long by Lemma \ref{lem:etalong_subterm}. 
It follows by Lemma \ref{lem:etalong_app} that $N~M_1~\dots~M_n$ is $\eta$-long as well.
\item If $M = \lambda y.~M'$, it follows directly from the IH.
\item If $M = (\lambda y.~M')~M_1~\dots~M_n$, it follows directly from IH and Lemma \ref{lem:etalong_decrec}. \qedhere
%
\end{itemize}
\end{proof}

\begin{lem}
If $M$ is $\eta$-long and $M\to_{\lhr} M'$, then $M'$ is $\eta$-long.
\label{lem:lhr_pres_etalong}
\end{lem}
\begin{proof}
By induction on the length of $M$. We skip all the trivial cases, and only detail $M = (\lambda x.~M')~M_1~\dots~M_n$.
Then, two subcases:
\begin{itemize}
\item If the head occurrence of $M$ is an occurrence $x_0$ of $x$, then we have:
$
M \to_{\lhr} M[M_1/x_0]
$.
But since $M$ is $\eta$-long, $M_1$ is $\eta$-long as well by Lemma \ref{lem:etalong_subterm}.
So, by Lemma \ref{lem:etalong_subst} we have that $M[M_1/x_0]$ is $\eta$-long as well.
\item If the head occurrence of $M$ is an occurrence $y_0$ of a variable or constant other than $x$, then if
$M \to_{\lhr} M[N/y_0]$, for some subterm $N$ of $M$, and we must have as well:
\[
M'~M_2~\dots~M_n \to_{\lhr} M'[N/y_0]~M_2~\dots~M_n
\]
by IH, $M'[N/y_0]~M_2~\dots~M_n$ is $\eta$-long. Finally, it follows from Lemma \ref{lem:etalong_decrec} that
$M_1$ is $\eta$-long and that recomposing $(\lambda x.~M'[N/y_0])~M_1~\dots~M_n$ yields an $\eta$-long term. \qedhere
\end{itemize}
\end{proof}

\subsubsection{Local scope}
Now that we have a notion of $\eta$-long term stable under composition, let us consider the syntactic counterpart of the second aspect
of skeletons: that their reduction is local. It is not clear at first what \emph{local} means in this context: just like
a game situation consists in two $\eta$-long normal forms interacting, a generalized game situation will consist in a ``tree'' of $\eta$-long
normal forms. Let us start with this slightly naive definition:

\begin{defi}
A term $M$ is \textbf{strongly locally scoped} (abbreviated sls) iff for any generalized redex $(\lambda x, N)$ in $M$, $N$ is closed.
\end{defi}

Unfortunately, sls terms do not quite fit for a syntactic counterpart of skeletons: they are not
preserved by lhr. It is easy to find a counter-example, for instance:
\[
(\lambda y.~(\lambda x.~x~y)~(\lambda z.~z))~\daimon \to_\lhr (\lambda y.~(\lambda x.~(\lambda z.~z)~y)~(\lambda z.~z))~\daimon
\]
Here, a new generalized redex $(\lambda z, y)$ is formed where $y$ is obviously not closed. Therefore, we must make skeletons 
correspond instead
with a generalization of sls terms preserved by lhr. This generalization comes from the observation that
in the right hand side term above, the violation of strong local scope is mitigated by the fact that the violating variable $y$ is part
of a generalized redex --- so its value is somehow already provided by an environment. Hence the following definition:

\begin{defi}
A variable $x$ in $M$ is \textbf{active} iff it is a free variable or if there is a generalized redex $(\lambda x, N)$ in $M$. It is \textbf{passive}
otherwise.
A term $M$ is \textbf{locally scoped} (abbreviated ls) if for any generalized redex $(\lambda x, N)$ in $M$ all the free variables in $N$ are active in $M$. 
\end{defi}

Local scope will be sufficient to ensure that the interpretation to skeletons is a simulation,
but the correspondence between terms and skeletons will be tighter for sls terms: the tree structure
of the skeleton will match the tree structure of nested generalized redexes.

Of course, we now need to prove that locally scoped terms are preserved by lhr. However, this is still not true!
Indeed, consider the following reduction:
\[
\lambda y.~(\lambda x.~x~y)~(\lambda z.~z) \to_\lhr \lambda y.~(\lambda x.~(\lambda z.~z)~y)~(\lambda z.~z)
\]
The left hand side term is (strongly) locally scoped, but the right hand side term is not because $y$ is passive, but appears
in the argument of a generalized redex. However, the problem disappears if we apply the two terms above to a constant
$\daimon$. In general, we will show that \emph{closed} locally scoped terms of \emph{ground type} are closed under lhr.
This may seem like a big restriction but it is not: an arbitrary term can be made closed and of ground type without changing
its possible reduction sequences significantly, by replacing its free variables with constants and applying it to as many constants as
required.

To prove stability of local scope by lhr, we start by stability under substitution.

\begin{lem}
If $M$ is a locally scoped term of ground type with head occurrence $x_0$ of a variable $x$,
and $N$ is a locally scoped term, then $M[N/x_0]$ is locally scoped.
\label{lem:locscope_subst}
\end{lem}
\begin{proof}
By induction on the length of $M$ (of ground type, so not an abstraction).
\begin{itemize}
\item If $M = x_0~M_1~\dots~M_n$, then the generalized redexes of $N~M_1~\dots~M_n$ are those of $N, M_1, \dots, M_n$ (that are
of the form $(\lambda y, S)$ where all free variables in $S$ are active, by definition of local scope), with possibly
the addition of generalized redexes of the form $(\lambda z, M_i)$. Free variables in $M_i$ are free in $N~M_1~\dots~M_n$,
so they are active.
\item If $M = (\lambda y.~M')~M_1~\dots~M_n$, it follows directly from IH. \qedhere
\end{itemize}
\end{proof}

\begin{lem}
If $\vdash M: o$ is locally scoped and $M \to_{\lhr} M'$, then $M'$ is locally scoped.
\label{lem:pres_loc_scope}
\end{lem}
\begin{proof}
By induction on the length of $M$, writing $(\lambda x, S)$ for the prime redex fired in $M \to_{\lhr} M'$. We only detail the
non-trivial cases.
\begin{itemize}
\item If $M = (\lambda y.~N)~M_1~\dots~M_n$ with $y\neq x$, then it direct that
$N~M_2~\dots~M_n$ is still ls. By IH, $N[S/x_0]~M_2~\dots~M_n$ is ls as well, and from that follows
that $M[S/x_0]$ is ls.
\item If $M = (\lambda x.~N)~M_1~\dots~M_n$, where $x_0$ is an occurrence of $x$. Then, necessarily
$N~M_2~\dots~M_n$ is ls as well. Likewise, $M_1$ is ls, otherwise $M$ could not be. Therefore by
Lemma \ref{lem:locscope_subst}, $N[M_1/x_0]~M_2~\dots~M_n$ is ls, and from that follows that $M'$ is ls.\qedhere
\end{itemize}
\end{proof}

\noindent We say that a term $M$ is a \textbf{generalized game situation} if it is closed, of ground type, and both $\eta$-long and locally scoped. 
By Lemmas \ref{lem:lhr_pres_etalong} and \ref{lem:pres_loc_scope}, we know that generalized game situations are preserved by linear
head reduction.

\subsection{Simulation of generalized game situations} 
\label{subsec:simggs}
We start this subsection by showing how one can associate
a skeleton with any generalized game situation -- in fact with any term, although this connection will only yield
a simulation for generalized game situations.

\begin{defi}
Let $\Gamma \vdash M:A$ be a term, with a \textbf{bs-environment} $\rho$, being defined as a partial function mapping each
variable $x$ of $\Gamma$ on which it is defined to a skeleton $\rho(x)$. Then the skeleton $\intr{M}_{\rho}$ is defined
by induction on the length of $M$, as follows:
\[
\begin{array}{rclcl}
\intr{\daimon~M_1~\dots~M_n}_\rho &=& 0\\
\intr{x_0~M_1~\dots~M_n}_\rho &=&  1 + \bigsqcup_{i=1}^n \intr{M_i}_{\rho} && \txt{\raisebox{8pt}{if $\rho(x)$ undefined}}\\
\intr{x_0~M_1~\dots~M_n}_{\rho} &=& (1 + \bigsqcup_{i=1}^n \intr{M_i}_{\rho})\cdot_{\lv(x)+1} \rho(x) && \raisebox{2.5pt}{\txt{if $\rho(x)$ defined}}\\
\intr{\lambda x.~M}_{\rho} &=& \intr{M}_{\rho}\\
\intr{(\lambda x.~M)~M_1~\dots~M_n}_\rho &=& \intr{M~M_2~\dots~M_n}_{\rho \cup \{x\mapsto \intr{M_1}_\rho\}}
\end{array}
\]
We write $\intr{M}$ for $\intr{M}_\emptyset$.
\label{def:intr}
\end{defi}

\subsubsection{Simulation of generalized game situations}

We now prove that $\intr{-}$ is a simulation.
Because $\intr{-}$ over-approximates terms it will not directly relate
$\to_\lhr$ and $\leadsto$. Rather, we will have that if $M\leadsto M'$, then there is $a$ such that 
$\intr{M} \leadsto a \hookleftarrow \intr{M'}$. This relaxed simulation will suffice for our purposes
since by Lemma \ref{lem:embedsbs} it implies that $\norm(\intr{M'}) \leq \norm(a)$.
We show in Figure \ref{fig:exsimulation} the skeletons corresponding with all $\lhr$-reducts of the
term $(\lambda f^{o\to o}.~\lambda x^o. f~(f~x))~(\lambda y^o.~y)~\daimon_o$ from Example \ref{ex_lhr},
with explicit typing.

\begin{figure}
\begin{center}
\includegraphics[scale=0.25]{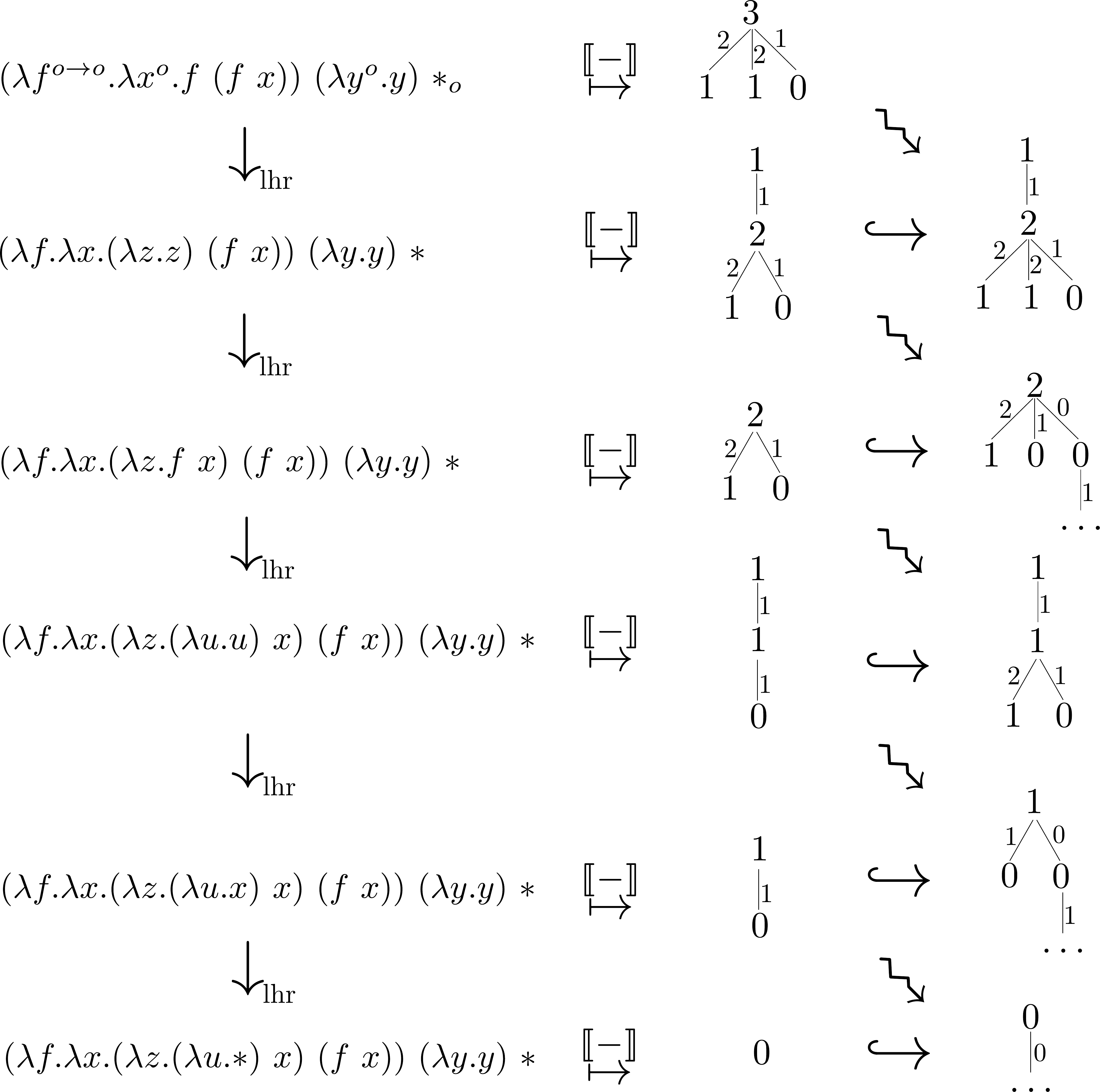}
\end{center}
\caption{Simulation of the reduction of Example \ref{ex_lhr} in skeletons}
\label{fig:exsimulation}
\end{figure}

We now aim to prove our simulation result for generalized game situations.

\begin{defi}
Let $M$ be a term, and $\rho, \rho'$ two bs-environments. We write $\rho \sim_M \rho'$
iff $\dom(\rho)\cap \fv(M) = \dom(\rho')\cap \fv(M)$, and that for all $x\in \dom(\rho)\cap \fv(M)$,
$\rho(x) = \rho'(x)$.
\end{defi}

\begin{lem}
If $M$ is a term, $\rho$ and $\rho'$ are such that $\rho \sim_M \rho'$, then
$\intr{M}_{\rho'} = \intr{M}_{\rho}$.
\label{lem:simeq}
\end{lem}
\begin{proof}
Straightforward by induction on the length of $M$.
\end{proof}

\begin{lem}
Suppose $\Gamma \vdash M:A$ is such that $(x:B)\in \Gamma$ does not appear in any subterm $N$, for any generalized redex $(\lambda y, N)$ in $M$.
Let $\rho$ be a bs-environment, and a skeleton $a$. Then, we have
$
\intr{M}_{\rho \cup \{x \mapsto a\}} \embeds \intr{M}_\rho \cdot_{\lv(x)+1} a
$.
\label{lem:seplc}
\end{lem}
\begin{proof}
By induction on the length of $M$, with $\rho' = \rho \cup \{x \mapsto a\}$. Most cases follow directly
by unfolding the definition of $\intr{-}$, IH, Lemma \ref{lem:op_embeddings} and basic manipulations on embeddings. 

We detail the one non-trivial case where $M = (\lambda y.~M')~M_1~\dots~M_n$. Then:
\begin{eqnarray*}
\intr{(\lambda y.~M')~M_1~\dots~M_n}_{\rho \cup \{x \mapsto a\}} &=& \intr{M'~M_2~\dots~M_n}_{\rho \cup \{x\mapsto a\} \cup \{y \mapsto \intr{M_1}_{\rho \cup \{x \mapsto a\}}\}}\\
&=&\intr{M'~M_2~\dots~M_n}_{\rho \cup \{x\mapsto a\} \cup \{y \mapsto \intr{M_1}_{\rho}\}}\\
&\embeds& \intr{M'~M_2~\dots~M_n}_{\rho\cup \{y \mapsto \intr{M_1}_{\rho}\}} \cdot_{\lv(x)+1} a\\
&=& \intr{(\lambda x.~M')~M_1~\dots~M_n}_{\rho} \cdot_{\lv(x)+1} a
\end{eqnarray*}
where the second equation exploits Lemma \ref{lem:simeq} and the fact that $x$ does not appear in $M_1$ (since $M_1$ is in a generalized
redex $(\lambda y, M_1)$); the embedding is by IH, and the last equality is by definition.
\end{proof}

\begin{lem}
Let $\Gamma \vdash M:A_1 \to \dots \to A_n \to o$ be a locally scoped $\eta$-long term, and
$\Gamma \vdash N_i : A_i$ be terms. Finally, let $\rho$ be a bs-environment on $M~N_1~\dots~N_n$. Then:
\[
\intr{M~N_1~\dots~N_n}_{\rho} \embeds \intr{M}_\rho \cdot [\{\lv(N_i)+1\} \intr{N_i}_{\rho} \mid 1 \leq i \leq n]
\]
\label{lem:appbs}
\end{lem}
\begin{proof}
By induction on the length of $M$, detailing non-trivial cases.
\begin{itemize}
\item If $M = \lambda x.~M'$, we calculate:
\begin{eqnarray*}
\intr{(\lambda x.~M')~N_1~\dots~N_p}_{\rho}     &=& \intr{M'~N_2~\dots~N_p}_{\rho \cup \{x \mapsto \intr{N_1}_{\rho}\}}\\
                                                &\embeds& \intr{M'}_{\rho \cup \{x \mapsto \intr{N_1}_{\rho}\}} \cdot [\{\lv(N_i)+1\} \intr{N_i}_{\rho \cup \{x \mapsto \intr{N_2}_{\rho}\}}\mid 2\leq i \leq p]\\
&=& \intr{M'}_{\rho \cup \{x \mapsto \intr{N_1}_{\rho}\}} \cdot [\{\lv(N_i)+1\} \intr{N_i}_{\rho}\mid 2 \leq i \leq p]\\
&\embeds& \intr{M'}_{\rho} \cdot_{\lv(x)+1} \intr{N_1}_{\rho} \cdot [\{\lv(N_i)+1\} \intr{N_i}_{\rho} \mid 2\leq i \leq p]\\
&=& \intr{M'}_{\rho} \cdot [\{\lv(N_i)+1\} \intr{N_i}_{\rho} \mid 1\leq i \leq p]
\end{eqnarray*}
The first equality is by definition, the first embedding by IH, the second equality by Lemma \ref{lem:simeq}
since $x$ cannot appear in $N_i$, the second embedding by Lemma \ref{lem:seplc} since $\lambda x.~M'$ is ls,
so $x$, being passive, cannot appear in $N$ for any generalized redex $(\lambda z, N)$ of $M'$. Finally, the last equation is by
definition.
\item If $M = (\lambda x.~M')~M_1~\dots~M_p$, we calculate:
\begin{eqnarray*}
&&\intr{(\lambda x.~M')~M_1~\dots~M_p~N_1~\dots~N_n}_{\rho}\\
&=& \intr{M'~M_2~\dots~M_p~N_1~\dots~N_n}_{\rho \cup \{x \mapsto \intr{M_1}_{\rho}\}}\\
&\embeds& \intr{M'~M_2~\dots~M_p}_{\rho \cup \{x \mapsto \intr{M_1}_{\rho}\}} \cdot \\&&[\{\lv(N_i)+1\}\intr{N_i}_{\rho \cup \{x \mapsto \intr{M_1}_{\rho}\}}\mid 1 \leq i \leq p]\\
&=& \intr{(\lambda x.~M')~M_1~\dots~M_p}_{\rho} \cdot [\{\lv(N_i)+1\}\intr{N_i}_{\rho}\mid 1 \leq i \leq p]
\end{eqnarray*}
The first equality is by definition, the embedding by IH, the last equation by Definition \ref{def:intr} and
Lemma \ref{lem:simeq}, since $x$ does not appear in $N_i$.\qedhere
\end{itemize}
\end{proof}

\noindent In order to prove our simulation result between lhr on terms and skeletons, we will have to generalize
it to a connection between skeletons and terms-with-environments, or closures. Therefore, we will make use of the
following notions.

\begin{defi}
The notions of \emph{closure} and \emph{environments} are defined by mutual induction:
\begin{itemize}
\item A \textbf{closure} is $(\Gamma \vdash M : A, \sigma)$ with $\Gamma \vdash M : A$ a term, and
$\sigma$ an \textbf{environment} on $\Gamma$.
\item An \textbf{environment} $\sigma$ on $\Gamma$ is a partial function which, when defined, maps a variable $(x: A)$
to a closure $(\Gamma \vdash N : A, \tau)$. We write $\sigma_1(x) = N$ and $\sigma_2(x) = \tau$.
\end{itemize}
When $(\Gamma \vdash M : A, \sigma)$ is a closure, we will use the shorthand notation $M^\sigma$ and leave $\Gamma, A$ implicit.
An environment $\sigma$ for $\Gamma$ is \textbf{flat} if for each $(x:A)\in \Gamma$,
we have that $\sigma_2(x) \subseteq \sigma$.

Standard notions on terms directly translate to environments: $\sigma$ on $\Gamma$ is \textbf{$\eta$-long} iff for all $(x:A)\in \Gamma$,
if $\sigma(x)$ is defined then $\sigma_1(x)$ is $\eta$-long and $\sigma_2(x)$ is $\eta$-long. It is \textbf{locally scoped} iff
for all $(x:A)\in \Gamma$, if $\sigma(x)$ is defined then $\sigma_1(x)$ and $\sigma_2(x)$ are locally scoped.
\end{defi}

Using environments, we can generalize lhr on terms to a linear reduction on closures.

\begin{lem}
Define the following reduction:
\begin{mathpar}
\inferrule
        { }
        {x_i~M_1~\dots~M_n \to_{\sigma} \sigma_1(x)~M_1~\dots~M_n}
\and
\inferrule
        {M \to_{\sigma} M'}
        {\lambda y.~M \to_{\sigma} \lambda y.~M'}
\and
\inferrule
        {M'~M_2~\dots~M_n \to_{\sigma \cup \{y \mapsto M_1^\sigma\}} M''~M_2~\dots~M_n}
        {(\lambda y.~M')~M_1~\dots~M_n \to_{\sigma} (\lambda y.~M'')~M_1~\dots~M_n}
\end{mathpar}
Then, we have $M\to_{\emptyset} M'$ iff $M \to_{\lhr} M'$.
\label{lem:openlhr}
\end{lem}
\begin{proof}
We prove that for all $\sigma$:
\begin{itemize}
\item If the head occurrence of $M$ is bound in $M$, then $M \to_{\lhr} M'$ iff $M\to_{\sigma} M'$,
\item If the head occurrence $x_0$ of $M$ is an occurrence of a variable $x$ free in $M$ on which $\sigma$ is defined, then
$
M \to_{\sigma} M[\sigma_1(x)/x_0]
$.
\item If the head occurrence of $M$ is a $\daimon$ or an occurrence $x_0$ of $x$ on which $\sigma$ is undefined, then both $\to_{\lhr}$ and $\to_{\sigma}$ halt.
\end{itemize}
This is done by a straightforward induction on the length of $M$.
\end{proof}

Environments have a tree structure that corresponds closely to the tree structure of interaction skeletons. The following
definition associates a skeleton with any environment.

\begin{defi}
If $\sigma$ is an environment on $\Gamma$, then we define a bs-environment $\intr{\sigma}$ which for any $(x:A) \in \Gamma$,
\emph{(1)} is undefined if $\sigma(x)$ is undefined, \emph{(2)} associates $\intr{\sigma_1(x)}_{\intr{\sigma_2(x)}}$ otherwise.
\end{defi}

Finally, we are in position to prove our simulation result.

\begin{prop}
Let $\Gamma \vdash M, M':o$ be $\eta$-long ls terms, and $\sigma$ be an $\eta$-long ls flat environment on $\Gamma$ such that
$M \to_{\sigma} M'$. Then,
$
\intr{M}_{\intr{\sigma}} \to_{\bs} \hookleftarrow \intr{M'}_{\intr{\sigma}}
$.
So in particular, if $M \to_{\lhr} M'$ we have $\intr{M} \to_{\bs} \hookleftarrow \intr{M'}$.
\label{prop:intr_pres_red}
\end{prop}
\begin{proof}
By induction on the definition of $\to_{\sigma}$.
\begin{itemize}
\item If $x_0~M_1~\dots~M_n \to_{\sigma} \sigma_1(x)~M_1~\dots~M_n$. 
Since $\sigma_1(x)$ is locally scoped and $\eta$-long, by Lemma \ref{lem:appbs} we have:
\[
\intr{\sigma_1(x)~M_1~\dots~M_n}_{\intr{\sigma}} \embeds \intr{\sigma_1(x)}_{\intr{\sigma}}\cdot [\{\lv(M_i)+1\}\intr{M_i}_{\intr{\sigma}}\mid 1 \leq i \leq n]
\]
But we have $\sigma_2(x) \subseteq \sigma$ since $\sigma$ is flat, from which follows $\intr{\sigma_2(x)} \subseteq \intr{\sigma}$.
Moreover, $\fv(\sigma_1(x))\cap \dom(\sigma) = \fv(\sigma_1(x)) \cap \dom(\sigma_2)$ since both are defined on all free variables
of $\sigma_1(x)$. So by Lemma \ref{lem:simeq}, $\intr{\sigma_1(x)}_{\intr{\sigma}} = \intr{\sigma_1(x)}_{\intr{\sigma_2(x)}}$.
Finally, we have:
\[
\intr{\sigma_1(x)}_{\intr{\sigma_2(x)}}\cdot [\{\lv(M_i)+1\}\intr{M_i}_{\intr{\sigma}}\mid 1 \leq i \leq n]
\embeds 
\intr{\sigma_1(x)}_{\intr{\sigma_2(x)}}\cdot_{\lv(x)} (\bigsqcup_{i=1}^n \intr{M_i}_{\intr{\sigma}})
\]
Where the right hand side is a $\leadsto$-reduct of
\[
\intr{x_0~M_1~\dots~M_n}_{\intr{\sigma}} = (1 + \bigsqcup_{i = 1}^n \intr{M_i}_{\intr{\sigma}} ) \cdot_{\lv(x)+1} \intr{\sigma_1(x)}_{\intr{\sigma_2(x)}}
\]
\item If $\lambda y.~M \to_{\sigma} \lambda y.~M'$, then $M$ and $M'$ do not have ground type.
\item If $(\lambda y.~M)~M_1~\dots~M_n \to_{\sigma} (\lambda y.~M')~M_1~\dots~M_n$, then
note that $\sigma \cup \{y \mapsto M_1^\sigma\}$ is still flat (and $\eta$-long, locally scoped).
Then by IH we have:
\[
\intr{M~M_2~\dots~M_n}_{\intr{\sigma \cup \{y \mapsto M_1^\sigma\}}} \to_{\bs} \hookleftarrow 
\intr{M'~M_2~\dots~M_n}_{\intr{\sigma \cup \{y \mapsto M_1^\sigma\}}}
\]
But $\intr{\sigma \cup \{y \mapsto M_1^\sigma\}} = \intr{\sigma} \cup \{y \mapsto \intr{M_1}_{\intr{\sigma}}\}$ by definition
of interpretation of environments; the required reduction follows by def. of interpretation of terms.\qedhere
\end{itemize}
\end{proof}

\subsubsection{Relating generalized game situations and their interpretation}
To estimate lhr on generalized game situations, we need to define measures on terms that reflect the geometry of
the corresponding skeletons. So instead of the quantities traditionally used to evaluate the complexity of $\lambda$-terms
(like the height or length), we have two alternative quantities. 

\begin{defi}
The \textbf{depth} $\depth(M)$ of $M$ is defined by
induction on the length of $M$:
\begin{eqnarray*}
\depth(\daimon~M_1~\dots~M_n) &=& 1\\
\depth(x_0~M_1~\dots~M_n) &=& \max_{1 \leq i \leq n} \depth(M_i)\\
\depth(\lambda x.~M) &=& \depth(M)\\
\depth((\lambda x.~M)~M_1~\dots~M_n) &=& \max(\depth(M~M_2~\dots~M_n), \depth(M_1)+1)
\end{eqnarray*}
Likewise, the \textbf{local height} $\lh(M)$ of a term $M$ is defined by:
\begin{eqnarray*}
\lh(\daimon~M_1~\dots~M_n) &=& 0\\
\lh(x_0~M_1~\dots~M_n) &=& 1 + \max_{1\leq i \leq n} \lh(M_i)\\
\lh(\lambda x.~M) &=& \lh(M)\\
\lh((\lambda x.~M)~M_1~\dots~M_n) &=& \max(\lh(M~M_2~\dots~M_n), \lh(M_1))
\end{eqnarray*}
\end{defi}

We now aim to prove that these indeed reflect quantities on the corresponding skeletons. 
Lemma \ref{lem:intr_pres_depth}
deals with depth, Lemma \ref{lem:intr_pres_lh} with local height and Lemma \ref{lem:intr_pres_ord} with order.

\begin{lem}
If $\rho$ is a bs-environment, we set $\depth(\rho) = \max_{x\in \dom(\rho)} \depth(\rho(x))$. Then, for each strongly locally scoped term $M$
with a bs-environment $\rho$, we have:
\[
\depth(\intr{M}_\rho) \leq \max(\depth(M), \depth(\rho)+1)
\]
In particular, $\depth(\intr{M}) \leq \depth(M)$.
\label{lem:intr_pres_depth}
\end{lem}
\begin{proof}
By induction on the length of $M$, detailing the non-trivial cases.
\begin{itemize}
\item If $M = x_0~M_1~\dots~M_n$ and $\rho$ is not defined on $x$, we have $\depth(M) = \max_{1 \leq i \leq n} \depth(M_i)$. On the other hand,
$
\intr{x_0~M_1~\dots~M_n}_{\rho} = (1 + \bigsqcup_{i=1}^n \intr{M_i}_{\rho})
$.
But then we have:
\begin{eqnarray*}
\depth(1 + \bigsqcup_{i=1}^n \intr{M_i}_{\rho}) 
                                                &=& \max_{1\leq i \leq n} \depth(\intr{M_i}_\rho)\\
                                                &\leq& \max_{1\leq i \leq n} \max(\depth(M_i),\depth(\rho)+1)\\
                                                &=& \max(\depth(M), \depth(\rho) + 1)
\end{eqnarray*}
where the first equality is by definition of $+$, depth, and $\bigsqcup$,
the inequality is by IH, and the last equality is by definition of maximum and $\depth$.
\item If $M = x_0~M_1~\dots~M_n$, $\rho(x)$ defined, we still have $\depth(M) = \max_{1 \leq i \leq n} \depth(M_i)$. On the other hand,
$\intr{x_0~M_1~\dots~M_n}_{\rho} = (1 + \bigsqcup_{i=1}^n \intr{M_i}_{\rho})\cdot_{\lv(x)+1} \rho(x)$.
But then:
\begin{eqnarray*}
\depth((1 + \bigsqcup_{i=1}^n \intr{M_i}_{\rho})\cdot_{\lv(x)+1} \rho(x)) 
&=& \max(\depth(1 + \bigsqcup_{i=1}^n \intr{M_i}_{\rho}), \depth(\rho(x)) + 1)\\
&\leq& \max(\max_{1\leq i \leq n} \depth(\intr{M_i}_\rho), \depth(\rho) + 1)\\
&\leq& \max(\max_{1\leq i \leq n} \depth(M_i),\depth(\rho)+1)
\end{eqnarray*}
where the first equality is by definition of $\cdot$, the first inequality is by definition of $+$, $\bigsqcup$ and $\depth(\rho)$, the second
inequality is by IH and definition of $\max$.
\item If $M = \lambda x.~M'$, then it directly follows from the IH.
\item If $M = (\lambda x.~M')~M_1~\dots~M_n$, then we have
\[
\depth(M) = \max(\depth(M'~M_2~\dots~M_n), \depth(M_1)+1)
\]
 We calculate:
\begin{eqnarray*}
&&\depth(\intr{(\lambda x.~M')~M_1~\dots~M_n}_\rho)\\
&=& \depth(\intr{M'~M_2~\dots~M_n}_{\rho \cup \{x \mapsto \intr{M_1}_{\rho}\}})\\
&\leq& \max(\depth(M'~M_2~\dots~M_n),\depth(\rho \cup \{x \mapsto \intr{M_1}_{\rho}\})+1)\\
&\leq& \max(\depth(M'~M_2~\dots~M_n),\max(\depth(\rho), \depth(\intr{M_1}_{\rho}))+1)\\
&=&\max(\depth(M'~M_2~\dots~M_n),\max(\depth(\rho), \depth(\intr{M_1}_\emptyset))+1)\\
&\leq& \max(\depth(M'~M_2~\dots~M_n),\max(\depth(\rho), \max(\depth(M_1), 1))+1)\\
&\leq& \max(\depth(M'~M_2~\dots~M_n), \depth(M_1) + 1, \depth(\rho)+1)\\
&=& \max(\depth((\lambda x.~M')~M_1~\dots~M_n), \depth(\rho) + 1)
\end{eqnarray*}
Where the first equality is by definition of interpretation, the second line is by IH, the third line is by
definition of $\depth$ on environments, the fourth line uses that since $M$ is strongly locally scoped $M_1$ must be closed,
therefore $\rho \simeq_{M_1} \emptyset$ hence by Lemma \ref{lem:simeq} we have $\intr{M_1}_{\rho} = \intr{M_1}_{\emptyset}$. The
fifth line is by IH on $M_1$, and the last two lines are by easy manipulations on maximums (using that depth is
always greater than one) and definition of depth.\qedhere
\end{itemize}
\end{proof}

\begin{lem}
If $\rho$ is a bs-environment, we set $\max(\rho) = \max_{x\in \dom(\rho)} \max(\rho(x))$. Then, for each term $M$ with a bs-environment
$\rho$, for any natural number $m\in \mathbb{N}$, we have:
\[
\max(m + \intr{M}_\rho) \leq \max(\lh(M) + m, \max(\rho))
\]
In particular, $\max(\intr{M}) \leq \lh(M)$.\qed
\label{lem:intr_pres_lh}
\end{lem}

\proof
By induction on the length of $M$, omitting basic manipulations of expressions.
\begin{itemize}
\item If $M = \daimon~M_1~\dots~M_n$, it is direct.
\item If $M = x_0~M_1~\dots~M_n$ with $\rho(x)$ undefined, we calculate:
\begin{eqnarray*}
\max(m + \intr{M}_{\rho})   
                        &=& \max_{1\leq i \leq n} ((m+1) + \max(\intr{M_i}_\rho))\\
                        &\leq& \max_{1\leq i \leq n} \max(\lh(M_i) + m+1, \max(\rho))\\
                        &=& \max(\lh(M) + m, \max(\rho))
\end{eqnarray*}
\item If $M = x_0~M_1~\dots~M_n$ and $\rho(x)$ is defined, we calculate:
\begin{eqnarray*}
\max(m + \intr{M}_{\rho})   
                        &=& \max((\bigsqcup_{i=1}^n (m+1) + \intr{M_i}_\rho)\cdot_{\lv(x)+1} \rho(x))\\
                        &\leq& \max(\max_{1\leq i \leq n} \max(m+1 + \intr{M_i}_\rho), \max(\rho))\\
                        &\leq& \max(\max_{1\leq i\leq n} \max(\lh(M_i) + m+1, \max(\rho)), \max(\rho))\\
                        &=& \max(\lh(M) + m, \max(\rho))
\end{eqnarray*}
\item If $M = (\lambda x.~M)~M_1~\dots~M_n$, then we calculate:
\begin{eqnarray*}
\max(m + \intr{M}_\rho) &=& \max(m + \intr{M~M_2~\dots~M_n}_{\rho \cup \{x\mapsto \intr{M_1}_\rho\}})\\
                        &\leq& \max(\lh(M~M_2~\dots~M_n) + m, \max(\rho \cup \{x\mapsto \intr{M_1}_\rho\}))\\
                        &=& \max(\lh(M~M_2~\dots~M_n) + m, \max(\rho), \max(\intr{M_1}_\rho))\\
                        &\leq& \max(\lh(M~M_2~\dots~M_n) + m, \max(\rho), \max(\lh(M_1), \max(\rho)))\\
                        &\leq& \max(\max(\lh(M~M_2~\dots~M_n), \lh(M_1)) + m, \max(\rho))\\
                        &=& \max(\lh(M) + m, \max(\rho))\rlap{\hbox to
                            186 pt{\hfill\qEd}}
\end{eqnarray*}
\end{itemize}

\begin{lem}
If $\Gamma \vdash M:A$ is a term with a bs-environment $\rho$, we define
$\ord(\rho) = \max_{x\in \dom(\rho)} \ord(\rho(x))$. Then, we have:
\[
\ord(\intr{M}_\rho) \leq \max(\ord(M), \ord(\rho), \max_{x\in \dom(\rho)} (\lv(x) + 1))
\]
In particular, $\ord(\intr{M}) \leq \ord(M)$.
\label{lem:intr_pres_ord}
\end{lem}

\proof
By induction on the length of $M$, omitting some basic manipulations.
\begin{itemize}
\item If $M = x_0~M_1~\dots~M_n$, $\rho(x)$ undefined, then $\intr{M}_{\rho} = 1 + \bigsqcup_{i=1}^n \intr{M_i}_{\rho}$.
We calculate:
\begin{eqnarray*}
\ord(\intr{M}_{\rho})   
                        &=& \max_{1\leq i \leq n} \ord(\intr{M_i}_{\rho})\\
                        &\leq& \max_{1\leq i \leq n} \max(\ord(M_i), \ord(\rho), \max_{x\in \dom(\rho)} (\lv(x) + 1))\\
                        &=& \max(\max_{1\leq i \leq n} \ord(M_i), \ord(\rho), \max_{x\in \dom(\rho)} (\lv(x) + 1))\\
                        &\leq& \max(\ord(M), \ord(\rho), \max_{x\in \dom(\rho)} (\lv(x) + 1))
\end{eqnarray*}
\item If $M = x_0~M_1~\dots~M_n$, $\rho(x)$ defined, then 
$\intr{M}_{\rho} = (1 + \bigsqcup_{i=1}^n \intr{M_i}_{\rho})\cdot_{\lv(x)+1} \rho(x)$.
\begin{eqnarray*}
\ord(\intr{M}_\rho) &=& \max(\max_{1\leq i \leq n} \ord(\intr{M_i}_\rho), \lv(x)+1, \ord(\rho(x)))\\
&\leq& \max(\max_{1\leq i \leq n} \max(\ord(M_i), \ord(\rho), \max_{y\in \dom(\rho)} (\lv(y)+1)), \lv(x)+1, \ord(\rho))\\
&=& \max(\ord(M), \ord(\rho), \max_{y\in \dom(\rho)} (\lv(y) + 1))
\end{eqnarray*}
\item If $M = \lambda x.~M'$, then it is straightforward.
\item If $M = (\lambda x^B.~M')~M_1~\dots~M_n$, writing $\rho' = \rho \cup \{x\mapsto \intr{M_1}_\rho\}$:
\begin{eqnarray*}
\ord(\intr{M}_\rho)
 &=& \ord(\intr{M'~M_2~\dots~M_n}_{\rho \cup \{x\mapsto \intr{M_1}_\rho\}})\\
&\leq& \max(\ord(M'~M_2~\dots~M_n), \ord(\rho \cup \{x\mapsto \intr{M_1}_\rho\}), \max_{y\in \dom(\rho')} (\lv(y) + 1))\\
&=& \max(\ord(M'~M_2~\dots~M_n), \ord(\intr{M_1}_\rho), \ord(\rho), \max_{y\in \dom(\rho')} (\lv(y) + 1))\\
&\leq& \max(\ord(M'~M_2~\dots~M_n), \ord(M_1), \lv(x) + 1, \ord(\rho), \max_{y\in \dom(\rho)} \lv(y) + 1)\\
&=& \max(\ord(M), \ord(\rho), \max_{y\in \dom(\rho)} \lv(y) +
    1)\rlap{\hbox to 162 pt{\hfill\qEd}}
\end{eqnarray*}
\end{itemize}

\noindent We can now summarize the results of this subsection with the following proposition.

\begin{prop}
If $M$ is a strongly locally scoped, $\eta$-long term of ground type, then:
\[
\begin{array}{rclcrcl}
\depth(\intr{M}) &\leq& \depth(M) &~~~~~~~~~~& \max(\intr{M}) &\leq& \lh(M)\\
\ord(\intr{M}) &\leq& \ord(M) && \norm(\intr{M}) &\geq& \norm(M)
\end{array}
\]
\label{prop:ls_to_bs}
\end{prop}
\begin{proof}
For depth, local height and order, it is a consequence respectively of Lemmas \ref{lem:intr_pres_depth}, \ref{lem:intr_pres_lh}
and \ref{lem:intr_pres_ord}. For the norm, we prove it on locally scoped (not strongly) terms, by induction on 
$\norm(\intr{M})$.
If $\norm(\intr{M}) = 0$, then $\norm(M) = 0$ as well by Proposition \ref{prop:intr_pres_red}.
If $\norm(M) = 0$, this is obvious. Otherwise we have $M \to_{\lhr} M'$. By Lemma \ref{lem:pres_loc_scope}, $M'$ is still
locally scoped. By Lemma \ref{lem:lhr_pres_etalong}, it is also still $\eta$-long. By Lemma \ref{lem:openlhr}, we have $M \to_{\emptyset} M'$.
By Proposition \ref{prop:intr_pres_red}, it follows that there is a skeleton $a$ such that:
$\intr{M}_{\emptyset} \to_{\bs} a \hookleftarrow \intr{M'}_{\emptyset}$.
By definition of norm and Lemma \ref{lem:embedsbs}, we have
$
\norm(\intr{M}_{\emptyset}) > \norm(a) \geq \norm(\intr{M'}_{\emptyset})
$.
But by IH, we know that $\norm(M') \leq \norm(\intr{M'}_{\emptyset})$. So,
$\norm(M) = \norm(M') + 1 \leq \norm(\intr{M'}_{\emptyset})+1 < \norm(\intr{M}_{\emptyset})+1$.
\end{proof}

\subsection{Bounds for strongly locally scoped terms} 
\label{subsec:bounds}
With Proposition \ref{prop:ls_to_bs} and Theorem \ref{thm_upbs} we can already deduce upper bounds for the
length of lhr on sls $\eta$-long terms. However it will turn out that $\eta$-expansion does not change the
asymptotic bounds, so we first deal with $\eta$-expansion and we will then formulate bounds for non necessarily
$\eta$-long sls terms.

\subsubsection{How $\eta$-expansion affects the norm}

\begin{lem}
If $M$ has a head occurrence $x_0$ of a variable $x$, and $M \to_{\eta} M'$, then for any term $N$, we have $M[N/x_0] \to_{\eta} M'[N/x_0]$.
\label{lem:comm_eta_subst}
\end{lem}
\begin{proof}
Direct by induction on $M$.
\end{proof}

\begin{lem}
If $(\lambda x, N)$ is a prime redex in $M~N$, then for any $N'$, $(\lambda x, N')$ is a prime redex of $M~N'$.
\label{lem:arg_replace}
\end{lem}
\begin{proof}
Direct by induction on $M$. 
\end{proof}

\begin{lem}
Suppose $M \to_{\eta} M'$, with $(\lambda x, N)$ prime redex in $M~N$ and $x_0$ head occurrence of $M$, then
there is a term $M''$ such that
$
M'~N \to_{\lhr}^+ M''~N
$
and $M[N/x_0]~N \to_{\eta} M''~N$.
\label{lem:etalhrapp}
\end{lem}
\begin{proof}
By induction on the length of $M$.
If the $\eta$-expansion is external, \emph{i.e.} $M \to_{\eta} \lambda y.~M~y_0$, then by Lemma \ref{lem:arg_replace} we know that
$(\lambda x, y_0)$ is a prime redex in $\lambda y.~M~y_0$. It follows that we have 
$(\lambda y.M~y_0)~N \to_{\lhr} (\lambda y.~M[y_1/x_0]~y_0)~N$. But then, $y_1$ becomes the head occurrence, and $(\lambda y, N)$ is a prime redex, therefore:
\[
(\lambda y.~M[y_1/x_0]~y_0)~N \to_{\lhr} (\lambda y.~M[N/x_0]~y_0)~N
\]
and to conclude, we obviously have $M[N/x_0]~N \to_{\eta} (\lambda y.~M[N/x_0]~y_0)~N$.

If the $\eta$-expansion is internal, we reason by cases on the form of $M$.
\begin{itemize}
\item If $M = x_0~M_1~\dots~M_n$, then $M$ does not have any prime redex.
\item If $M = \lambda x.~M'$, then necessarily $x_0$ is an occurrence of $x$.
Then we have $M'\to_{\eta} M''$, and the prime redex $(\lambda x, N)$ in $M~N$ is still in
$(\lambda x.~M'')~N$. Moreover, the head occurrence of $M''$ is still $x_0$. Therefore,
$(\lambda x.~M'')~N \to_{\lhr} (\lambda x.~M''[N/x_0])~N$.
By Lemma \ref{lem:comm_eta_subst}, $M'[N/x_0] \to_{\eta} M''[N/x_0]$, so
$(\lambda x.~M'[N/x_0])~N \to_{\eta} (\lambda x.~M''[N/x_0])~N$.
\item If $M = (\lambda y.~M')~M_1~\dots~M_n$, then by cases on the location of the $\eta$-expansion:
\begin{itemize}
\item For $M \to_{\eta} (\lambda y.~M'')~M_1~\dots~M_n$, it follows directly from IH.
\item For the remaining cases, the same lhr reduction as in $M$ is possible and obviously
commutes with the $\eta$-expansion.\qedhere
\end{itemize}
\end{itemize}
\end{proof}

\begin{lem}
If $M \to_\eta N$ and $M \to_{\lhr} S$, then there is $T$ such that $N \to_{\lhr^+} T$ and $S \to_{\eta^+} T$.
\label{lem:loccomm_etalhr}
\end{lem}
\begin{proof}
By induction on $M$, detailing only non-trivial cases.
\begin{itemize}
\item If $M$ has the form $\lambda x.~M'$, then $M\to_\lhr S$ with $S = \lambda x.~S'$.
Since we have $M \to_\eta N$, several cases arise:
\begin{itemize}
\item If $N = \lambda y.~M~y_0$, it is immediate with $T = \lambda y.~(\lambda x. S')~y$.
\item If $M' \to_\eta N'$, IH provides $T'$ with $N' \to_{\lhr^+} T'$ and $S' \to_{\eta^+} T'$. 
Setting $T = \lambda x.~T'$, the commutation follows.
%
%
\end{itemize}
\item If $M$ has the form $M_1~M_2$. Three cases arise:
\begin{itemize}
\item If $M_1~M_2 \to_{\eta} \lambda x.~M_1~M_2~x$, then the prime redexes of $\lambda x.~M_1~M_2~x$ are included in those of
$M_1~M_2$, and the head occurrence is the same. It follows that if $M_1~M_2 \to_{\lhr} M'_1~M_2$, we have that
$\lambda x.~M_1~M_2~x \to_{\lhr} \lambda x.~M'_1~M_2~x$ as well. Setting $T = \lambda x.~M'_1~M_2~x$ yields the required diagram.
\item If $M_1~M_2 \to_{\eta} M'_1~M_2$ (so $M_1 \to_{\eta} M'_1$). Then, two cases arise.
\begin{itemize}
\item If the reduction $M_1~M_2 \to_{\lhr} S_1~M_2$ involves a prime redex $(\lambda y, M_2)$, and the head occurrence of $M_1$ is $y_0$.
Then by Lemma \ref{lem:etalhrapp},
there is a term $T'$ such that $M'_1~M_2 \to_{\lhr}^+ T'~M_2$
and $M_1[M_2/y_0]~M_2 \to_{\eta} T'~M_2$. Setting $T = T'~M_2$, we have the required commutation.
\item If the reduction $M_1~M_2 \to_{\lhr} S_1~M_2$ does not involve $M_2$, then we also have
$M_1 \to_{\eta} M'_1$, and $M_1 \to_{\lhr} S_1$. By IH, there is $T_1$ such that $M'_1 \to_{\lhr^+} T_1$ and
$S_1 \to_{\eta^+} T_1$. Setting $T = T_1 M_2$ gives the required commutation.
\end{itemize}
\item If $M_1~M_2 \to_{\eta} M_1~M'_2$, then two sub-cases arise:
\begin{itemize}
\item If $M_1~M_2 \to_{\lhr} S_1~M_2$ is obtained by firing a prime redex $(\lambda x, M_2)$ with $x_0$ head occurrence of $M_1$, then
setting $T = M_1[M'_2/x_0]$ yields $M_1~M'_2 \to_{\lhr} T$ and $S_1~M_2 = M_1[M_2/x_0]~M_2 \to_{\eta^2} M_1[M'_2/x_0]~M'_2$.
\item If $M_1~M_2 \to_{\lhr} S_1~M_2$ is obtained by firing a prime redex not involving $M_2$, then setting $T = S_1 M'_2$ immediately
yields the required commutation.\qedhere
\end{itemize}
\end{itemize}
\end{itemize}
\end{proof}

\begin{lem}
If $M\to_{\eta^n} N$ and $M\to_{\lhr^p} S$, then there are $T$, $n'\geq n$ and $p'\geq p$ such that $N\to_{\lhr^{p'}} T$ and $S\to_{\eta^{n'}} T$.
\label{lem:commstar_etalhr}
\end{lem}
\begin{proof}
We reason by induction on the lexicographic ordering of pairs $(\norm(N), n)$, using that $\norm(N)$ is finite by 
Proposition \ref{prop:lhr_term}. If $n=0$ or $p=0$, it is obvious.
Otherwise, we have $M \to_{\eta} N_1 \to_{\eta^{n-1}} N$ and $M\to_{\lhr} S_1 \to_{\lhr^{p-1}} S$. By Lemma \ref{lem:loccomm_etalhr}, there
is $N'_1$ such that $M_1 \to_{\eta^+} N'_1$ and $N_1 \to_{\lhr^+} N'_1$, let us write $M_1 \to_{\eta^q} N'_1$ and
$N_1 \to_{\lhr^r} N'_1$, with $r, q\geq 1$.
We represent the situation on the following diagram:
\[
\xymatrix{
M       \ar[r]^{\eta}
        \ar[d]_{\lhr}&
N_1     \ar[r]^{\eta^{n-1}}
        \ar[d]^{\lhr^r}&
N\\
S_1     \ar[r]^{\eta^q}
        \ar[d]_{\lhr^{p-1}}&
N'_1\\
S
}
\]
We have $n-1 < n$ and $\norm(N)$ is unchanged, therefore by IH, there is $N'$ as displayed here:
\[
\xymatrix{
M       \ar[r]^{\eta}
        \ar[d]_{\lhr}&
N_1     \ar[r]^{\eta^{n-1}}
        \ar[d]^{\lhr^r}&
N       \ar[d]^{\lhr^{r'}}\\
S_1     \ar[r]^{\eta^q}
        \ar[d]_{\lhr^{p-1}}&
N'_1    \ar[r]^{\eta^{n'}}&
N'\\
S
}
\]
with $r' \geq r$ and $n'\geq n-1$.
Finally, $\norm(N') < \norm(N)$, therefore by IH there is $T$ closing the diagram:
\[
\xymatrix{
M       \ar[r]^{\eta}
        \ar[d]_{\lhr}&
N_1     \ar[r]^{\eta^{n-1}}
        \ar[d]^{\lhr^r}&
N       \ar[d]^{\lhr^{r'}}\\
S_1     \ar[r]^{\eta^q}
        \ar[d]_{\lhr^{p-1}}&
N'_1    \ar[r]^{\eta^{n'}}&
N'      \ar[d]^{\lhr^{p'}}\\
S       \ar[rr]^{\eta^{n''}}&&
T
}
\]
with $p' \geq p-1$ hence $p'+r' \geq p-1 + r \geq p-1 + 1 \geq p$, and $n'' \geq q + n' \geq 1 + n-1 \geq n$.
\end{proof}

\begin{prop}
If $M\to_{\eta} N$, then $\norm(N) \geq \norm(M)$.
\label{prop:etanorm}
\end{prop}
\begin{proof}
We show by induction on $\norm(M)$ that if $M\to_{\eta}^* N$, then $\norm(N) \geq \norm(M)$.
If $\norm(M) = 0$, this is obvious. Otherwise, there is $M \to_{\lhr} M'$. By Lemma \ref{lem:commstar_etalhr} there is $N'$
such that $N \to_{\lhr^+} N'$ and $M' \to_{\eta}^* N'$. But $\norm(M') < \norm(M)$, so by IH
we have $\norm(N') \geq \norm(M')$. Since $\norm(M) = \norm(M') + 1$ and $\norm(N) \geq \norm(N') + 1$, we have
$\norm(N) \geq \norm(M)$ as well.
\end{proof}

\subsubsection{How $\eta$-expansion affects other quantities on terms}

\begin{defi}
An $\eta$-expansion step $M \to_{\eta} M'$ is \textbf{restricted} if $M$ and $M'$ have the same number of generalized redexes. 
We write $M \to_{\eta_r} M'$.
\end{defi}

\begin{lem}
On any term $M$, the expansion $\to_{\eta_r}$ terminates.
\label{lem:etaterm}
\end{lem}
\begin{proof}
%
Suppose $A$ is a type, then its \textbf{size} written $\sz(A)$, is defined by induction on $A$ by
$\sz(o) = 1$ and $\sz(A \to B) = \sz(A) + \sz(B)$.
If $M$ is a term, its \textbf{deficiency} quantifies the lack of $\eta$-expansion within $M$. It is defined by induction on the length of $M$,
as follows:
\begin{itemize}
\item If $M = x_0~M_1~\dots~M_n$, where $x_0$ is an occurrence of a variable or constant, and if $x: A_1 \to \dots \to A_p \to o$,
then set
$\df(M) = \sum_{i=1}^n \df(M_i) + \sum_{i=n+1}^p \sz(A_i)$.
\item If $M = \lambda x.~M'$, we set $\df(M) = \df(M')$.
\item If $M = (\lambda x.~M')~M_1~\dots~M_n$, we set $\df(M) = \df(M'~M_2~\dots~M_n) + \df(M_1)$.
\end{itemize}

Restricted $\eta$-expansion strictly decreases deficiency, by induction on the length of $M$:

\begin{itemize}
\item If $M = x_0~M_1~\dots~M_n$ where $x_0$ is a variable or constant occurrence having type
$A_1 \to \dots \to A_p \to o$, then by definition we have $\df(M) = \sum_{i=1}^n \df(M_i) + \sum_{i=n+1}^p \sz(A_i)$.
We reason by cases on the location of the $\eta$-expansion. If it is:
\[
x_0~M_1~\dots~M_n \to_{\eta_r} \lambda y.~x_0~M_1~\dots~M_n~y_0
\]
then $\df(\lambda y.~x_0~M_1~\dots~M_n~y_0) = \sum_{i=1}^n \df(M_i) + \df(y_0) + \sum_{i=n+2}^p \sz(A_i)$.
But it is obvious by definition that $\df(y_0) = \sz(A_{n+1})-1$, so deficiency is reduced.
If the $\eta$-expansion is $x_0~M_1~\dots~M_n \to_{\eta_r} x_0~M_1~\dots~M'_i~\dots~M_n$, 
then $M_i \to_{\eta_r} M'_i$, so by IH we have $\df(M'_i) < \df(M_i)$. So,
$\df(x_0~M_1~\dots~M'_i~\dots~M_n) < \df(x_0~M_1~\dots~M_n)$.
\item If $M = \lambda x.~M'$, the $\eta$-expansion cannot be $\lambda x.~M' \to \lambda y.~(\lambda x.~M')~y_0$ as that would
create a new generalized redex, so it is $M' \to_{\eta_r} M''$. The property follows from IH.
\item If $M = (\lambda x.~M')~M_1~\dots~M_n$, then the only possible restricted $\eta$-expansions are:
\begin{eqnarray*}
(\lambda x.~M')~M_1~\dots~M_n &\to_{\eta_r}& \lambda y.~(\lambda x.~M')~M_1~\dots~M_n~y_0\\
(\lambda x.~M')~M_1~\dots~M_n &\to_{\eta_r}& (\lambda x.~M'')~M_1~\dots~M_n\\
(\lambda x.~M')~M_1~\dots~M_n &\to_{\eta_r}& (\lambda x.~M')~M_1~\dots~M'_i~\dots~M_n
\end{eqnarray*}
All the others create a generalized $\beta$-redex. In all three cases, the result follows from IH and
definition of deficiency.
\end{itemize}
We have a positive strictly decreasing measure for $\eta_r$, so it terminates.
\end{proof}

\begin{lem}
If $M$ is a term and $M \to_{\eta_r}^* M'$, then $\lh(M') \leq \lh(M) + \ord(M)$.
\label{lem:etapreslh}
\end{lem}
\begin{proof}
We define a quantity $\lh'(M)$ such that $\lh(M) \leq \lh'(M) \leq \lh(M) + \ord(M)$
and we show that $\lh'$ is preserved by restricted $\eta$-expansion. It is defined as for $\lh$, except for variables where
we set, with $x_0$ occurrence of $x:A_1 \to \dots \to A_p \to o$:
\[
\lh'(x_0~M_1~\dots~M_n) = 1 + \max(\max_{1\leq i \leq n} \lh'(M_i), \max_{1\leq i\leq n} \lv(A_i))
\]
Clearly, we have $\lh(M) \leq \lh'(M) \leq \lh(M) + \ord(M)$. We prove by induction on $M$ that it is preserved by
restricted $\eta$-expansion, skipping the trivial case.
\begin{itemize}
\item If $M = x_0~M_1~\dots~M_n$ where $x_0$ is a variable occurrence of a variable $x: A_1 \to \dots \to A_p \to o$, then
several cases depending on the location of the $\eta$-expansion.
Firstly, if the $\eta$-expansion is:
$
x_0~M_1~\dots~M_n \to_{\eta_r} \lambda y.~x_0~M_1~\dots~M_n~y_0
$ we have
\begin{eqnarray*}
\lh'(x_0~M_1~\dots~M_n)                 &=& 1 + \max(\max_{1\leq i \leq n} \lh'(M_i), \max_{n+1\leq i \leq p} \lv(A_i))\\
\lh'(\lambda y.~x_0~M_1~\dots~M_n~y)    &=& 1 + \max(\max_{1\leq i \leq n} \lh'(M_i), \lh'(y), \max_{n+2\leq i \leq p} \lv(A_i))
\end{eqnarray*}
But then, writing $y: B_1 \to \dots \to B_m \to o$, we have $\lh'(y) = 1 + \max_{1\leq i \leq m} \lv(B_i) = \lv(A_{n+1})$,
so those two quantities are equal.
Secondly, if the restricted $\eta$-expansion is within some $M_i$, then the result follows immediately by IH.
\item If $M = (\lambda x.~M')~M_1~\dots~M_n$, then several cases following the location of the $\eta$-expansion.
For $(\lambda x.~M')~M_1~\dots~M_n \to_{\eta_r} \lambda y.~(\lambda x.~M')~M_1~\dots~M_n~y_0$, we calculate:
\begin{eqnarray*}
\lh'((\lambda x.~M')~M_1~\dots~M_n)     &=& \max(\lh'(M'~M_2~\dots~M_n), \lh'(M_1))\\
                                        &=& \max(\lh'(\lambda y.~M'~M_2~\dots~M_n~y_0), \lh'(M_1))\\
                                        &=& \max(\lh'(M'~M_2~\dots~M_n~y_0), \lh'(M_1))\\
                                        &=& \lh'((\lambda x.~M')~M_1~M_2~\dots~M_n~y_0))\\
                                        &=& \lh'(\lambda y.~(\lambda x.~M')~M_1~M_2~\dots~M_n~y_0)
\end{eqnarray*}
If the $\eta$-expansion is within $M'$ or $M_i$, it follows directly from IH. \qedhere
\end{itemize}
\end{proof}

\begin{lem}
If $M \to_{\eta_r} M'$, then $\depth(M) = \depth(M')$.
\label{lem:etapresdepth}
\end{lem}
\begin{proof}
Immediate by induction on $M$.
\end{proof}

\begin{lem}
If $M \to_{\eta} M'$, then $\ord(M) = \ord(M')$
\label{lem:etapresord}
\end{lem}
\begin{proof}
There is a subterm $N$ such that $N \to_{\eta} \lambda y.~N~y_0$. But $\lv(\lambda y.~N~y_0) = \lv(N)$ and
$\lv(N~y_0) \leq \lv(N)$, so the new subterms have lower level than the original ones.
\end{proof}

Finally, it remains to note that $\eta$-expansion preserves strong local scope.

\begin{lem}
If $M$ is strongly locally scoped and $M \to_{\eta_r} M'$, then $M'$ is sls.
\label{lem:eta_pres_ls}
\end{lem}
\begin{proof}
B straightforward induction on the length of $M$.
\end{proof}

\begin{prop}
If $M$ is a term, then there is an $\eta$-long term $M'$ such that:
\[
\begin{array}{rclcrcl}
\lh(M') &\leq& \lh(M) + \ord(M) &~~~~~~~~~~& \depth(M') &=& \depth(M)\\
\ord(M') &=&\ord(M) && \norm(M') &\geq& \norm(M)
\end{array}
\]
Moreover if $M$ was strongly locally scoped, so is $M'$.
\label{prop:etalongform}
\end{prop}
\begin{proof}
By Lemma \ref{lem:etaterm}, there is $M'$ such that $M \to_{\eta_r}^* M'$, and there is no further restricted $\eta$-expansion.
By definition, $M'$ is $\eta$-long. Moreover, the preservations of depth, order, local height and
norm follow respectively from Lemmas \ref{lem:etapresdepth}, \ref{lem:etapresord}, \ref{lem:etapreslh} and 
Proposition \ref{prop:etanorm}. The construction preserves strong local scope by Lemma \ref{lem:eta_pres_ls}.
\end{proof}

\subsubsection{Bounds for strongly locally scoped terms}
Putting everything together, we estimate:
\[
\loc_n(h, d)    = \max \{\norm(M) \mid \ord(M) \leq n~\&~\lh(M) \leq h~\&~\depth(M) \leq d~\&~\raisebox{3.5pt}{\txt{$M$ sls}}\}
\]

\begin{prop}
Suppose $M$ is a sls term of order at least one. Then,
\[
\norm(M) \leq 2_{\ord(M)-1}^{\depth(M)\log(\lh(M) + \ord(M) + 1)}
\]
\label{prop:boundetals}
\end{prop}
\begin{proof}
If $\Gamma \vdash M:A_1\to\dots\to A_n \to o$ is a sls term, we first make it of ground type by forming
$\Gamma \vdash M~\daimon_{A_1}~\dots~\daimon_{A_n} : o$ -- its norm can only increase, the other quantities stay unchanged and
the term is still sls. By Proposition \ref{prop:etalongform}, there is $M'$ $\eta$-long, of ground type, and 
sls such that $\lh(M') \leq \lh(M) + \ord(M)$, $\depth(M') = \depth(M)$, $\ord(M') = \ord(M)$ and $\norm(M') \geq \norm(M)$.
We conclude by Proposition \ref{prop:ls_to_bs} and Theorem \ref{thm_upbs}.
\end{proof}

We now prove the optimality of this upper bound by exhibiting a family of terms whose reduction length asymptotically reaches
it. This family of terms is closely related to the example used in Section \ref{sec:games} for game situations. For 
$n, k, p\geq 0$ and $M:A_p$, we define:
\[
\begin{array}{rclcrcl}
[n]^0_p(M) &=& M &~~~~~~~~~~~~~&[n]^{k+1}_p(M) &=& \church{n}_{p+1}~[n]^k_p(M)
\end{array}
\]
One can immediately check that $[n]_p^k(M): A_p$ and that for all $q\in \mathbb{N}$, $[n]_p^k(\church{q}_p) \to_\beta^* \church{q^{n^k}}_p$.
Exploiting this construction we set, for $n, k, p\geq 0$:
\[
S_{n, k, p} = [n]_p^k(\church{2}_p)~\church{2}_{p-1}~\dots~\church{2}_0
\]
For which it is immediate to check that for all $n, k, p \geq 0$ we have
$S_{n, k, p} \to_{\beta}^* \church{2_{p}^{2^{n^k}}}_0$.
Moreover, by construction of $S_{n, k, p}$, for $n\geq 2$ and $p, k\geq 1$ we have
$\lh(S_{n, k, p}) = n+1$, $\depth(S_{n, k, p}) = k+1$ and $\ord(S_{n, k, p}) = p+3$, and
$S_{n, k, p}$ is sls. To deduce
a lower bound from this, we use:

\begin{lem}
If $M \to_{\beta}^* \church{n}_0$, then
$
\norm(M~\id_o) \geq n
$,
where $\id_o = \lambda x^o.~x$.
\label{lem:main:churchlhr}
\end{lem}
\begin{proof} 
By induction on $n$, exploiting that lhr preserves $\beta$-equivalence.
\end{proof}

\begin{thm}
For fixed $n \geq 2$ we have
$
\loc_n(h, d) = 2_{n-1}^{\Theta(d\log(h))}
$.
\end{thm}
\begin{proof} 
We start with $n\geq 3$, $n=2$ requires a separate construction for the lower bound. Let us fix $h \geq 3$ and
$d\geq 2$. By Proposition \ref{prop:boundetals}, we already know that $\loc_n(d, h) \leq 2^{d\log(h + n + 1)}_{n - 1}$.
Moreover, we have $\lh(S_{h-1, d-1, n -3}~\id_o) = h$ and $\depth(T_{h-1, d-1, n-3}~\id_o) = d$, and by Lemma \ref{lem:main:churchlhr} we
have $\norm(T_{h-1, d-1, n-3}~\id_o) \geq 2_{n-1}^{(d-1)\log(h-1)}$. To summarize:
\[
2_{n-1}^{(d-1)\log(h-1)} \leq \loc_n(d, h) \leq 2^{d\log(h + n + 1)}_{n - 1}
\]
Therefore, with $n\geq 3$ fixed and $d, h$ parameters we have $\loc_n(h, d) = 2_{n-1}^{\Theta(d\log(h))}$.

For $n = 2$, the upper bound still holds. For $d, p\geq 2$, define:
\[
U_{n, d} = \church{n}_1~(\church{n}_1~\dots (\church{n}_1~\id_o)\dots)
\]
with $d$ copies of $\church{n}_1$ in total. Then, $U_{n, d}$ is sls and $\lh(U_{n, d}) = n+1$,
$\depth(U_{n, d}~\id_o) = d+1$, $\ord(U_{n, d}~\id_o) = 2$ and $\norm(U_{n, d}) \geq n^d = 2^{d\log(n)}$. 
It follows that $\loc_2(d, h) = 2^{\Theta(d\log(h))}$.
\end{proof}

In particular, reduction length for sls second-order terms of fixed depth is bounded by a polynomial of degree
less than the depth.

\subsection{Generalization to arbitrary terms}
\label{subsec:general}

In this final subsection, we deduce from the study of lhr of ls terms a
bound on the length of lhr of arbitrary terms. The key observation is that any $\lambda$-term can
be transformed into a locally scoped form through $\lambda$-lifting \cite{llift}.

\subsubsection{Lambda-lifting to sls terms}
Take a term $M = \lambda x^A.~(\lambda y^A.~y)~x$.
Obviously, $M$ is not sls: indeed there is a prime redex $(\lambda y, x)$ and the subterm $x$ has $x$ free. In order
to make the variable $x$ ``local'', we modify the abstraction subterm $\lambda y.~y$ to forward explicitly the variable $x$. We get
the term
$M' = \lambda x^A.~(\lambda y^{A\to A}.~y~x) (\lambda {x'}^A.~x')$.
The type of $y$ has changed, but not the type of the overall term. Note that the terms $M$ and $M'$ are still $\beta$-equivalent, although 
we are not going to use that explicitly. More importantly, the norm has increased, the order has increased by one, and the other quantities
are essentially unchanged. We formalize this construction by the $\lambda$-lifting expansion $\to_{\laml}$, defined in Figure \ref{fig:lamldef}.

\begin{figure}
\begin{mathpar}
\inferrule
        {y\in \fv(M_1)}
        {(\lambda x.~M)~M_1~\dots~M_n \to_{\laml} (\lambda x.~M[x~y/x])~(\lambda y'.~M_1[y'/y])~\dots~M_n}
\and
\inferrule
        {M_i \to_{\laml} M'_i}
        {x_0~M_1~\dots~M_n \to_{\laml} x_0~M_1~\dots~M'_i~\dots~M_n}
\and
\inferrule
        {M \to_{\laml} M'}
        {\lambda x.~M \to_{\laml} \lambda x.~M'}
\and
\inferrule
        {M_1 \to_{\laml} M'_1}
        {(\lambda x.~M)~M_1~\dots~M_n \to_{\laml} (\lambda x.~M)~M'_1~\dots~M_n}
\and
\inferrule
        {M~M_2~\dots~M_n \to_{\laml} M'~M'_2~\dots~M'_n \\ \text{($M_1$ closed)}}
        {(\lambda x.~M)~M_1~\dots~M_n \to_{\laml} (\lambda x.~M')~M_1~M'_2~\dots~M'_n}
\end{mathpar}
\caption{Definition of the $\lambda$-lifting expansion $\to_{\laml}$}
\label{fig:lamldef}
\end{figure}

First, we prove that $\to_\laml$ indeed allows us to convert any term $M$ into a sls $M'$. This is done
by showing that $\to_\laml$ terminates, and that its normal forms are sls.

\begin{lem}
Let $M$ be any term. Then, $\to_{\laml}$ terminates on $M$.
\label{lem:lamlterm}
\end{lem}
\begin{proof}
For each occurrence $x_0$ of $x$ in $M$, we define its \textbf{binding distance} $d_M(x_0)$:
\[
\begin{array}{rcll}
d_{x_0~M_1~\dots~M_n}(x_0) &=& 0\\
d_{\daimon~M_1~\dots~M_n}(x_0) &=& d_{M_i}(x_0) & (x_0 \in M_i)\\
d_{y_0~M_1~\dots~M_n}(x_0) &=& d_{M_i}(x_0) & (x_0 \in M_i)\\
d_{\lambda x.~M'}(x_0) &=& d_{M'}(x_0)\\
d_{(\lambda y.~M')~M_1~\dots~M_n}(x_0) &=& d_{M'}(x_0) & (x_0 \in M')\\
d_{(\lambda y.~M')~M_1~\dots~M_n}(x_0) &=& \max_{y_i\in oc_y(M')} d_{M'}(y_i) + d_{M_1}(x_0) + 1 & (x_0 \in M_1 \wedge x \in \fv(M'))\\
d_{(\lambda y.~M')~M_1~\dots~M_n}(x_0) &=& d_{M_1}(x_0) & (x_0 \in M_1 \wedge x_0 \not \in \fv(M_1))\\
d_{(\lambda y.~M')~M_1~\dots~M_n}(x_0) &=& d_{M'~M_2~\dots~M_n}(x_0) & (x_0 \in M_i \wedge i\geq 2)
\end{array}
\]
where $x_0 \in M$ means that $x_0$ is a variable occurrence $x$ in $M$ and $oc_y(M)$ is the set of occurrences of a
variable $y$ in $M$.
Then, for any term $M$ we consider the multiset $\mathcal{M}(M) = \{d_M(x_0) \mid x_0 \in oc_x(M)\}$
where by $\{-\mid-\}$ we denote here multiset comprehension.
Then, we show by induction on $M\to_{\laml} M'$ that $\mathcal{M}(M') < \mathcal{M}(M)$ for $<$ the multiset well-founded ordering on
multisets of natural numbers.
\begin{itemize}
\item If $M = (\lambda x.~M')~M_1~\dots~M_n \to_{\laml} (\lambda x.~M'[x~y/x])~(\lambda y'.~M_1[y'/y])~\dots~M_n = M''$:
\begin{itemize}
\item Occurrences of variables other than $y'$ in $M_1[y'/y]$ keep their binding distance,
\item Occurrences of $y$ in $M_1$ of distance $d_M(y_i) = \max_{x_i\in oc_x(M')} d_{M'}(x_i) + d_{M_i}(y_i) + 1$ have been replaced by occurrences of
$y'$ of distance $d_{M''}(y'_i) = d_{M_i}(y_i)$.
\item New occurrences of $y$ of distance less than $\max_{x_i\in oc_x(M')} d_{M'}(x_i)$ appear.
\end{itemize}
So, each variable occurrence is either unchanged or replaced by occurrences of strictly smaller distance, so
the reduction decreases the multiset ordering on $\mathcal{M}(M)$.
\item The other cases follow from IH since $\mathcal{M}(-)$ is additive over all term constructors.\qedhere
\end{itemize}
\end{proof}

\begin{lem}
If a term $M$ is $\laml$-normal, $M$ is strongly locally scoped.
\label{lem:normlocsc}
\end{lem}
\begin{proof}
Straightforward induction on the length of $M$.
\end{proof}

\subsubsection{How $\lambda$-lifting increases the norm}
Now, it remains to investigate how $\lambda$-lifting affects the relevant quantities on terms. First, we show that
it can only increase the norm.

\begin{lem}
If $M$ has a head occurrence $x_0$ where $x$ is free in $M$ and $N$ is a term, then
if $M\to_{\laml} M'$, then the head occurrence of $M'$ is still $x_0$ and $M[N/x_0] \to_{\laml} M'[N/x_0]$.
\label{lem:lamlsubst1}
\end{lem}
\begin{proof}
By induction on the length of $M$. The only non-trivial case is  for the term $M = (\lambda y.~M')~M_1~\dots~M_n$, where
three subcases arise:
\begin{itemize}
\item If $M' \to_{\laml} M''$, then by IH we know that $M'[N/x_0] \to_{\laml} M''[N/x_0]$, therefore
\[(\lambda y.~M'[N/x_0])~M_1~\dots~M_n \to_{\laml} (\lambda y.~M''[N/x_0])~M_1~\dots~M_n\]
\item If $M_i \to_{\laml} M'_i$, then 
$(\lambda y.~M'[N/x_0])~M_1~\dots~M_n \to_{\laml} (\lambda y.~M'[N/x_0])~M_1~\dots~M_n$.
\item If
$(\lambda y.~M')~M_1~\dots~M_n \to_{\laml} (\lambda y.~M'[y~z/y])~(\lambda z.~M_1)~\dots~M_n$, then
since $x$ is free in $M'$, necessarily $y\neq x$ ($y$ is bounded). Hence $M'[y~z/y][N/x_0] = M'[N/x_0][y~z/y]$,
so $(\lambda y.~M'[N/x_0])~M_1~\dots~M_n \to_{\laml} (\lambda y.~M'[y~z/y][N/x_0])~(\lambda z.~M_1)~\dots~M_n$.\qedhere
\end{itemize}
\end{proof}

\begin{lem}
If $M$ has a head occurrence $x_0$ where $x$ is free in $M$ and $N$ is a term with $N \to_{\laml} N'$, then we have $M[N/x_0] \to_{\laml} M[N'/x_0]$.
\label{lem:lamlsubst2}
\end{lem}
\begin{proof}
Straightforward by induction on the length of $M$.
\end{proof}

\begin{lem}
If $M \to_{\laml} M'$ and $M\to_{\lhr} N$, then there exists $S, T$ such that:
\[
\xymatrix@R=10pt{
M       \ar[r]^{\laml}
        \ar[dd]^{\lhr}&
M'      \ar[dr]^{\lhr}\\
&&T\\
N       \ar[r]^{\laml^+}&
S       \ar@{<-}[ur]^{\beta_{si}^\bullet}
}
\]
Where in $\beta_{si}^\bullet$, by $s$ we mean the $\beta$-reduction of a \emph{spine redex}, \emph{i.e.} a $\beta$-redex that is also a prime
redex, by $i$ we mean that it is \emph{internal}, in the sense that it is not the leftmost redex, and by $\bullet$ we mean that there are either
zero or one reduction steps.
\label{lem:lhrlaml}
\end{lem}
\begin{proof}
By induction on $M$, detailing only the non-trivial cases.
\begin{itemize}
\item If $M = \lambda x.~M'$, it follows directly from IH.
\item If $M = (\lambda x.~M')~M_1~\dots~M_n$ and the head occurrence of $M'$ is an occurrence $x_0$ of $x$, then 
$(\lambda x.~M')~M_1~\dots~M_n \to_{\lhr} (\lambda x.~M'[M_1/x_0])~M_1~\dots~M_n$.
There are several cases. If the $\to_{\laml}$ step comes from $M'~M_2~\dots~M_n \to_{\laml} M''~M'_2~\dots~M'_n$,
then by Lemma \ref{lem:lamlsubst1} we also have:
\[
M'[M_1/x_0]~M_2~\dots~M_n \to_{\laml} M''[M_1/x_0]~M'_2~\dots~M'_n
\]
It immediately follows:
\[
(\lambda x.~M'[M_1/x_0])~M_1~\dots~M_n \to_{\laml} (\lambda x.~M''[M_1/x_0])~M_1~M'_2~\dots~M'_n
\]
which provides the necessary commutation. If the $\to_{\laml}$ comes from $M_1 \to_{\laml} M'_1$, then by Lemma \ref{lem:lamlsubst2} we
also have $M'[M_1/x_0] \to_{\laml} M'[M'_1/x_0]$. It immediately follows:
\[
(\lambda x.~M'[M_1/x_0])~M_1~\dots~M_n \to_{\laml^2} (\lambda x.~M'[M'_1/x_0])~M'_1~\dots~M_n
\]
which provides the necessary commutation.
Finally, if the $\to_{\laml}$ step is:
\[
(\lambda x.~M')~M_1~\dots~M_n \to_{\laml} (\lambda x.~M'[x~z/x])~(\lambda z.~M_1)~\dots~M_n
\]
Then the head occurrence of $(\lambda x.~M'[x~z/x])~(\lambda z.~M_1)~\dots~M_n$ is still $x_0$. Moreover,
\begin{eqnarray*}
(\lambda x.~M'[x~z/x])~(\lambda z.~M_1)~\dots~M_n &\to_{\lhr}& (\lambda x.~M'[x~z/x][(\lambda z.~M_1)/x_0])~(\lambda z.~M_1)~\dots~M_n\\
&=& (\lambda x.~M'[((\lambda z.~M_1)~z)/x_0][x~z/x])~(\lambda z.~M_1)~\dots~M_n\\
&\rightarrow_{\beta_{si}}& (\lambda x.~M'[M_1/x_0][x~z/x])~(\lambda z.~M_1)~\dots~M_n\\
&\leftarrow_\laml& (\lambda x.~M'[M_1/x_0])~M_1~\dots~M_n
\end{eqnarray*}
\item If $M = (\lambda y.~M')~M_1~\dots~M_n$ whose head occurrence $x_0$ is not an occurrence of $y$, then three cases.
\begin{itemize}
\item If the $\to_{\laml}$-reduction comes from
$M'~M_2~\dots~M_n \to_{\laml} M''~M'_2~\dots~M'_n$, it follows directly from IH.
%
\item If the $\to_{\laml}$-reduction comes from $M_1 \to_{\laml} M'_1$, then it is obvious.
\item If the $\to_{\laml}$-reduction is:
\[
(\lambda y.~M')~M_1~\dots~M_n \to_{\laml} (\lambda y.~M'[y~z/y])~(\lambda z.~M_1)~\dots~M_n
\]
{\sloppy The head occurrence of $M'$ is not an occurrence of $y$, therefore the $\lhr$ reduct of
$(\lambda y.~M')~M_1~\dots~M_n$ must have the form $(\lambda y.~M'[N/x_0])~M_1~\dots~M_n$ for some $N$.
The head occurrence of $(\lambda y.~M'[y~z/y])~(\lambda z.~M_1)~\dots~M_n$ is still $x_0$, and:}
\begin{itemize}
\item If the associated prime redex $(\lambda x, N)$ has not changed, the same lhr is possible
and commutes with the $\laml$-reduction.
\item If the associated prime redex has become $(\lambda x, N[y~z/y])$,
we have $M'[N/x_0][y~z/y] = M'[N[y~z/y]/x_0][y~z/y]$, so the same reasoning as above applies.\qedhere
\end{itemize}
\end{itemize}
\end{itemize}
\end{proof}

\begin{prop}
If $M \to_{\laml} M'$, then $\norm(M') \geq \norm(M)$.
\label{prop:laml_inc_norm}
\end{prop}
\begin{proof}
By induction on $\norm(M')$. By Lemma \ref{lem:lhrlaml}, there are terms $S, T$ such that:
\[
\xymatrix@R=10pt{
M       \ar[r]^{\laml}
        \ar[dd]^{\lhr}&
M'      \ar[dr]^{\lhr}\\
&&T\\
N       \ar[r]^{\laml^+}&
S       \ar@{<-}[ur]^{\beta^*}
}
\]
By Proposition \ref{prop:betanorm}, we have $\norm(S) \leq \norm(T) = \norm(M') - 1$. Moreover, we have a chain:
\[
N = N_1 \to_{\laml} N_2 \to_{\laml} \dots \to_{\laml} N_n = S
\]
But since $\norm(S) < \norm(M')$, it follows by immediate induction that $\norm(N) \leq \norm(S)$.
Therefore, we have $\norm(M)    = \norm(N) + 1 \leq \norm(S) \leq \norm(M')-1 < \norm(M')$.
\end{proof}

\subsubsection{How $\lambda$-lifting preserves other quantities on terms}
Finally, We examine how $\lambda$-lifting affects the other quantities on terms. We first notice that
it preserves the depth.

\begin{lem}
If $M \to_{\laml} M'$, then $\depth(M) = \depth(M')$.
\label{lem:laml_pres_depth}
\end{lem}
\begin{proof}
First, we prove by induction on $M$ that for any $x$ free in $M$ and variable $y$, we have
$\depth(M) = \depth(M[x~y/x])$. The lemma follows by induction on the definition of $\to_{\laml}$.
\end{proof}

Unfortunately it does affect the order and the local height; however we will show that multiple applications
of $\to_\laml$ can only change them by one. To prove that, we will construct \emph{weighted}
variants of order and local height that give different weight to variables according to their behaviour
with respect to $\to_\laml$. By design they will be preserved by $\to_\laml$, but will remain within 
small bounds of the original level and local height. We start with the order.

Say that a variable $x$ in a term $M$ is \textbf{local} in $M$ iff for any generalized redex $(\lambda y, N)$ of $M$, $x$
does not appear free in $N$. If $x$ is a variable of $M$ define the \textbf{weighted level} $\lv'(x)$ by $\lv(x) + 1$
if $x$ is local and $\lv(x) + 2$ otherwise. Then, define the \textbf{weighted order} $\ord'(M)$ as being the maximum over all
$\lv'(x)$ for variables $x$ in $M$, and all $\lv(A)$ for constants $\daimon_A$ in $M$. 
First, we prove that $\ord'(M)$ remains closely related to $\ord(M)$.

\begin{lem}
Let $M$ be a closed term, then $\ord(M) \leq \ord'(M) \leq \ord(M) + 1$.
\label{lem:ord_word}
\end{lem}
\begin{proof}
We prove first that $\ord(M) \leq \ord'(M)$. If $M$ admits as a subterm a constant $\daimon_A$ with
$\ord(M) = \lv(A)$, then $\ord'(M)\geq \lv(A) \geq \ord(M)$. Otherwise, $M$ has a subterm $\lambda x^A.~M'$ 
such that $\ord(M) = \lv(A) + 1$ -- indeed if a subterm of $M$ of maximal level has the form $M_1~M_2$
then $M_1$ has higher level, and if it has the form $\lambda x^A.~M'$ with $\lv(M') > \lv(A) + 1$ then $M'$
still has maximal level so the property follows by induction. So in particular $M$ has a variable $x$
with $\ord(M) = \lv(x) + 1$, which is always less than $\ord'(M)$.

We now prove that $\ord'(M) \leq \ord(M) + 1$. Firstly, if there is a constant $\daimon_A$ in $M$ such that
$\ord'(M) = \lv(A)$, then $\ord(M) + 1 \geq \ord(M) \geq \lv(A)$ as well. Secondly, if there is a local variable
$x^A$ such that $\ord'(M) = \lv(A) + 1$, then since $M$ is closed it has a subterm of the form $\lambda x^A.~M'$,
so $\ord(M) + 1 \geq \lv(\lambda x^A.~M') + 1 \geq \lv(A) + 2 > \lv(A) + 1$. Finally if there is a non-local
variable $x^A$ such that $\ord'(M) = \lv(A) + 2$, same reasoning.
\end{proof}
                                                   
Now we prove that the weighted order is preserved by $\to_{\laml}$.

\begin{lem}
For any term $M$ with $M \to_{\laml} M'$, then $\ord'(M) = \ord'(M')$.

Therefore (by Lemma \ref{lem:ord_word}) if $M$ is closed with $M \to_{\laml}^* N$, then $\ord(N) \leq \ord(M) + 1$.
\label{lem:laml_pres_ord}
\end{lem}
\begin{proof}
By induction on the definition of $\to_{\laml}$.

\begin{itemize}
\item If $M_l = (\lambda x.~M')~M_1~\dots~M_n \to_{\laml} (\lambda x.~M'[x~y/x])~(\lambda y'.~M_1[y'/y])~\dots~M_n = M_r$,
both terms contain the same constants. If $z$ is a variable on the left hand side, reason by cases. For clarity we
write the levels of variables as $\lv_l(x)$ or $\lv'_l(x)$ on the left hand side,
and $\lv_r(x)$ or $\lv_r'(x)$ on the right hand side. We describe a mapping of variables of $M_l$ to variables of $M_r$
increasing the weighted level.
\begin{itemize}
\item We map $x$ to $x$ if $\lv_l(x) \geq \lv_l(y) + 1$. Then $x$ is local in $M_l$ iff it is local in $M_r$, so there is
$i\in \{1, 2\}$ such that we have $\lv'_l(x) = \lv_l(x) +i$ and $\lv_r'(x) = \lv_r(x) + i = \max(\lv_l(x), \lv_l(y) + 1) + i = \lv_l(x) + i$.
If $\lv_l(x) \leq \lv_l(y)$, we associate $y$ to $x$. Then, we note that $x$ is local in $M_l$ iff $y$ is local
in $M_r$; it follows that 
$\lv_r'(y) = \lv_r(y) + i = \lv_l(y) + i \geq \lv_l(x) + i = \lv_l'(x)$.
\item We map $y$ to $x$ as well. Indeed, $\lv'_l(y) = \lv_l(y) + 2 \leq \lv_r(x) + 1 \leq \lv'_r(x)$.
\item Any other variable is unchanged, its weighted level unaffected by the reduction.
\end{itemize}
Likewise, we give a mapping of variables in $M_r$ to variables in $M_l$.
\begin{itemize}
\item We map $x$ to $x$ if $\lv_l(x) \geq \lv_l(y) + 1$ (then we have $\lv'_r(x) = \lv_r(x) + 1 = \max(\lv_l(x), \lv_l(y)+1) + 1 =
\lv_l(x) + 1 = \lv'_l(x)$), and $y$ if $\lv_l(x) \leq \lv_l(y)$ (then $\lv'_r(x) = \lv_l(y) + 2 = \lv'_l(y)$).
\item We map $y$ and $y'$ to $y$, the weighted level can only increase.
\item To any other $z$ we leave it unchanged, its weighted level is unaffected.
\end{itemize}
\item For all other cases, it directly follows from IH.\qedhere
\end{itemize}
\end{proof}

For preservation of local height, we need a few preliminaries. If $M$ is a term and $x$ is a variable of $M$, 
we say that $x$ is a \textbf{carrier variable} iff there is a generalized redex $(\lambda x, N)$ in $M$ such that
$N$ has a free variable. We also define the \textbf{$V$-weighted local height} of $M$, parametrized by a set 
$V$ of variables and written $\lh_V(M)$, by induction on $M$ as for $\lh$ except that the base case for 
variable occurrences is enriched with:
\[
\lh_V(x_0~M_1~\dots~M_n) = 1 + \max(1, \max_{1\leq i \leq n} \lh_V(M_i))
\]
when $x_0$ is an occurrence of $x\in V$. The \textbf{weighted local height} of $M$, written $\lh'(M)$,
is defined as $\lh_V(M)$ where $V$ is the set of carrier variables of $M$.

First, we show how the weighted local height compares with the local height.

\begin{lem}
For any term $M$, $\lh(M) \leq \lh'(M) \leq \lh(M) + 1$.
\label{lem:wlh_lh}
\end{lem}
\begin{proof}
Immediate by induction on the length of $M$.
\end{proof}


\begin{lem}
If $M$ is a term such that $M \to_\laml M'$, then $\lh'(M) = \lh'(M')$.

Therefore (by Lemma \ref{lem:wlh_lh}) that for any term $M$, if $M \to_{\laml}^* M'$, then $\lh(M') \leq \lh(M) + 1$.
\label{lem:laml_pres_lh}
\end{lem}
\begin{proof}
By induction on the definition of $\laml$, detailing non-trivial cases.
\begin{itemize}
\item If we have, with $y \in \fv(M_1)$:
\[
M_l = (\lambda x.~M)~M_1~\dots~M_n \to_{\laml} (\lambda x.~M[x~y/x])~(\lambda y'.~M_1[y'/y])~\dots~M_n = M_r
\]
Then, it is direct to show by induction on the length of a term $M$ that if
$x, y$ are not bound in $M$ and $x\in V$, $y\not \in V$, and if $V'$ is either $V$ or $V\setminus \{x\}$,
then $\lh_V(M) = \lh_{V'}(M[x~y/x])$. Since $y$ is free in $M_l$ it is not a carrier variable, 
so applying the property above it is direct that $\lh'(M_l) = \lh'(M_r)$.
\item Suppose we have $(\lambda x.~M)~M_1~\dots~M_n \to_\laml (\lambda x.~M')~M_1~M'_2~\dots~M'_n$
from $M_l = M~M_2~\dots~M_n \to_\laml M'~M'_2~\dots~M'_n$ and $M_1$ closed. Let $V$ be the set of carrier
variables of $M_l$. Since $M_1$ is closed, $x \not \in V$. We calculate:
\begin{eqnarray*}
\lh_V((\lambda x.~M)~M_1~\dots~M_n) &=& \max(\lh_V(M~M_2~\dots~M_n), \lh_V(M_1))
\end{eqnarray*}
But $\lh_V(M~M_2~\dots~M_n) = \lh'(M~M_2~\dots~M_n)$: by Barendregt's convention variables
of $V$ appearing in $M~M_2~\dots~M_n$ are bound in $M~M_2~\dots~M_n$ and therefore be
carrier variables in $M~M_2~\dots~M_n$. Therefore, the property follows from IH. \qedhere
\end{itemize}
\end{proof}

\noindent Putting all of the above together, we obtain the following lemma:

\begin{lem}
For any closed term $M$, there is a strongly locally scoped $M'$ such that:
\[
\begin{array}{rclcrcl}
\norm(M') &\geq& \norm(M) &~~~~~~~~~& \ord(M') &\leq& \ord(M) + 1\\
\lh(M') &\leq& \lh(M)+1 &~~~~& \depth(M') &=& \depth(M)
\end{array}
\]
\label{lem:local_scopization}
\end{lem}
\begin{proof}
Starting from $M$, apply
$\to_{\laml}$ as much as possible. By Lemma \ref{lem:lamlterm}, this reduction must terminate on a term $M'$. By Lemma \ref{lem:normlocsc},
$M'$ is strongly locally scoped. By Proposition \ref{prop:laml_inc_norm}, we have $\norm(M') \geq \norm(M)$. By Lemmas \ref{lem:laml_pres_ord}
and \ref{lem:laml_pres_lh}, we have $\lh(M') \leq \lh(M)+1$ and $\ord(M') \leq \ord(M) + 1$. Finally by
Lemma \ref{lem:laml_pres_depth}, we have $\depth(M') = \depth(M)$.
\end{proof}

\subsubsection{Expanding variables}
Until now, terms have been measured through their local height and depth. However for general terms those are rather unnatural
quantities, and bounds for $\lambda$-calculi are usually not expressed in terms of them. 

The \textbf{height} of a term $M$, is the quantity defined by induction by $\h(\daimon) = 0, \h(x_0) = 1, \h(\lambda x.~M) = \h(M), 
\h(M_1~M_2) = \max(\h(M_1), \h(M_2) + 1)$. We give a final norm-increasing
term transformation, allowing us to convert depth and local height to height.

\begin{lem}
For any term $M$, there exists a term $M'$ such that $M\to_{\eta}^* M'$,
\[
\begin{array}{rclcrcl}
\lh(M') &\leq& 2 &~~~~~~~~~~~& \depth(M') &\leq& \h(M)\\
\ord(M') &=& \ord(M) && \norm(M') &\geq& \norm(M)
\end{array}
\]
\label{lem:expansion}
\end{lem}
\begin{proof}
We define the \emph{expansion} $\exp(M)$ by induction on $M$ as follows:
\begin{eqnarray*}
\exp(x_0) &=& \lambda {y^1}^{A_1}.~\dots \lambda {y^n}^{A_n}.~x_0~y_1~\dots~y_n\\
\exp(\daimon) &=& \daimon\\
\exp(\lambda x.~M) &=& \lambda x.~\exp(M)\\
\exp(M~N) &=& \exp(M)~\exp(N)
\end{eqnarray*}
where in the first case, $x$ has type $A_1 \to \dots \to A_n \to o$.
It is direct that $M \to_{\eta}^* \exp(M)$, so by Proposition \ref{prop:etanorm} and Lemma \ref{lem:etapresord}
we have $\norm(\exp(M)) \geq \norm(M)$ and $\ord(\exp(M)) = \ord(M)$. Finally, $\lh(\exp(M)) \leq 2$
and $\depth(\exp(M)) \leq \h(M)$ are immediate by induction.
\end{proof}

\subsubsection{Exact bounds for general terms}

We are now interested in estimating the quantity:
\[
\gen_n(h) =  \max \{\norm(M) \mid \ord(M) \leq n~\&~\h(M) \leq h\}
\]
We do that by applying the tools developed earlier to get an upper bound on the length of reduction, and then prove a matching lower
bound by providing terms whose length of reduction asymptotically reaches the upper bound.

\begin{prop}[Upper bound]
Suppose $M$ is a term. Then, $\norm(M) \leq 2_{\ord(M)}^{\h(M)\log(\ord(M) + 5)}$.
\label{prop:genupbound}
\end{prop}
\begin{proof}
Start with a term $M$, call it $M_0$. Without loss of generality we can consider $M$ closed, otherwise we replace occurrences of
free variables by occurrences of constants without changing the norm, and only reducing the local height, depth and order.
Then, we apply the following transformations.

%
First we expand variables: by Lemma \ref{lem:expansion}, there is a term closed term $M_1$ with:
\[
\begin{array}{rclcrcl}
\norm(M_1) &\geq& \norm(M_0) &~~~~& \ord(M_1) &=& \ord(M_0)\\
\lh(M_1) &\leq& 2 &~~~~& \depth(M_1) &\leq& \h(M_0)
\end{array}
\]
Then, we make it strongly locally scoped. By Lemma \ref{lem:local_scopization}, there is $M_2$ sls such that:
\[
\begin{array}{rclcrcl}
\norm(M_2) &\geq& \norm(M_1) &~~~~& \ord(M_2) &\leq& \ord(M_1) + 1\\
\lh(M_2) &\leq& \lh(M_1)+1 &~~~~& \depth(M_2) &=& \depth(M_1)
\end{array}
\]
If $\ord(M_2) = 0$, then $M_2$ does not have any prime redex (nor any application, in fact), so $\norm(M_2) = 0$. Otherwise,
by Proposition \ref{prop:boundetals}, $\norm(M_2) \leq 2_{\ord(M_2)-1}^{\depth(M_2)\log(\lh(M_2) + \ord(M_2) + 1)}$.
By the inequalities above, we immediately deduce that $\norm(M) \leq 2_{\ord(M)}^{\h(M)\log(\ord(M) + 5)}$.
\end{proof}

\begin{thm}
For fixed $n\geq 3$ we have $\gen_n(h) = 2_n^{\Theta(h)}$.
\end{thm}
\begin{proof}
\emph{Upper bound.} By Proposition \ref{prop:genupbound}. 

\emph{Lower bound.} The construction is essentially the same as the one used in \cite{beck} for the lower bound in terms of height.
For $p\geq 1$ and $k\geq 0$, we define $b_0^p = \church{2}_{p}$ and $b_{k+1}^p = \lambda x^{A_{p-1}}.~b_k^p~(b_k^p~x)$.
Then, we set:
\[
B_k^p = b_k^p~\church{2}_{p-1}~\dots~\church{2}_0
\]
Note that this term is \emph{not} sls.
By standard arithmetic of Church numerals, we have that for any $p\geq 1, k\geq 0$,
$B_k^p \to_{\beta}^* \church{2_{p+2}^k}_0$.
By Lemma \ref{lem:main:churchlhr} it follows that $\norm(B_k^p~\id_o) \geq 2_{p+2}^k$. It is direct to check that 
$\ord(B_k^p) = p+2$ and $\h(B_k^p) = k+3$ (for $k\geq 1$), concluding the proof.
\end{proof}

For a term $M$ of height $h$ and order $n$, Beckmann's results \cite{beck} predict that any $\beta$-reduction chain of $M$ terminates
in less than $2_{n+1}^{\Theta(h)}$ steps. It might seem counter-intuitive that our bound (with lhr) is smaller
than Beckmann's (with $\beta$-reduction) since we substitute only one occurrence at a time, which is obviously longer. However,
Beckmann considers arbitrary $\beta$-reduction, not head $\beta$-reduction. The possibility of reducing in arbitrary
locations of the term unlocks much longer reductions, since higher-order free variables or constants can isolate sections of the term that
will never arrive in head position but can still be affected by arbitrary $\beta$-reduction. The fact that the length of lhr has
the same order of magnitude as head $\beta$-reduction is not surprising in the light of Accattoli and Dal Lago's recent result \cite{adl} that
a similar notion of lhr is quadratically related to head reduction.

\bigskip

\noindent\textbf{Acknowledgment.} The author is grateful to the anonymous referees for their very helpful comments and suggestions.

\bibliographystyle{plain}
\bibliography{main}

%

\end{document}